\documentclass[11pt]{article}
\usepackage[utf8]{inputenc}
\usepackage{amsmath, bm}
\usepackage{bbm}
\usepackage{amssymb}
\usepackage{amsthm}
\usepackage{algorithm}
\usepackage{algpseudocode}
\usepackage{comment}
\usepackage{fullpage}
\usepackage{xcolor}
\usepackage[hidelinks]{hyperref}
\usepackage{cleveref}
\usepackage{IEEEtrantools}
\usepackage{graphicx}
\usepackage{subcaption}
\usepackage[affil-it]{authblk}
\usepackage{cancel}

\usepackage{natbib}
\DeclareMathOperator*{\argmin}{arg\,min}

\DeclareMathOperator*{\as}{\overset{\text{\normalfont{a.s.}}}{\longrightarrow}}

\def\FD{\operatorname{FD}}
\def\DFD{\operatorname{DFD}}

\def\S{\mathcal{S}}
\def\P{\mathbb{P}}

\def\E{\mathbb{E}}

\def\H{\mathcal{H}}

\def\R{\mathbb{R}}
\def\X{\mathcal{X}}

\def\N{\mathbb{N}}
\def\D{\nabla}

\newtheorem{assumption}{Assumption}
\newtheorem{sassumption}{Standing Assumption}
\newtheorem{theorem}{Theorem}

\newtheorem{lemma}{Lemma}

\newtheorem{proposition}{Proposition}

\newtheorem{definition}{Definition}

\newtheorem{remark}{Remark}
\newtheorem{example}{Example}

\newenvironment{talign*}
 {\csname align*\endcsname}
 {\endalign}

\title{Generalised Bayesian Inference for Discrete Intractable Likelihood}

\author[1]{Takuo Matsubara}
\author[2]{Jeremias Knoblauch}
\author[2,4]{Fran\c{c}ois-Xavier Briol}
\author[3,4]{Chris. J. Oates}

\affil[1]{The University of Edinburgh, UK}
\affil[2]{University College London, UK}
\affil[3]{Newcastle University, UK}
\affil[4]{The Alan Turing Institute, UK}

\date{}

\begin{document}

\maketitle

\begin{abstract}
    Discrete state spaces represent a major computational challenge to statistical inference, since the computation of normalisation constants requires summation over large or possibly infinite sets, which can be impractical.
    This paper addresses this computational challenge through the development of a novel generalised Bayesian inference procedure suitable for discrete intractable likelihood.
    Inspired by recent methodological advances for continuous data, the main idea is to update beliefs about model parameters using a discrete Fisher divergence, in lieu of the problematic intractable likelihood.
    The result is a generalised posterior that can be sampled from using standard computational tools, such as Markov chain Monte Carlo, circumventing the intractable normalising constant.
    The statistical properties of the generalised posterior are analysed, with sufficient conditions for posterior consistency and asymptotic normality established.
    In addition, a novel and general approach to calibration of generalised posteriors is proposed.
    Applications are presented on lattice models for discrete spatial data and on multivariate models for count data, where in each case the methodology facilitates generalised Bayesian inference at low computational cost.
\end{abstract}


\section{Introduction}

This paper focuses on statistical models for data defined on a discrete set $\X$, whose probability mass function $p_{\theta}$ involves a parameter $\theta$ to be inferred.
In this setting, there is an urgent need for computational methodology applicable to models that are \emph{intractable}, in the specific sense that 
\begin{align}
p_{\theta}(\bm{x}) = \frac{\tilde{p}_{\theta}(\bm{x})}{ Z_{\theta}}, \qquad 
Z_{\theta} := \sum_{\bm{x} \in \X} \tilde{p}_{\theta}(\bm{x}),  \label{eq: def intractable}
\end{align}
where the positive function $\tilde{p}_\theta$ is straightforward to evaluate but direct computation of the normalising constant $Z_\theta \in (0,\infty)$ is impractical.
This situation is ubiquitous in the discrete data context, since it is often impractical to compute a sum over a large or infinite discrete set. To limit scope, this paper considers generalised Bayesian inference where, to date, several computational approaches have been proposed.
These approaches, which are recalled in \Cref{sec:background}, are mainly applicable in settings where it is possible to simulate data $\bm{x}$, conditional on the parameter $\theta$.
However, in several of the most scientifically important instances of \eqref{eq: def intractable}, exact (or even approximate) simulation from the model is not practical.

Important examples of statistical models exhibiting these computational challenges include lattice models of spatial data 
\citep{moores2020scalable}, statistical models for graph-valued data \citep{lusher2013exponential}, and statistical models for multivariate count data \citep{inouye2017review}.
In each case, the normalising constant involves summation over a set whose cardinality is exponential in the dimension of the lattice,  in the size of the nodal set of the graph, or even infinite, rendering direct computation and simulation of data intractable in general.

To circumvent both computation of the normalising constant and simulation from the statistical model, \citet{Matsubara2021} proposed a generalised Bayesian posterior, called \emph{KSD-Bayes}, which is based on a Stein discrepancy.
The resulting generalised posterior is consistent and asymptotically normal, and thus shares many of the properties of the standard Bayesian posterior whilst admitting a form which does not require the computation of an intractable normalisation constant.  
However, a major limitation of KSD-Bayes is the dependence of the generalised posterior on a user-specified symmetric positive definite function, called a kernel, which determines precisely how beliefs are updated.
In continuous domains, such as $\R^d$, there are several natural choices of kernel available, and their associated Stein discrepancies have been well-studied \citep{Anastasiou2021}.
However, in discrete domains there are often no natural choices of kernel, or when natural choices exists \citep[such as a heat kernel;][]{chung1997spectral} they can be computationally impractical.

This paper presents \emph{DFD-Bayes}, the first generalised Bayesian inference method tailored to inference with discrete intractable likelihood. The approach is based on a novel discrete version of the Fisher divergence which, in contrast to KSD-Bayes, does not require a kernel to be specified.
Further, the DFD-Bayes posterior has computational complexity $O(nd)$, where $n$ is the number of data and $d$ is the data dimension, which compares favourably to the KSD-Bayes computational complexity of $O(n^2 d)$.
The DFD-Bayes methodology is supported by asymptotic guarantees, presented in \Cref{sec: methods}, and empirical results, in \Cref{sec: experiment}, demonstrate state-of-the-art performance in the applications considered. 
Before setting out the proposed methodology, we first we review related work in \Cref{sec:background}.

\section{Background} \label{sec:background}

The aim of this section is to briefly review existing Bayesian and generalised Bayesian methodology for intractable statistical models, extending the discussion to include both continuous and discrete data.
Frequentist estimation for intractable models is not discussed \citep[we refer the reader to e.g.][]{Hyvarinen2005}.

\paragraph{Approximate Likelihood}

Faced with an intractable model, a pragmatic approach is simply to employ standard Bayesian inference with a tractable approximation to the likelihood \citep[e.g.][]{Bhattacharyya2019}.
A classical example of approximate likelihood is the pseudolikelihood of \cite{Besag1974}, which replaces the joint probability mass function of the data with a product of conditional probability mass functions, each of which is sufficiently low-dimensional (or otherwise tractable enough) to permit normalising constants to be computed.
Generalisations of this approach are sometimes referred to as composite likelihood \citep{varin2011overview}.
These approximations are usually model-specific, and analysis of the approximation error may be difficult in general \citep{lindsay2011issues}.

\paragraph{Simulation-Based Methods}
One class of intractable statistical models that has been explored in detail are models for which it is possible to simulate data $\bm{x}$ conditional on the parameter $\theta$.
A well-known approach to inference in this class of models is the exchange algorithm of \citet{Moller2006} and \citet{Murray2006}, which constructs a Markov chain on an extended state space for which the standard Bayesian posterior occurs as a marginal.
Simulation of the Markov chain requires both exact simulation from the statistical model and evaluation of $\tilde{p}_\theta(\bm{x})$.
Further methodological development has been focused on removing the requirement to evaluate $\tilde{p}_\theta(\bm{x})$, with approximate Bayesian computation \citep{Marin2012}, Bayesian synthetic likelihood \citep{Price2018}, MMD-Bayes \citep{Cherief-Abdellatif2019,pacchiardi2021generalized} and the posterior boostrap \citep{Dellaporta2022} emerging as likelihood-free methods, which require only that data can be simulated. 
Unfortunately, for many statistical models of discrete data, exact simulation \citep[the state-of-the-art being e.g.][]{propp1998coupling} from the model is impractical.

\paragraph{Markov Chain-Based Methods}
Another pragmatic approach is to substitute exact simulations with approximate simulations, such as obtained from a Markov chain.
This idea works in specific instances; see the review of \citet{park2018bayesian}.
The main drawback of these approaches, as far as this paper is concerned, is that they require the design of a rapidly mixing Markov chain on a possibly large (or infinite) discrete set.
As such, these methods require bespoke implementations for each class of statistical model considered, and for many models of interest appropriate Markov chains have yet to be developed.
Thus Markov chain-based methods do not represent a general solution to discrete intractable likelihood.

\paragraph{Russian Roulette}
The pseudo-marginal approach justifies replacing the intractable likelihood $p_\theta(\bm{x})$ with a positive unbiased estimator $\hat{p}_\theta(\bm{x})$ of the likelihood in the context of a Metropolis--Hastings algorithm \citep{andrieu2009pseudo}.
The practical difficulty of this approach is to construct a positive unbiased estimator.
\citet{Lyne2015} proposed the Russian roulette estimator for intractable statistical models, a simulation technique from the physics literature 
which involves random truncation of the sum (or of an integral in the continuous context) defining the normalising constant.
The Russian roulette estimator is unbiased but is not guaranteed to be positive, meaning that post hoc re-weighting of the Markov chain sample path is required.
The ergodicity of Russian roulette has not, to the best of our knowledge, been theoretically studied.
Further, the mixing time of the Markov chain is known to be sensitive to the variance of $\hat{p}_\theta(\bm{x})$, which can be large for estimators based on random truncation (especially when there is no clear a priori ordering for the summands, which can occur in the discrete context).
As such, the pseudo-marginal approach does not at present represent a general computational solution to intractable likelihood.

\paragraph{Generalised Bayesian Inference} 
Motivated by the absence of general computational methodology for intractable likelihood, \cite{Matsubara2021} proposed a solution called KSD-Bayes.
The setting for this approach is the nascent field of generalised Bayesian inference.
Given a prior $\pi(\theta)$, a dataset $\{\bm{x}_i\}_{i=1}^n \subset \X$, and a constant $\beta > 0$, generalised Bayesian inference updates beliefs using a loss function $D_n(\theta)$, producing a generalised posterior
\begin{align}
	\pi_n^D(\theta) \propto \pi(\theta) \exp( - \beta D_n(\theta) ) . \label{eq: beta in gen bayes}
\end{align}
The standard posterior is recovered by the negative log-likelihood $D_n(\theta) = - \sum_{i=1}^{n} \log p_\theta(\bm{x}_i)$, while several alternative loss functions have been developed to confer robustness in settings where the statistical model is misspecified (see the survey in \citet{Bissiri2016} for the case of additive loss functions, and \citet{Knoblauch2019} for further generalisation). 
KSD-Bayes \citep{Matsubara2021}  is distinguished among existing generalised Bayesian inference methods by its applicability to statistical models involving an intractable normalising constant \citep[see also Section~4.2 of][]{Giummole2019}.
This was achieved by selecting $D_n(\theta)$ to be a Stein discrepancy between the statistical model $p_\theta$ and the empirical distribution of the dataset, which can be computed without the normalising constant.
Strikingly, a fully conjugate treatment of the continuous exponential family model, and a straight-forward treatment of the discrete exponential family model using Markov chain Monte Carlo, is possible using KSD-Bayes; this in principle provides a solution to many of the aforementioned instances of intractable likelihood.
However, the dependence of KSD-Bayes on a user-specified kernel renders the approach unattractive for discrete domains, where there are often no natural choices of kernel, or where natural choices\footnote{A natural choice is the heat kernel, whose origins lie in spectral graph theory \citep{chung1997spectral}.  However, computation of the heat kernel requires a $O(D^3)$ cost where $D = \text{card}(\X)$, which is often impractical.  
For example, the Ising model on a lattice $\X = \{0,1\}^d$ has $D = 2^d$, while the Conway--Maxwell--Poisson model of \Cref{subsec: CMP} has $D = \infty$, meaning approximation of the heat kernel would be required.} are computationally impractical.
Furthermore, the $O(n^2d)$ computational cost of KSD-Bayes is super-linear in the size of the dataset.

\vspace{5pt}

This paper presents general methodology for inferring the parameters of a intractable discrete statistical model.
The main idea is to employ a discrete Fisher divergence as a loss function in a generalised Bayesian inference context.
The resulting \emph{DFD-Bayes} method does not require a choice of kernel, enjoys theoretical guarantees, and can be computed at cost $O(nd)$ linear in the size of the dataset.
Full details are provided next.

\section{Methodology}
\label{sec: methods}

This section presents and analyses DFD-Bayes.
First, we present a novel discrete formulation of the Fisher divergence in \Cref{sec:discfd}.
DFD-Bayes is introduced in \Cref{sec:fdpos}, where posterior consistency and asymptotic normality are established.
\Cref{sec:calibration} presents a novel approach to calibration of generalised posteriors, which may be of independent interest.
Limitations of DFD-Bayes are discussed in \Cref{sec:limitation}.

\paragraph{Notation}
Denote by $\X$ a countable set in which data are contained, and by $\Theta$ the set of permitted values for the parameter $\theta$, where $\Theta$ is a Borel subset of $\R^v$ for some $v \in \N$.
Probability distributions on $\X$ are identified with their probability mass functions, with respect to the counting measure on $\X$.
The $i$-th coordinate of a function $f: \X \to \R^d$ is denoted by $f_i : \X \to \R$.
For a probability distribution $q$ on $\X$ and $d,p \in \N$, denote by $L^p(q , \R^d)$ the Lebesgue space of measurable functions $f: \X \to \R^d$ such that $ \sum_{i=1}^{d} \E_{X \sim q}[ | f_i(X) |^p ]  < \infty$, in which two elements $f, g \in L^p(q, \R^d)$ are identified if they are $p$-almost everywhere equal. 
The notation $\| \cdot \|$ indicates the Euclidean norm of $\R^m$, and will be applied also to matrices and tensors interpreted, respectively, as elements of $\R^{v \times v}$ and $\R^{v \times v \times v}$.
A Dirac measure at $x \in \X$ is denoted by $\delta_x$.

\subsection{A Discrete Fisher Divergence} \label{sec:discfd}

The Fisher divergence underpins several frequentist estimators for intractable statistical models, most notably score matching \citep{Hyvarinen2005}, and has been used in the context of Bayesian model selection \citep[e.g.][]{Dawid2015}.
It is classically defined for continuous domains; for
(sufficiently regular) densities $p$ and $q$ on $\R^d$, the Fisher divergence is $\operatorname{FD}(p \| q) = \E_{X \sim q}[ \| \nabla \log p(X) - \nabla \log q(X) \|^2 ]$
where $\nabla$ denotes the gradient operator in $\R^d$.
Its main advantage is that it can be computed without knowledge of the normalising constant\footnote{The Fisher divergence depends only on $\nabla \log p$, equal to the ratio $(\nabla p) / p$, meaning it is sufficient to know $p$ up to a normalising constant.} of $p$ and, furthermore, expectations with respect to $p$ are not required.
The Fisher divergence was extended to discrete domains in \citet{Lyu2009,xu2022generalized}.
However, existing work focuses on domains $\X$ of finite cardinality or one-dimensional models, and a technical contribution of this paper, which may be of independent interest, is to present an extension of Fisher divergence to certain discrete domains which may be a countably infinite set in multiple dimensions.
The extended divergence satisfies the requirements of a proper local scoring rule and thus complements existing scoring rule methodology developed in the finite domain context in \citet{dawid2012proper}. 

\begin{sassumption}
    Let $\X = S_1 \times \dots \times S_d$, where for each $i = 1,\dots,d$ there is an order isomorphism $S_i \cong I_i \subseteq \mathbb{Z}$, and $d \in \mathbb{N}$.
\end{sassumption}

\noindent 
This setting is general enough to include diverse data types, such as multivariate count data, or network data with a fixed vertex set.
For any set $S \cong I \subseteq \mathbb{Z}$, precisely one of the following must hold: (i) no smallest or largest elements of $S$ exist; (ii) both a smallest element, $s_{\min}$, and a largest element, $s_{\max}$, exist; (iii) only $s_{\min}$ exists; (iv) only $s_{\max}$ exists.
Without loss of generality, we will identify the case (iv) with (iii) by reversing the ordering of $S$.
In addition, it will be useful to extend the domains $S_i$ to include an additional state (not part of the ordering), denoted $\star$, and to this end we let $S_i^\star = S_i \cup \{\star\}$ and $\X^\star = S_1^\star \times \dots \times S_d^\star$.
A function $h: \X \to \R$ extends to a function $h : \X^\star \to \R$ by setting $h(\bm{x}) = 0$ whenever any of the coordinates of $\bm{x}$ are equal to $\star$.	

\begin{definition} \label{asmp:oss}
	Let $S \cong I \subseteq \mathbb{Z}$.
	For consecutive elements $r < s < t$ in $S$ we let $s^- := r$ and $s^+ := t$.
	If both $s_{\min}$ and $s_{\max}$ exist, we let $s_{\min}^- := s_{\max}$ and $s_{\max}^+ := s_{\min}$ or, if only $s_{\min}$ exists, we let $s_{\min}^{-} := \star$ and $\star^+ = s_{\min}$.
	For $\bm{x} = (x_1, \dots, x_d) \in \X$, define $\bm{x}^{i +} := (x_1, \dots, x_i^{+}, \dots, x_d)$ and $\bm{x}^{i -} := (x_1, \dots, x_i^{-}, \dots, x_d)$.
\end{definition}

\noindent Simply put, this ensures that each element $s$ has both a preceding and proceeding element, so that increments and decrements are well-defined.
The above structure can be exploited to define an operator for $\X$ that is analogous to the gradient operators for $\R^d$:

\begin{definition} \label{def: diffop}
	For $h: \X \to \R$, define the backward difference operator by
	\begin{align*}
		\D^- h(\bm{x}) := \big[ h(\bm{x}) - h(\bm{x}^{1-}), ~\cdots~, ~h(\bm{x}) - h(\bm{x}^{d-}) \big]^\top \in \R^d .
	\end{align*}
\end{definition}

\noindent 
Based on \Cref{asmp:oss,def: diffop}, we can construct a divergence applicable to discrete domains $\X$, which we term a \emph{discrete Fisher divergence}.
Recall that values of $f \in L^p(q, \R^d)$ in a measure zero domain of $q$ i.e.~$\{ \bm{x} \in \X \mid q(\bm{x}) = 0 \}$ are arbitrary and not involved in the integral with respect to $q$ \citep[Remark 1.37, p.29]{Rudin1987}. In what follows, it is sufficient for functions $( \nabla^- p ) / p, (\nabla^- q) / q \in L^2( q , \R^d)$ to be well-defined in the support of $q$.

\begin{definition} \label{def: dif_sd}
	Let $p$ and $q$ be probability distributions on $\X$, such that $( \nabla^- p ) / p, (\nabla^- q) / q \in L^2( q , \R^d)$.
	The \emph{discrete Fisher divergence} is defined as 
	\begin{align}
		\DFD(p \| q) := \E_{X \sim q}\left[ \left\| \frac{\nabla^-p(X)}{p(X)} - \frac{\nabla^-q(X)}{q(X)} \right\|^2 \right] .  \label{eq: DFD def}
	\end{align}
\end{definition}

\noindent
The choice of a Euclidean norm in \eqref{eq: DFD def} is not critical and other norms could be employed, but for expository purposes the standard Euclidean norm will be used throughout.
\Cref{prop: dif_sd_1} justifies the name `divergence' and offers an alternative, computable formula for \eqref{eq: DFD def}.

\begin{proposition} \label{prop: dif_sd_1}
	The discrete Fisher divergence satisfies $\DFD(p \| q) \geq 0$ for any $p, q$, with equality if and only if $p = q$.
	Furthermore, if $p(\bm{x}^{j+}) > 0$ for all $\bm{x}$ and $j = 1, \dots, d$ in the support of $q$, it admits the following alternative formula
        \begin{align}
		\DFD(p \| q) & = \E_{X \sim q}\left[ \sum_{j=1}^{d} \left( \frac{p(X^{j-})}{p(X)} \right)^2 - 2 \left( \frac{p(X)}{p(X^{j+})} \right) \right] + C(q) \label{eq:ssd} , 
	\end{align}
	where the term $C(q) := \E_{X \sim q}[ \sum_{j=1}^{d} 1 + ( 1 - q(X^{j-}) / q(X) )^2 ]$ is $p$-independent.
\end{proposition}

\noindent The proof is provided in \Cref{apx: proof_dif_sd_1}.
Note that $\DFD(p \| q)$ can be computed without the normalising constant of $p$, analogously to $\FD(p \| q)$ in $\R^d$.
All models $p_\theta$ used in this paper are positive on $\X$, for which the assumption $p(\bm{x}^+) > 0$ in Proposition 1 is automatically satisfied.
From \Cref{prop: dif_sd_1}, the discrete Fisher divergence between a model $p_\theta$ and an empirical distribution $p_n = \frac{1}{n} \sum_{i=1}^n \delta_{\bm{x}_i}$ corresponding to data $\{ \bm{x}_i \}_{i=1}^{n}$, is computed as 
\begin{align}
	\DFD(p_\theta \| p_n) \; 
 	\stackrel{\theta}{=} \; & \frac{1}{n} \sum_{i=1}^{n} \sum_{j=1}^{d} \left( \frac{p_\theta(\bm{x}_i^{j-})}{ p_\theta(\bm{x}_i)} \right)^2 - 2 \left(\frac{p_\theta(\bm{x}_i)}{ p_\theta(\bm{x}_i^{j+})}\right)  \label{eq:dfd_model}
\end{align}
where $\stackrel{\theta}{=}$ indicates equality up to an additive, $\theta$-independent constant.
In contrast to the continuous Fisher divergence, the $\theta$-independent constant $C(p_n) = \frac{1}{n} \sum_{i=1}^{n} \sum_{j=1}^{d} 1 + ( 1 - p_n(\bm{x}_i^{j-}) / p_n(\bm{x}_i) )^2$ is well-defined for an empirical density $p_n$ in the discrete Fisher divergence.

\begin{remark}
\label{rem: cost}
    The computational cost associated with evaluation of \eqref{eq:dfd_model} is $O(nd)$, which improves on the $O(n^2 d)$ cost of kernel Stein discrepancy.
	Furthermore, if $\X$ is a finite set and count data are provided, indicating the number of times each of the elements of $\X$ occurred, then the complexity of \eqref{eq:dfd_model} reduces to $O(d)$, independent of the size of the dataset.
\end{remark}

\begin{remark}
    The discrete Fisher divergence can also be interpreted as a Stein discrepancy constructed based on an $L^2$-ball Stein set \citep{Barp2019}.
    This implies that discrete Fisher divergence is stronger than popular kernel Stein discrepancies; see \Cref{app: ksd connection}.
\end{remark}

\subsection{A Generalised Posterior} \label{sec:fdpos}

We are now in a position to present DFD-Bayes.

\begin{definition}[DFD-Bayes] \label{def: dksd-bayes}
    Given a prior distribution $\pi$ on $\Theta$, a statistical model $p_\theta: \X \to (0,\infty)$ parametrised by $\theta \in \Theta$, and a dataset $\{\bm{x}_i\}_{i=1}^n$, the \emph{DFD-Bayes posterior} is 
	\begin{align}
		\pi_n^D(\theta) & \propto \pi(\theta) \exp \left(- \beta n \DFD(p_{\theta} \| p_n) \right)  ,
		\label{eq: DFD posterior}
	\end{align}
	where $\beta \in (0,\infty)$ is a constant to be specified.
\end{definition}

\noindent This is clearly a special case of the generalised posterior in \eqref{eq: beta in gen bayes} with $D_n(\theta) = n \DFD(p_\theta\|p_n)$. 
The $\theta$-independent constant $C(p_n)$ of $\DFD(p_\theta \| p_n)$ will be cancelled out by normalisation of the DFD-Bayes posterior. It is thus sufficient to use \eqref{eq:dfd_model} in place of $\DFD(p_\theta \| p_n)$ for computation.
The role of $n$ in \eqref{eq: DFD posterior} is to ensure correct scaling of the generalised posterior as $n \rightarrow \infty$ limit, while the appropriate choice of $\beta$ is crucial in calibrating the coverage of the generalised posterior at finite $n$, and will be discussed in \Cref{sec:calibration}.
\Cref{app: walk-through_example} contains a detailed worked example of the DFD-Bayes posterior and a comparison with other posteriors using simple tractable models.
For the moment, two important properties are highlighted:

\begin{remark}
In contrast to standard posteriors for intractable likelihoods, the DFD-Bayes posterior is directly amenable to standard Markov chain Monte Carlo because \eqref{eq:dfd_model} is independent of the intractable constant, with the cost of evaluating \eqref{eq: DFD posterior} as low as $O(d)$ (c.f. \Cref{rem: cost}).
\end{remark}

\begin{remark}
In contrast to KSD-Bayes, DFD-Bayes is invariant to order-preserving transformations of the data. 
Note that the discrete Fisher divergence upper bounds the kernel Stein discrepancies; see \Cref{subsec: KSD vs DFD}.
\end{remark}
The asymptotic behaviour of the standard Bayesian posterior is well-understood, with sufficient conditions for posterior consistency and asymptotic normality providing frequentist justification for Bayesian inference in the large data limit.
Our attention now turns to establishing analogous conditions for DFD-Bayes.
\begin{sassumption} \label{asmp:standing}
	The data $\{\bm{x}_i\}_{i=1}^n$ consist of independent samples from a probability distribution $p$ on $\X$.
	The distribution $p$ and the statistical model $p_\theta$ for these data satisfy $(\nabla^- p) / p, (\nabla^- p_\theta) / p_\theta \in L^2( p, \R^d)$, for all $\theta \in \Theta$.
\end{sassumption}

\noindent The setting of independent data is broad enough to contain important examples of discrete intractable likelihood, including the models studied in \Cref{sec: experiment}.
The other assumption simply ensures that $\DFD(p_{\theta} \| p_n)$ is well-defined, due to \Cref{prop: dif_sd_1}.
In this setting, a natural first requirement is that the statistical model is identifiable in the large data limit:

\begin{assumption} \label{asmp:minimiser}
    There exists a unique minimiser $\theta_*$ of $\theta \mapsto \DFD(p_{\theta} \| p)$ and there exists a sequence $\{ \theta_n \}_{n=1}^{\infty}$ such that $\theta_n$ minimises $\theta \mapsto \DFD(p_\theta \| p_n)$ almost surely for all $n$ sufficiently large.
    Further, there exists a bounded convex open set $U \subseteq \Theta$ such that $\theta_* \in U$ and $\theta_n \in U$ almost surely for all $n$ sufficiently large.
\end{assumption}

\noindent 
The existence of $U$ in \Cref{asmp:minimiser} essentially implies that for large enough $n$, we can restrict our theoretical analysis to a bounded subset $U \subseteq \Theta$.
This is not restrictive: it can be enforced by re-parameterising the model $p_{\theta}$ so that its new parameter space is bounded and convex.\footnote{For example, we can re-parameterise any unbounded parameter $\kappa$ through the logistic function and define the invertible transformation $\theta = (1 + e^{-\kappa})^{-1} \in [0,1]$.}
The existence of $\{ \theta_n \}_{n=1}^{\infty}$ and $\theta_*$ is more difficult to assess in practice, since the true data generating distribution $p$ is unknown.
That being said, assuming their existence is common in the asymptotic analysis of Bayesian procedures \citep[see e.g.][Section~10]{Vaart1998}.
It is worth highlighting that \Cref{asmp:minimiser} does not require the model family $\{ p_\theta \mid \theta \in \Theta \}$ to contain  $p$, which is in contrast to the assumptions needed for the classical asymptotic normality result \citep[Theorem~10.1]{Vaart1998}.
On the other hand, if the model family $\{ p_\theta \mid \theta \in \Theta \}$ contains $p$ uniquely, existence of $\theta_*$ is immediate since the discrete Fisher divergence is a divergence and hence $\DFD(p_{\theta} \| p) = 0$ if and only if $p_{\theta} = p$. 

Our second main requirement is a technical condition on the derivatives and moments of the model, to ensure that the asymptotic limit is well-defined.
It is helpful to introduce the shorthand $r_{j-}(\bm{x}, \theta) := p_{\theta}(\bm{x}^{j-}) / p_{\theta}(\bm{x})$.
For a function $g: \Theta \to \R$, let $\nabla_\theta^2 g(\theta) \in \R^{v \times v}$ with entries $ \partial_i \partial_j  g(\theta)$, and let $\nabla_\theta^3 g(\theta) \in \R^{v \times v \times v}$ with entries $ \partial_i  \partial_j  \partial_k  g(\theta)$.
\begin{assumption} \label{asmp:derivative}
	Assume that $\theta \mapsto p_{\theta}(\bm{x})$ is three times continuously differentiable in $U$ for any $\bm{x} \in \X$, and
	\begin{align*}
	    \E_{X \sim p} \left[ \sup_{\theta \in U} \| \nabla_{\theta}^{s} r_{j-}(X^{j+}, \theta) \| \right] < \infty \qquad \text{ and } \qquad \E_{X \sim p} \left[ \sup_{\theta \in U} \| \nabla_{\theta}^{s} ( r_{j-}(X, \theta)^2 ) \| \right] < \infty
	\end{align*}
	for all $j = 1, \dots, d$ and $s = 1, 2, 3$.
\end{assumption}

\noindent In contrast to \Cref{asmp:minimiser}, it is easier to verify \Cref{asmp:derivative}, as illustrated in \Cref{ex:exp_asmp}. 
It considers the exponential family, a large class of models which encompasses the models in our experiments in \Cref{sec: experiment}. 
For example, any model on a space $\X$ of finite cardinality is an exponential family model \cite[Ch.~2.2.2]{Amari2016}.
\begin{example}[Exponential Family] \label{ex:exp_asmp}
	Consider an exponential family model $p_\theta(\bm{x}) \propto \exp( \eta(\theta) \cdot T(\bm{x}) + b(\bm{x}) )$, where $\eta: \Theta \to \R^k$, $T: \X \to \R^k$ and $b: \X \to \R$ for some $k \in \N$.
	For this model, we have $r_{j-}(\bm{x}, \theta) = \exp( \eta(\theta) \cdot ( T(\bm{x}^{j-}) - T(\bm{x}) ) + b(\bm{x}^{j-}) - b(\bm{x}) )$.
	\Cref{asmp:derivative} is satisfied if, for $j = 1, \dots, d$, (i) $\| \eta(\theta) \|$ and $\| \nabla_{\theta}^{s} \eta(\theta) \|$ for $s = 1, 2, 3$ are bounded over $\theta \in U$, (ii) $\| T(\bm{x}^{j-}) - T(\bm{x}) \|$ is bounded over $\bm{x} \in \X$, and (iii) $\E_{X \sim p}[ \exp( b(X^{j-}) - b(X) )^2 ] < \infty$.
	The requirements (ii) and (iii) are immediate if $\X$ is a finite set.
\end{example}

\noindent The calculations that accompany \Cref{ex:exp_asmp} are provided in \Cref{apx:example}.
The following theorem establishes that both consistency and asymptotic normality hold.
The former implies that our generalised posterior concentrates around the population minimiser $\theta_*$ with probability 1 when $n \to \infty$.
The latter establishes that our generalised posterior is normal around $\theta_*$ in the same asymptotic limit.

\begin{theorem} \label{thm:bvm}
	Suppose Assumptions \ref{asmp:minimiser} and \ref{asmp:derivative} hold.
	Assume that the prior $\pi$ is positive and continuous at $\theta_*$.
	Let $B_\epsilon(\theta_*) := \{ \theta \in \Theta \mid \| \theta - \theta_* \|_2 < \epsilon \}$.
	Then for any $\epsilon > 0$,
	\begin{align}
		\int_{B_\epsilon(\theta_*) } \pi_n^{D}(\theta) \mathrm{d} \theta \as 1 \qquad \text{ as } n \rightarrow \infty.
	\end{align}
	Denote by $\widetilde{\pi}_n^{D}$ a density on $\R^d$ of a random variable $\sqrt{n} (\theta - \theta_n)$ for $\theta \sim \pi_n^{D}$.
	If $H_* := \beta \nabla_{\theta}^2 \DFD(p_\theta \| p) |_{\theta = \theta_*}$ is nonsingular, then
	\begin{align}
		\int_{\R^p} \left| \widetilde{\pi}_n^{D}(\theta) - \frac{1}{ \sqrt{ \text{\normalfont det}(2 \pi H_*^{-1})} } \exp\left( - \frac{1}{2} \theta \cdot H_* \theta \right) \right| \mathrm{d} \theta \to 0  \qquad \text{ as }  n \rightarrow \infty .
	\end{align}
\end{theorem}

\noindent The proof of \Cref{thm:bvm} is provided in \Cref{apx: proof_bvm}.
The result was established using similar arguments from early work by \cite{hooker2014bayesian, Ghosh2016} and extended techniques of \cite{Miller2019, Matsubara2021}.

\subsection{A New Approach to Calibration of Generalised Posteriors} \label{sec:calibration}

The weight $\beta$ in \eqref{eq: beta in gen bayes} controls the scale of the generalised posterior, and the selection of an appropriate value for $\beta$ is critical to ensure the generalised posterior is calibrated.
The literature on this topic is under-developed, but two existing approaches stand out. 
The first approach was proposed in the recent review paper of \citet{Syring2019}. It consists of a new stochastic sequential update algorithm for choosing $\beta$, such that a $95\%$ highest posterior density region coincides with a $95\%$ confidence interval. Unfortunately, this approach leads to a large computational cost and is therefore often  impractical.
The second approach is due to \citet{Lyddon2018} and consists in setting $\beta$ such that the scale of the posterior's asymptotic covariance matrix coincides with that of a frequentist counterpart with correct coverage.
\citet{Matsubara2021} numerically showed that this approach is unstable when $n$ is not large enough or when $\theta$ is high dimensional.
In addition, the second approach does not take the prior $\pi$ into account, because it depends only the generalised posterior's asymptotic covariance matrix.

In order to remedy some of these issues, the present paper proposes a novel selection criterion for $\beta$ that can be viewed as a more analytically tractable alternative to \citet{Syring2019}.
This criterion is applicable to generalised posteriors beyond DFD-Bayes and may therefore be of independent interest. 
Our approach consists of two steps: (i) computing minimisers of $B$ ``bootstrapped" losses and (ii) estimating an appropriate value of $\beta$ using the closed-form expression in \Cref{thm:beta_choice}.
In contrast to \citet{Syring2019}, step (ii) is non-iterative and exact.
Additionally, computation of each minimiser in step (i) is embarrassingly parallel.
Relative to the approach of \citet{Lyddon2018}, the advantage of our method is that it does not rely on asymptotic quantities, takes the prior into account, and maintains numerical stability even if the parameter $\theta$ is high-dimensional.

To describe the method we first define the minimiser $\theta_n \in \argmin_{\theta \in \Theta} D_n(\theta)$, where $D_n$ is a loss function based on a dataset $\{ \bm{x}_i \}_{i=1}^{n}$. 
To make the dependence on $\beta$ explicit, we denote the posterior $\pi_n^D$ by $\pi_{n, \beta}^D$.
In step (i), bootstrap datasets $ \{\bm{x}_i^{(b)}\}_{i=1}^n$, $b = 1,\dots,B$, are generated by sampling each $\bm{x}_i^{(b)}$ uniformly with replacement from the original dataset.
Then, for each bootstrap dataset, we compute a minimiser $\theta^{(b)}_n = \argmin_{\theta \in \Theta} D_n^{(b)}(\theta)$, where the superscript indicates that $D_n^{(b)}$ is based on the $b^{\text{th}}$ bootstrap dataset. 
This leads to an empirical measure $\delta_\theta^B = \frac{1}{B} \sum_{b=1}^B \delta (\theta_n^{(b)})$ which approximates the sampling distribution of the estimator $\theta_n$. 
In step (ii), we choose $\beta$ to minimise a statistical divergence between $\pi_{n, \beta}^{D}$ and $\delta_\theta^B$.
However, this is not straight-forward, since the majority of statistical divergences (e.g. Kullback--Liebler divergence) require the normalising constant of $\pi_{n,\beta}^D$ for every $\beta$.
Interestingly, this is the same computational challenge posed by intractable likelihood.
Our proposal is therefore to employ a divergence that circumvents computational of the normalisation constant; here we minimise the score matching loss in the continuous domain $\Theta$ \citep{Hyvarinen2005}:
\begin{align}
	\beta_* \in \argmin_{\beta > 0} \frac{1}{n} \sum_{b=1}^{B} \big\| \nabla \log \pi_{n, \beta}^D\big( \theta_n^{(b)} \big) \big\|^2 + 2 \operatorname{Tr}\big( \nabla^2 \log \pi_{n, \beta}^D\big( \theta_n^{(b)} \big) \big) . 
	\label{eq: beta star main text}
\end{align}
This leads to an explicit score-matching estimator for $\beta$, circumventing intractability of \eqref{eq: beta in gen bayes}: 
\begin{theorem} \label{thm:beta_choice}
    Consider a generalised posterior $\pi_{n, \beta}^D$  with twice differentiable loss function $D_n: \Theta \to \R$.
    Suppose that there exists at least one $\theta_n^{(b)}$ s.t.~$\nabla_\theta D_n(\theta_n^{(b)}) \ne 0$ and that $\sum_{b=1}^{B} \nabla_\theta D_n(\theta_n^{(b)}) \cdot \nabla_\theta \log \pi(\theta_n^{(b)}) + \operatorname{Tr}( \nabla_\theta^2 D_n(\theta_n^{(b)}) ) > 0$.
    Then $\beta_*$ in \eqref{eq: beta star main text} is unique, with
    \begin{align}
        \beta_* = \frac{ \sum_{b=1}^{B} \nabla_\theta D_n(\theta_n^{(b)}) \cdot \nabla_\theta \log \pi(\theta_n^{(b)}) + \operatorname{Tr}( \nabla_\theta^2 D_n(\theta_n^{(b)}) ) }{ \sum_{b=1}^{B} \| \nabla_\theta D_n(\theta_n^{(b)}) \|^2 } > 0 . \label{eq:beta_optimal}
    \end{align}
\end{theorem}

\noindent The proof is provided in \Cref{apx: proof_betaclib}.
The condition in \Cref{thm:beta_choice} directly implies existence and positivity of \eqref{eq:beta_optimal}.
However, in practice, computing \eqref{eq:beta_optimal} and verifying the existence and positivity directly
is strikingly easier than validating the local convexity of $D_n$ and $\log \pi$.
Note that \eqref{eq:beta_optimal} is straight-forward to compute whenever the loss $D_n$ is amenable to automatic differentiation.
For completeness, we also provide an explicit expression in \Cref{apx:example_2} for the case of the DFD-Bayes posterior with an exponential family model.

\begin{remark}
    Step (i) of our algorithm is embarrassingly parallelisable over bootstrap samples. Each component inside the sum in \eqref{eq:beta_optimal} can also be computed in parallel during step (ii). Overall, the total cost can be reduced linearly in the number of available cores $K$, and the cost of step (ii) is $O( p^2 \times C \times B / K)$, where $C$ is the cost of evaluating $D_n(\theta)$ and $\pi(\theta)$ at $\theta$.
\end{remark}

\subsection{Limitations} \label{sec:limitation}

There are at least two important limitations of the DFD-Bayes methodology, which will now be discussed.
First, DFD-Bayes was not derived as an approximation to standard Bayesian inference, and thus the semantics associated with the generalised posterior should not be confused with the semantics of standard Bayesian inference; see \cite{Bissiri2016,Knoblauch2019} for a detailed discussion of this point.
In particular, we need to calibrate DFD-Bayes through the selection of $\beta$, which is not a feature of standard Bayesian inference under well-specified models.
Although we expect our bootstrap approach to outperform existing alternative approaches for small sample size $n$, it is possible that in those cases the bootstrap criterion for selecting $\beta$ in \Cref{sec:calibration} will fail, and in these circumstances the generalised posterior will fail to be calibrated.
Second, the generalised posterior may suffer from similar drawbacks to score-based methods for continuous data, including insensitivity to mixing proportions \citep{wenliang2020blindness}.
Indeed, for a two-component mixture model $p_\theta(\bm{x}) = (1 - \theta) p_1(\bm{x}) + \theta p_2(\bm{x})$, we can compute the ratios
\begin{align*}
\rho_j := \frac{p_\theta(\bm{x}^{j-})}{p_\theta(\bm{x})} = \frac{  (1 - \theta) p_1(\bm{x}^{j-}) + \theta p_2(\bm{x}^{j-}) }{ (1 - \theta) p_1(\bm{x}) + \theta p_2(\bm{x}) } 
\end{align*}
on which the discrete Fisher divergence is based.
Suppose, informally, that the high probability regions $R_1$ of $p_1$ and $R_2$ of $p_2$ are separated, meaning $p_2 \approx 0$ on $R_1$ and $p_1 \approx 0$ on $R_2$.
Then these ratios are approximately independent of $\theta$ on $R_1 \cup R_2$, since $\rho_j \approx p_1(\bm{x}^{j-}) / p_1(\bm{x})$ for $\bm{x} \in R_1$ and $\rho_j \approx p_2(\bm{x}^{j-}) / p_2(\bm{x})$ for $\bm{x} \in R_2$.
It follows that $\operatorname{DFD}(p_\theta \| p_n)$ is approximately independent of $\theta$ whenever the data $\{\bm{x}\}_{i=1}^n \subseteq R_1 \cup R_2$.
See \Cref{subsec:limitation_poi_mix} for an empirical demonstration using a mixture model of two Poisson distributions.
Thus, although DFD-Bayes may be applied to mixture models, supported by the theoretical guarantees of \Cref{thm:bvm}, the inferences for mixing proportions so-obtained can be data-inefficient.

\section{Experimental Assessment}
\label{sec: experiment}

To complement the theoretical assessment we now provide a detailed empirical assessment, focusing on three important instances of discrete intractable likelihood.
First, in \Cref{subsec: CMP} we consider a relatively simple model for over- and under-dispersed count data, called the Conway--Maxwell--Poisson model.
\Cref{subsec: Ising assessment} concerns an application to Ising-type models for discrete spatial data.
Finally, we apply DFD-Bayes to perform inference for the parameters of flexible multivariate models for count data in \Cref{subsec: multivar count}.
Source code to reproduce these experiments can be downloaded from \url{https://github.com/takuomatsubara/Discrete-Fisher-Bayes}.

\subsection{Conway--Maxwell--Poisson Model}
\label{subsec: CMP}

The first model we consider is a generalisation of the Poisson model for over- and under-dispersed count data, due to \citet{Conway1962}.
This model is on $\X = \mathbb{N}\cup\{0\}$ (hence $d=1$ and $\text{card}(\X) = \infty$) and generalises the Poisson distribution through the inclusion of an additional parameter controlling how the data are dispersed. 
Since the work of \citet{Shmueli2005}, this model has been used in a wide range of fields including transport, finance and retail. 
The model has two parameters $\theta \in \Theta = (0, \infty)^2 \cup ([0,1] \times \{0\}) $ (and hence $p=2$) and its probability mass function is given by $p_\theta(x) = \tilde{p}_{\theta}(x) Z_\theta^{-1}$ where $\tilde{p}_\theta(x) = (\theta_1)^x (x!)^{-\theta_2}$. The normalising constant is given by $Z_\theta = \sum_{y=0}^{\infty} \tilde{p}_\theta(y)$, which has no analytical form  except for certain special cases of $\theta \in \Theta$, including the case $\theta_2 = 1$ for which the standard Poisson model is recovered.

\begin{figure}[t!]
	\centering
	\hfill
	\includegraphics[height=0.3\textheight, trim={1.0cm 0.0cm 0.0cm 0.0cm}, clip]{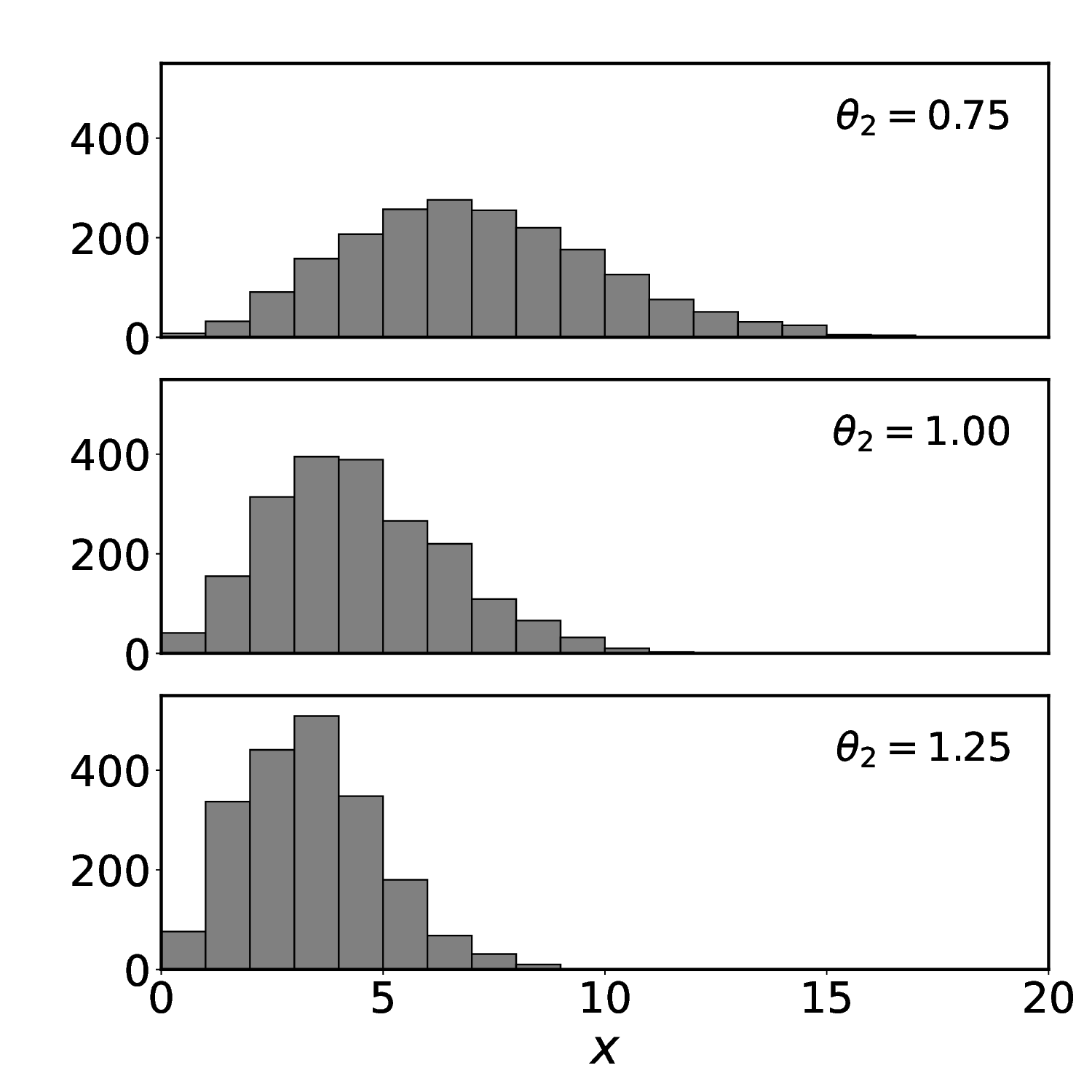}
	\hfill
	\includegraphics[height=0.3\textheight]{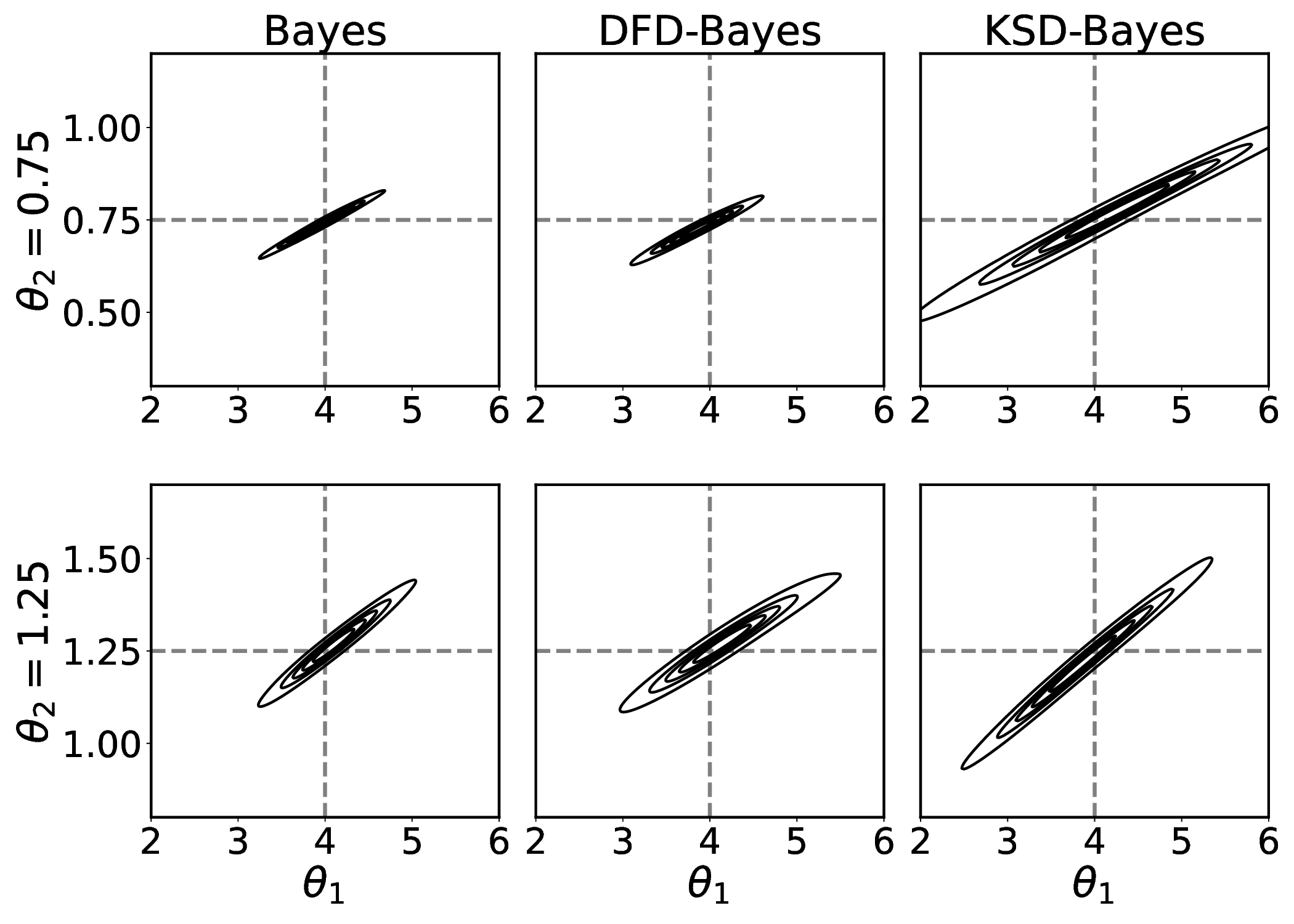}
	\hfill
	\hfill
	\caption{Comparison of standard Bayesian inference with the generalised posteriors from DFD-Bayes and KSD-Bayes on the Conway--Maxwell--Poisson model in the over-dispersed case $\theta_2 = 0.75$ and the under-dispersed case $\theta_2 = 1.25$ for $n=2,000$.}
	\label{fig:CMP_posteriors}
\end{figure}

This model is an ideal test-bed for DFD-Bayes: although the likelihood is formally intractable, it is relatively straightforward to directly approximate the normalising constant\footnote{The standard Bayesian inferences reported in this section used the approximation $Z_\theta \approx \sum_{y=0}^{99} \tilde{p}_\theta(y)$ and the associated approximate likelihood. Alternative estimators are available; see  \citet{Benson2021}.}.
This enables a direct comparison with standard Bayesian inference in the case where the model is well-specified.
To this end, we simulated two datasets from the model: (i) an under-dispersed case where $\theta^* = (4,1.25)$, and (ii) an over-dispersed case where $\theta^* = (4,0.75)$, shown in \Cref{fig:CMP_posteriors} (left).
Three inference methods were compared: standard Bayesian inference, the KSD-Bayes method of \cite{Matsubara2021}, and the DFD-Bayes method we have proposed.
The settings of KSD-Bayes are described in \Cref{subsec: CMP extra KSD Bayes}. 
In each case, the prior $\pi$ was taken to be the chi-squared distribution with $3$ degrees of freedom for each of $\theta_1$ and $\theta_2$ independently.
A Metropolis--Hastings algorithm was used to sample from all the posteriors; and details can be found in \Cref{subsec: CMP extra MCMC}.
The weight $\beta$ in DFD-Bayes and KSD-Bayes was calibrated by our approach described in \Cref{sec:calibration}.

\Cref{fig:CMP_posteriors} (right) illustrates the posteriors, based on typical datasets of size $n=2,000$. 
The estimated value of $\beta_*$ was $1.91$ for DFD-Bayes and $5.04$ for KSD-Bayes in the over-dispersed case $\theta_2 = 0.75$, and $0.46$ for DFD-Bayes and $2.51$ for KSD-Bayes in the under-dispersed case $\theta_2 = 1.25$.
The left panel of \Cref{fig:CMP_cost_beta} displays the distribution of calibrated weight $\beta_*$ as in \Cref{sec:calibration} over multiple instances of the dataset, along with the values advocated in \citet{Lyddon2018}.
For both methods, the calibrated weight is stably estimated.

The inferences obtained using DFD-Bayes resembled those obtained using standard Bayesian inference, irrespective of whether the data were over- or under-dispersed.
Those obtained using KSD-Bayes were more conservative than standard Bayes and DFD-Bayes, although the maximum a posteriori estimator approximated the true parameter well.
Note that the credible regions of the generalised posteriors can substantially differ from those of standard Bayesian inference; in our approach a credible region of a generalised posterior is calibrated with reference to the distribution of a corresponding frequentist estimator estimated by bootstrapping, leading to approximately correct frequentist coverage as shown in \Cref{fig:CMP_cost_beta} (middle).
Calibration led to improved inference outcomes for both DFD-Bayes and KSD-Bayes.
In the KSD-Bayes case for example, the value of $\beta_* \ge 1$ intensified the concentration around the true parameter by placing more importance on the loss than the prior. 
In addition, our approach to calibration is relatively more conservative than \citet{Lyddon2018} because the prior is taken into account.

There is a stark difference in computational cost between DFD-Bayes and KSD-Bayes\footnote{The cost of standard Bayesian inference in this experiment is entirely determined by the accuracy with which the normalisation constant is approximated; since direct approximation of the normalisation constant is infeasible in general, we do not report this cost.}, as demonstrated in the right panel of \Cref{fig:CMP_cost_beta}. 
Indeed, the computational cost of DFD-Bayes is seen to increase linearly with $n$, while the cost of KSD-Bayes increases quadratically.

\begin{figure}[t!]
	\centering
	\includegraphics[height=0.18\textheight]{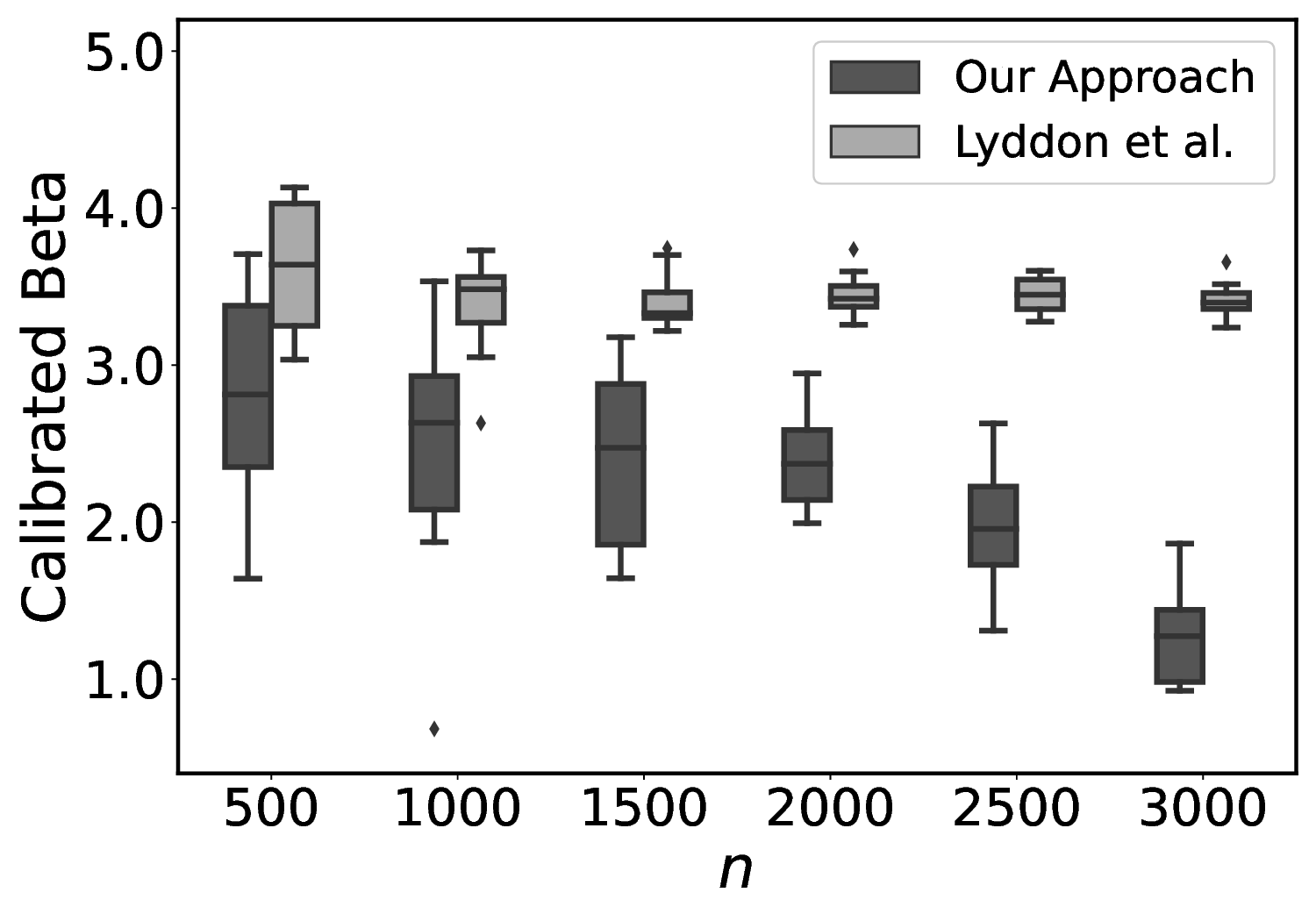}
	\includegraphics[height=0.18\textheight]{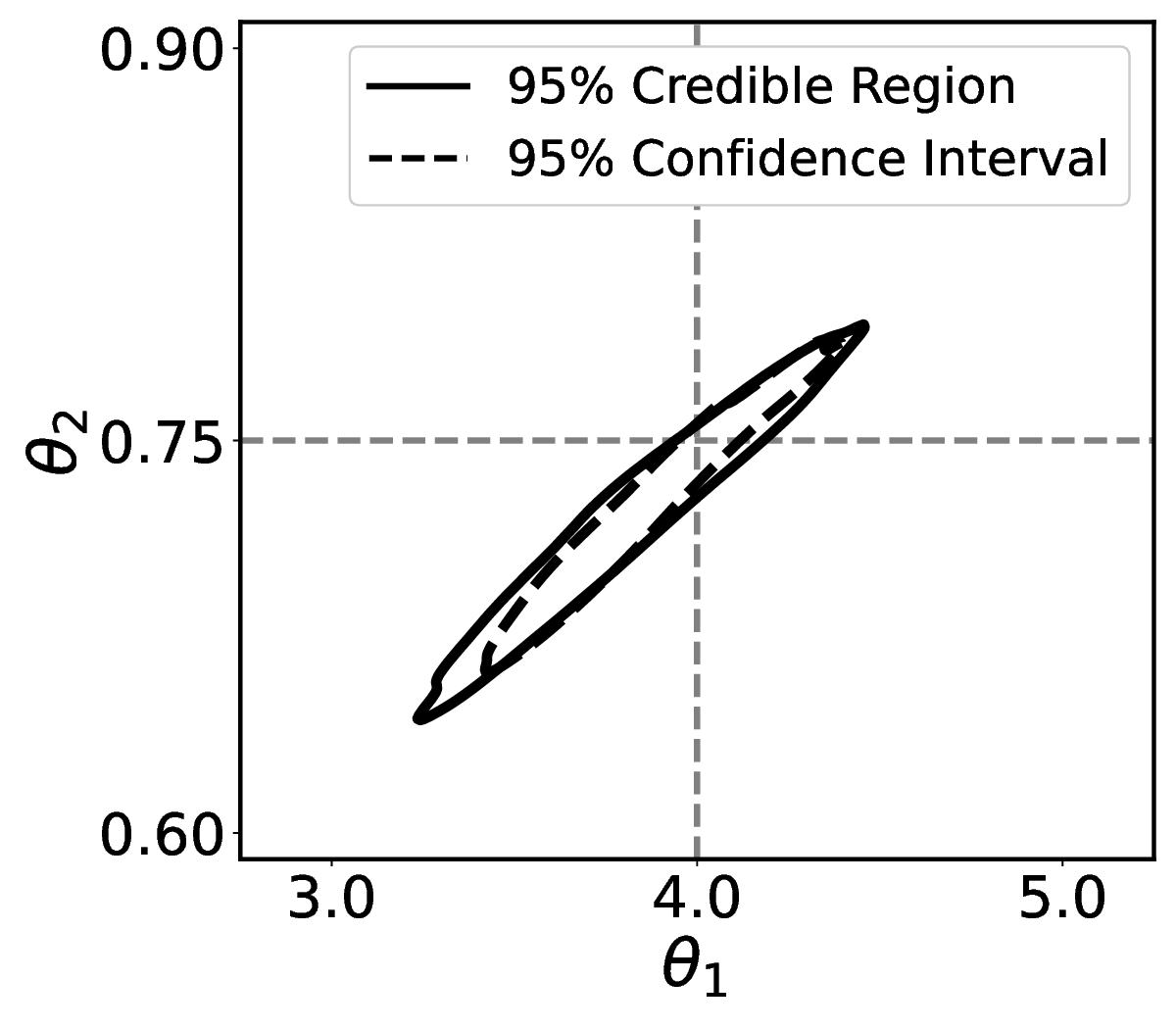}
	\includegraphics[height=0.18\textheight]{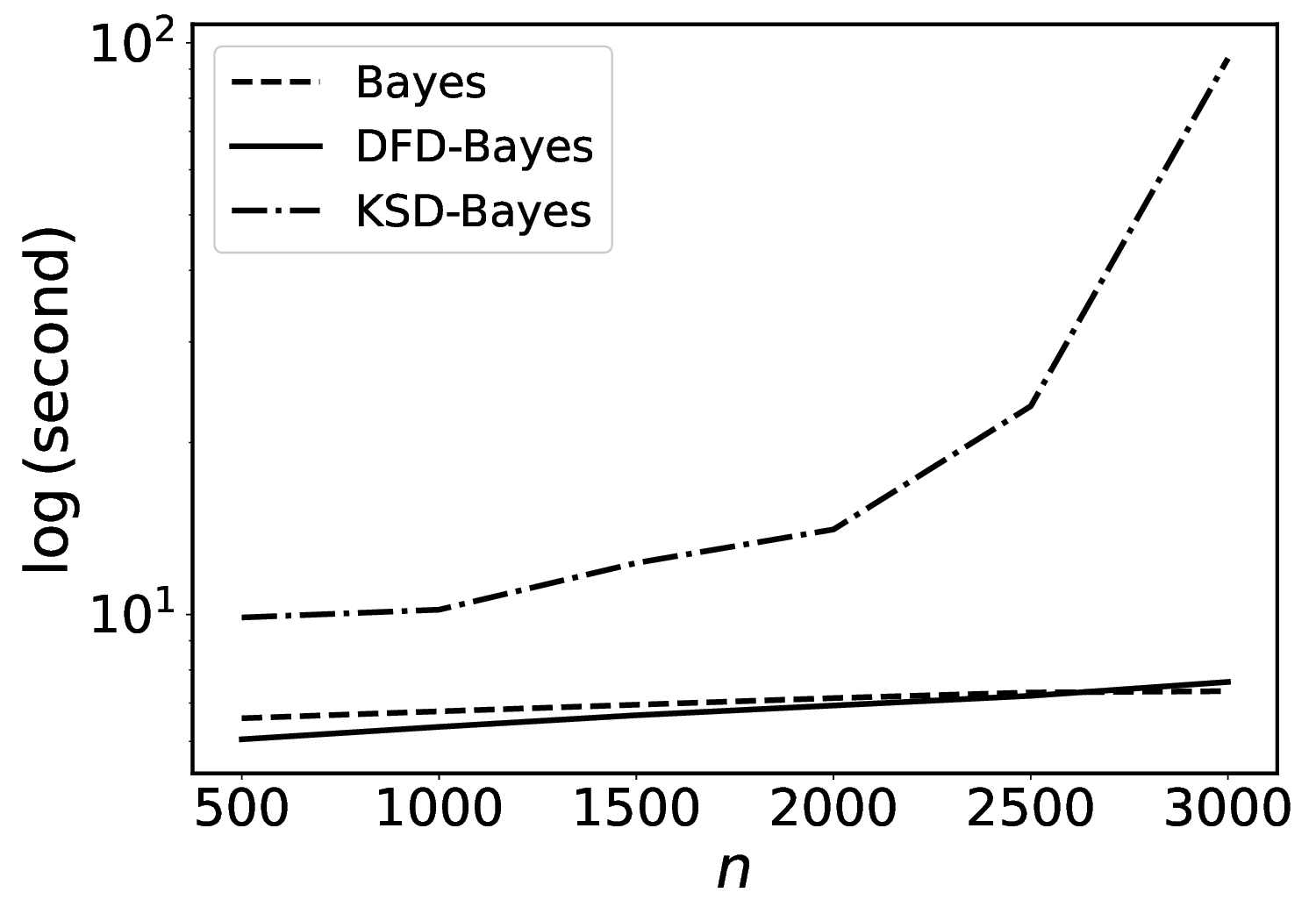}
	\hfill
	\caption{Distribution of $\beta_*$ across different realisations of the dataset at each data number $n$ for $\theta_2 = 0.75$ (left), comparison of a 95\% credible region of the DFD-Bayes posterior and a 95\% confidence interval of the frequentist counterpart for $n = 2000$ (centre), and comparison of computational times of each Metropolis--Hastings algorithm (right). 
	The confidence interval was estimated by a 95\% highest probability density region of a kernel density estimator applied to the $100$ bootstrap minimisers used in calibration of $\beta$.
	} 
	\label{fig:CMP_cost_beta}
\end{figure}

Finally, to assess performance in a real-world data setting, we apply DFD-Bayes to infer the parameters of a Conway--Maxwell--Poisson model using the sales dataset of \cite{Shmueli2005}. 
All relevant details are contained in \Cref{subsec: CMP extra dataset}.
\Cref{fig:CMP_sales} compares our fitted model to a standard Bayesian analysis using the Poisson distribution, which is the closest analysis one can perform without confronting an intractable likelihood. 
As observed in the central panel of \Cref{fig:CMP_sales}, the Poisson model is not able to capture over-dispersion of the data, whereas the Conway--Maxwell--Poisson model fitted using DFD-Bayes, shown in the right panel, provides a reasonable fit. 
The DFD-Bayes posterior (left) appears approximately normal, in line with \Cref{thm:bvm}. 

\begin{figure}[t!]
	\centering
	\includegraphics[height=0.18\textheight]{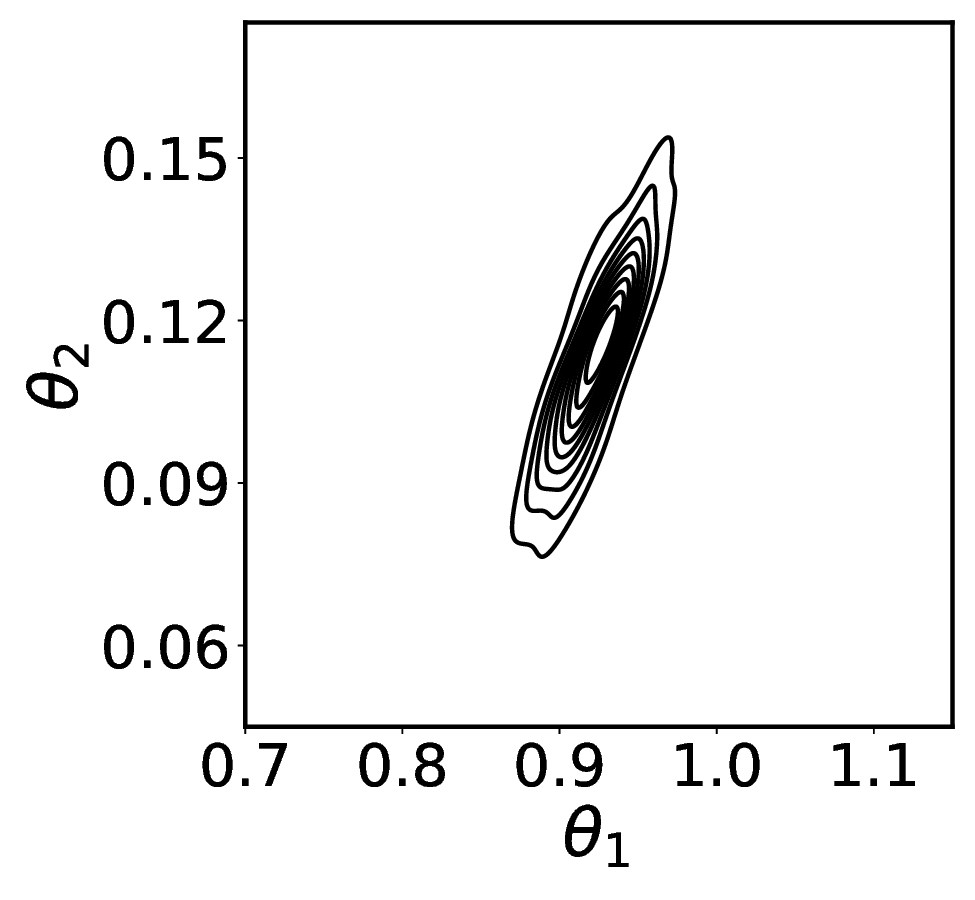}
	\includegraphics[height=0.18\textheight]{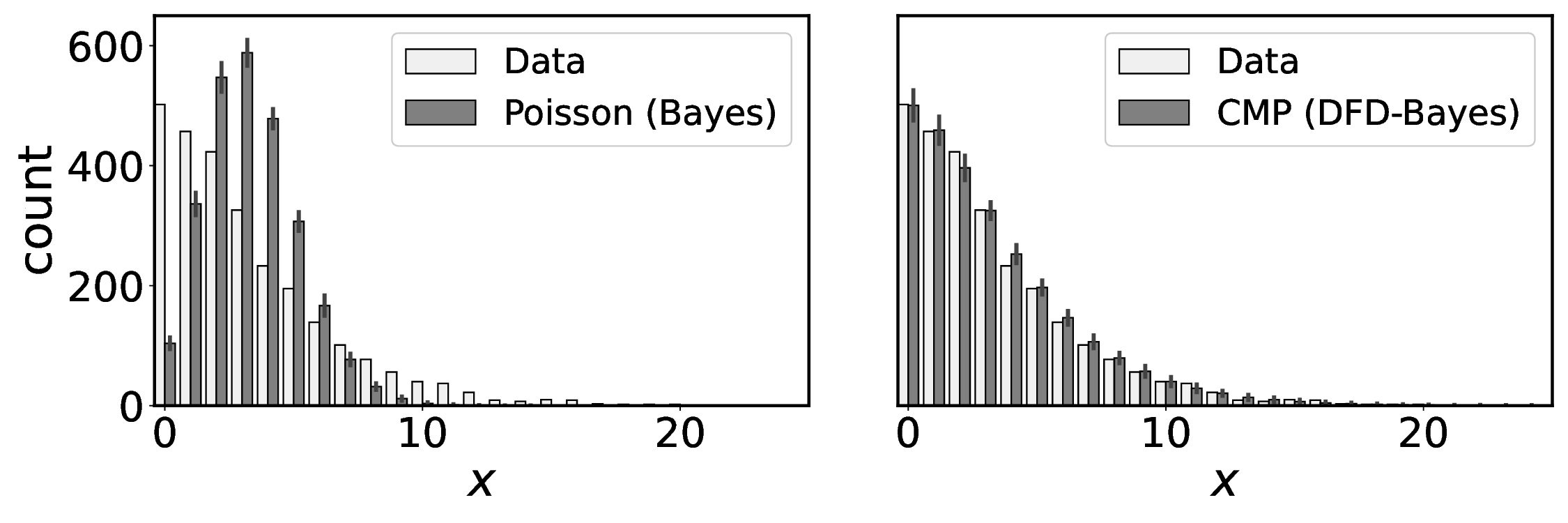}
	\caption{Comparison of DFD-Bayes for the Conway--Maxwell--Poisson model and standard Bayes for the Poisson distribution on the sales data of \cite{Shmueli2005}. 
		Left: The generalised posterior distribution produced using DFD-Bayes.
		Centre:  Posterior predictive distribution, at the level of the data, for a Poisson model with standard Bayesian inference performed.
		Right: Posterior predictive distribution, at the level of the data, for a Conway--Maxwell--Poisson model with DFD-Bayes inference performed.
		In both cases, error bars indicate one standard deviation of the posterior predictive distribution.
	}
	\label{fig:CMP_sales}
\end{figure}

\subsection{Ising Model} \label{subsec: Ising assessment}

The aim of this section is to consider a more challenging instance of discrete intractable likelihood, where the data are high-dimensional (i.e. $d$ is large) and the cardinality of each coordinate domain $S_i$ is small.
A small cardinality of $S_i$ is particularly interesting, because the intuition that our difference operators arise from discretisation of continuous differential operators fails to hold.
This setting is typified by the Ising model (which has $S_i =  \{0,1\}$), variants of which are used to model diverse phenomena, e.g., the network structure of the amino-acid sequences \citep{xue2012nonconcave}.
The computational challenge of performing Bayesian inference for Ising-type models has, to-date, principally been addressed using techniques such as pseudo-likelihood
\citep[see the recent survey in][]{Bhattacharyya2019}.
As with the case of generalised Bayesian inference, these do not necessarily lead to the same asymptotic distribution as standard Bayesian inference since the original likelihood is replaced by an approximation \citep{Gong1981}. 

Let $G$ be an undirected graph on a $d$-dimensional vertex set and let $\mathcal{N}_i$ denote the neighbours of node $i$, with self-edges excluded.
An Ising model describes a discrete process that assigns each vertex of $G$ either the value $0$ or $1$, and thus the data domain is $\X = \{ 0, 1 \}^d$. 
The probability mass function has the exponential family form
\begin{align}
	p_\theta(\bm{x}) \propto \exp\left( \frac{1}{\theta} \sum_{i=1}^{d} \sum_{j \in \mathcal{N}_i} x_i x_j \right) \label{eq:ising_model}
\end{align}
where $\theta$ is a temperature parameter, controlling the propensity for neighbouring vertices to share a common value.
Here we consider the ferromagnetic Ising model, which has $\theta \in (0, \infty)$.
To conduct a simulation study we consider the case where $G$ is a $m \times m$ grid.
Simulating from Ising models is challenging due to the high-dimensional discrete domain, so here we restrict attention to $m \in \{5,\dots,10\}$ to ensure that data were accurately simulated.
A total of $n = 1,000$ data points were generated from an Ising model with $\theta = 5$, using an extended run of a Metropolis–-Hastings algorithm, the details of which are contained in \Cref{subsubsec: Ising data}.
A chi-squared prior with degree of freedom $3$ was used.
Three inference methods were compared: the KSD-Bayes method of \citet{Matsubara2021}, the proposed DFD-Bayes method, and standard Bayesian inference based on a the pseudo-likelihood 
\begin{align*}
\tilde{p}_\theta(\bm{x}) = \prod_{i=1}^d p_\theta(x_i | \{x_j : j \in \mathcal{N}_i\}) ,
\end{align*}
where $p_\theta(x_i | \{x_j : j \in \mathcal{N}_i\})$ is a restriction of the original model \eqref{eq:ising_model} to the $i$-th coordinate $x_i$ under the condition $\{x_j : j \in \mathcal{N}_i\}$ that results in a Bernoulli distribution of $x_i$ for each $i = 1, \cdots, d$ \citep{Besag1974}. 
The latter Pseudo-Bayes approach can be viewed as a special case of generalised Bayes inference, since it replaces the original likelihood loss of the model \eqref{eq:ising_model} with the pseudo-likelihood loss, and therefore we also applied the proposed calibration procedure to this method.
The settings of KSD-Bayes are described in \Cref{subsec: Ising KSD}.
A Metropolis--Hastings algorithm was also used to sample from all generalised posteriors, the details for which are contained in \Cref{subsec: Ising extra MCMC}.

\begin{figure}[t!]
	\centering
	\hfill
	\includegraphics[height=0.17\textheight]{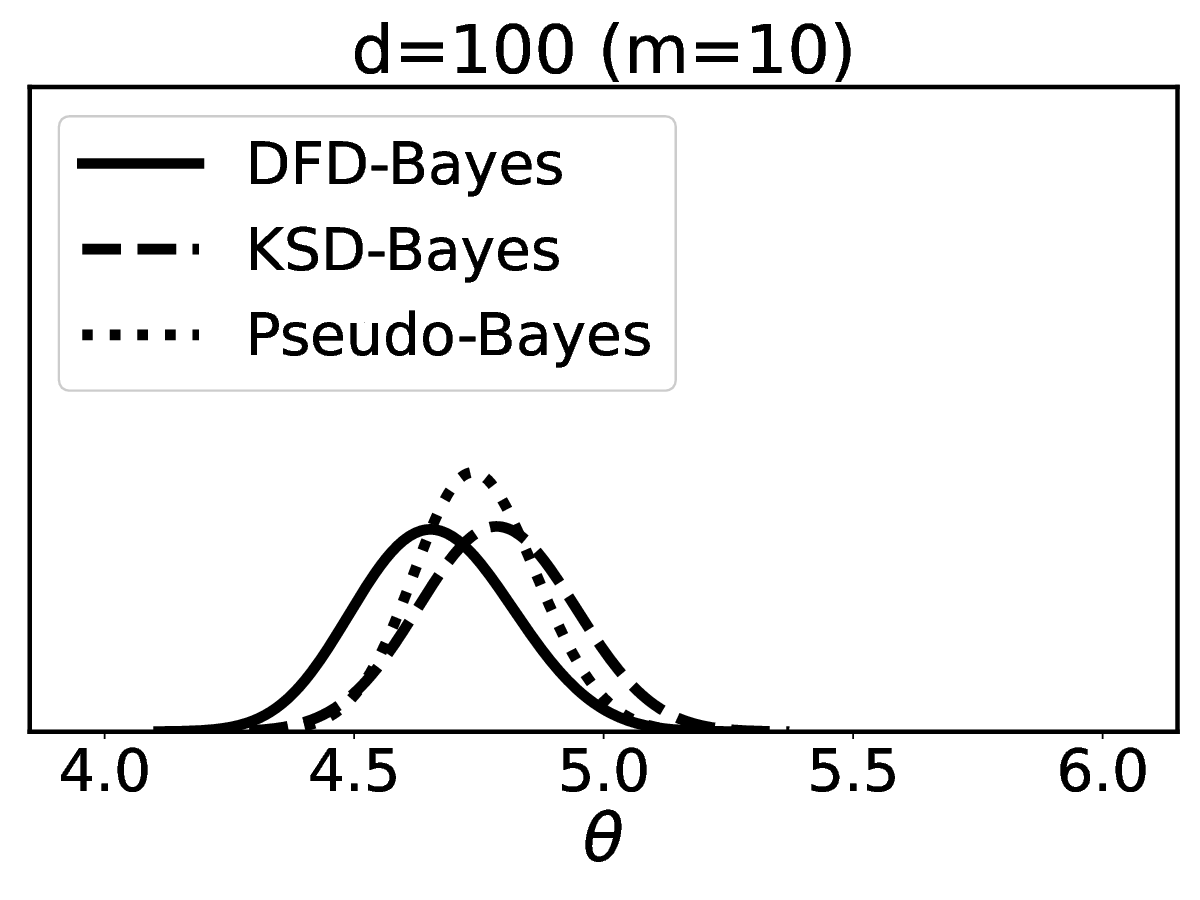}
	\hfill
	\includegraphics[height=0.17\textheight]{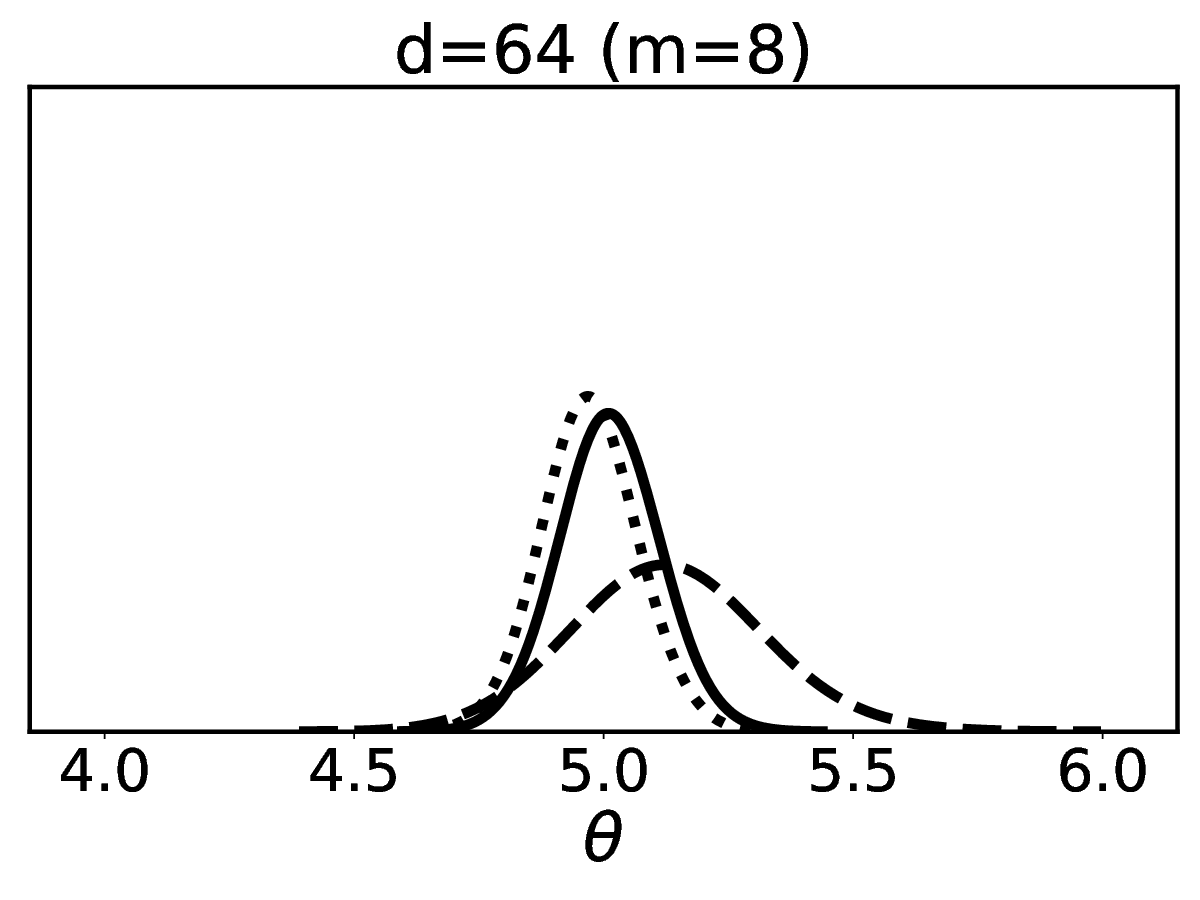}
	\hfill
	\includegraphics[height=0.17\textheight]{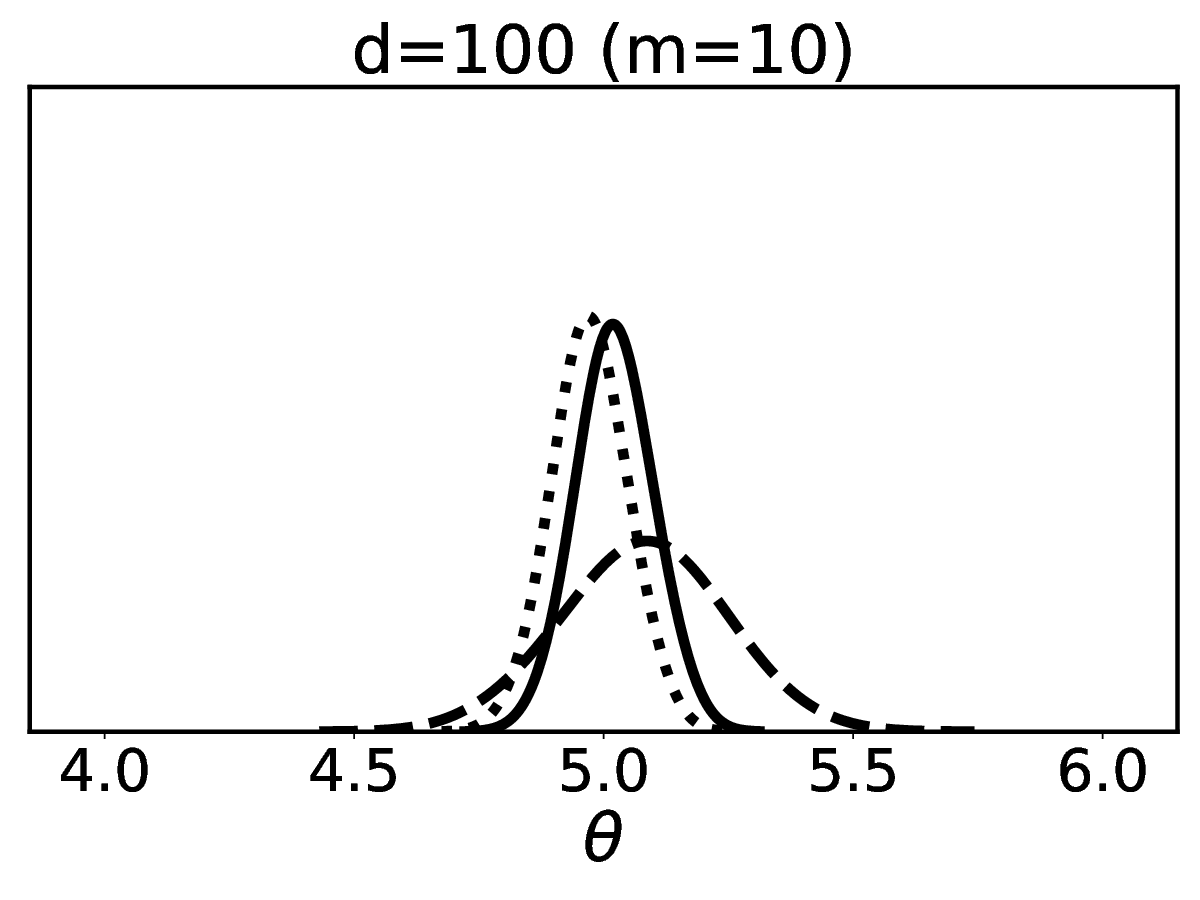}
	\hfill
	\hfill
	\caption{
		Comparison of approximate Bayesian inference based on pseudo-likelihood, DFD-Bayes and KSD-Bayes, applied to the Ising model with $\theta = 5$ for $n=1,000$ and $d = 10 \times 10$. For all methods, the value $\beta_*$ from \Cref{sec:calibration} was used.
	}
	\label{fig:Ising_posteriors}
\end{figure}

Results are presented for three different datasets of size $n = 1,000$ and dimension $d = 36~(m = 6)$, $d = 64~(m = 8)$, and $d = 100~(m = 10)$ in \Cref{fig:Ising_posteriors}.
For the lowest dimension $d = 36$, all the approaches produced similar posteriors.
For the higher-dimensional cases, it can be seen that the DFD-Bayes and Pseudo-Bayes posteriors concentrate around the true parameter $\theta = 5$. The KSD-Bayes posterior is more conservative, whilst DFD-Bayes gives a comparable result to Pseudo-Bayes.
For $d=100$, the total computational time required to perform this analysis (including calibration) was $540$ seconds for DFD-Bayes, $2,353$ seconds for KSD-Bayes, and $1,053$ seconds for Pseudo-Bayes each in average over 10 independent experiments, confirming that DFD-Bayes incurs a significantly lower computational cost than both alternatives.
The value of the weight obtained through our calibration method for $d = 100$ in \Cref{fig:Ising_posteriors} was $0.013$ for DFD-Bayes, $0.157$ for KSD-Bayes, and $0.579$ for Pseudo-Bayes.
These small values of weight indicated that the calibration worked effectively, preventing the over-concentration of each posterior.

\subsection{Multivariate Count Data} \label{subsec: multivar count}

Finally we consider a problem involving multivariate count data.
Count data occur in diverse application areas, and variables in such data are rarely independent, yet the literature on statistical modelling of such data is limited.
Poisson graphical models and their extensions have emerged as a powerful tool for modelling such data; see the recent review of \citet{inouye2017review}.
To the best of our knowledge a complete Bayesian treatment of Poisson graphical models has yet to be attempted\footnote{A pairwise Markov random field whose marginals are close to being Poisson was used in \citet{roy2020nonparametric}, and a specific generalisation of the Conway-Maxwell-Poisson was used in \citet{Piancastelli2021}.}, and we speculate that this is due to the computational challenges of the associated intractable likelihood.
Our aim here is to assess the suitability of DFD-Bayes for learning the parameters of a Poisson graphical model.

Let $G$ be an undirected graph on a vertex set $\{1,\dots,d\}$ and let $\mathcal{M}_i$ denote the neighbours of node $i$ that are contained in the set $\{i+1,\dots,d\}$.
A Poisson graphical model has probability mass function
\begin{align*}
p_\theta(\bm{x}) \propto \exp\left( \sum_{i=1}^d \theta_i x_i - \sum_{i=1}^d \sum_{j \in \mathcal{M}_i} \theta_{i,j} x_i x_j - \sum_{i=1}^d \log(x_i!) \right)
\end{align*}
where the parameters $\theta$ consist of both the linear coefficients $\theta_i \in (-\infty,\infty)$ and the interaction coefficients $\theta_{i,j} \in [0,\infty)$.
Our aim is to reproduce an analysis of a breast cancer gene expression dataset described in \citet{inouye2017review}, but in a generalised Bayesian framework.
For this problem, $n = 878$, $d = 10$, and $p = 64$ which renders the computational cost of $O(n^2 d)$ at every MCMC step and of $O(p^2 n^2 d)$ at calibration associated with KSD-Bayes inefficient. 
Full details of the dataset are contained in \Cref{subsubsec: multivar dataset}.
Independent standard normal priors were employed for each $\theta_i$, and half-normal distributions with scale $( d (d - 1) / 2 )^{-1}$ were employed for each $\theta_{i,j}$.
A No-U-Turn Sampler was used to sample from the DFD-Bayes posterior, as described in \Cref{subsec: multivar extra MCMC}.
The gradient of the discrete Fisher divergence is available whenever $p_\theta(\bm{x}) = q_\theta(\bm{x}) / C(\theta)$ with $q_\theta(\bm{x})$ differentiable with respect to $\theta$ at any $\bm{x} \in \X$; see \Cref{subsec: DFD gradient}.
The total computational time required for this analysis, including calibration, was $1,896$ seconds.
Results, in \Cref{fig: PGM results}, demonstrate that the Poisson graphical model is in fact a poor fit for these data, which exhibit under-dispersion relative to the standard Poisson model.
However, in terms of identifying the best parameter values for this model, DFD-Bayes appears to have performed well.

\begin{figure}[t!]
    \centering
    \includegraphics[height=0.19\textheight, trim={6.5cm 1.25cm 7cm 1.25cm}, clip]{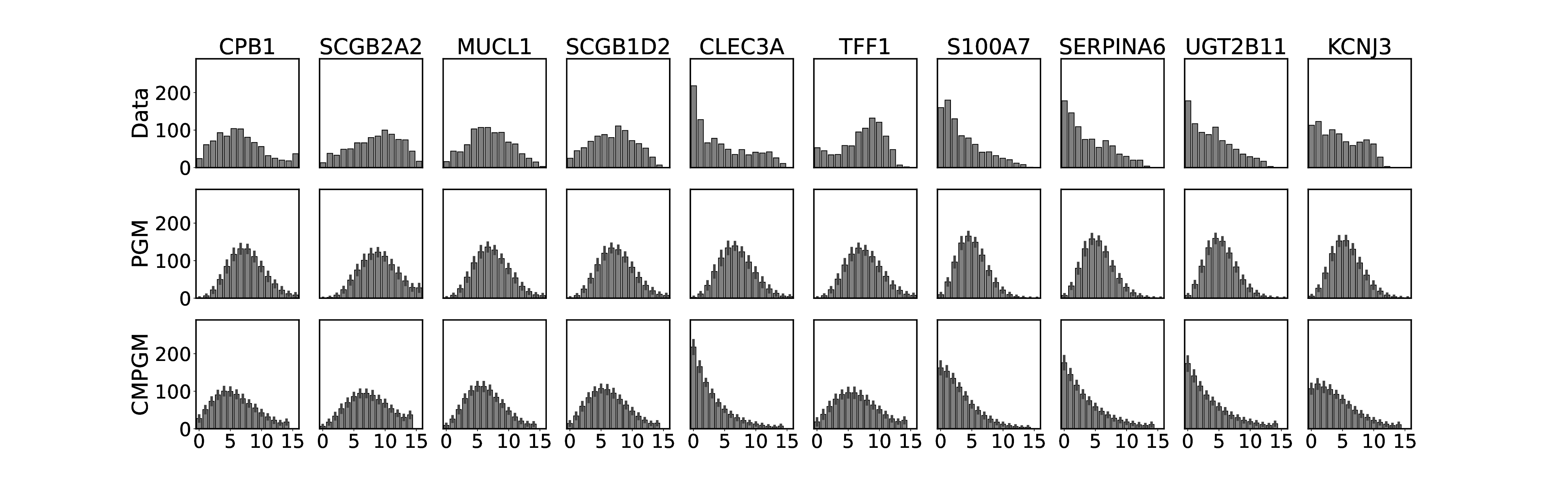}
    \includegraphics[height=0.19\textheight]{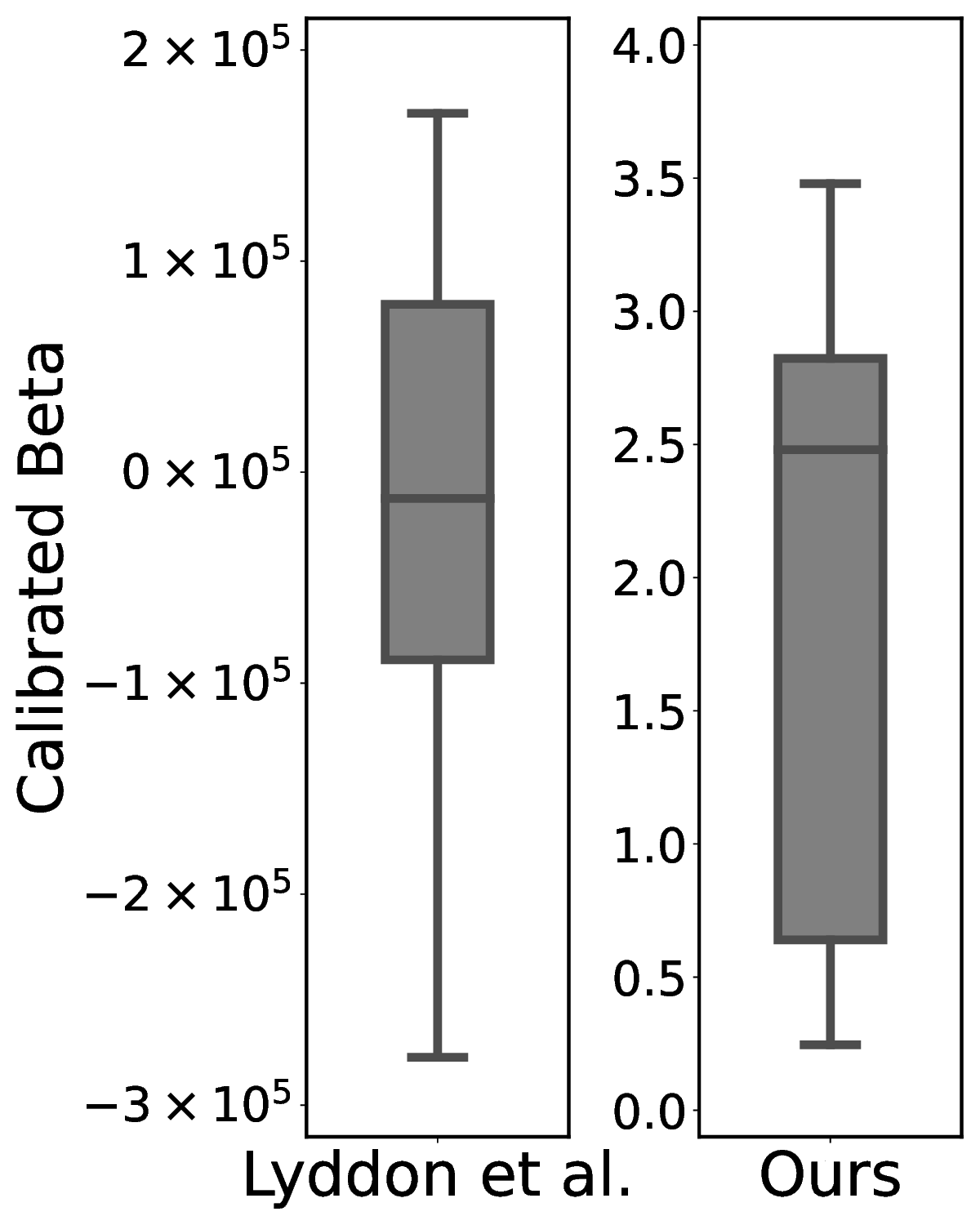}
    \caption{Left: Posterior predictive distributions from the Poisson graphical model and the Conway--Maxwell--Poisson graphical model. Right: Sampling distributions of $\beta_*$ for the Conway--Maxwell--Poisson graphical model by \cite{Lyddon2018} and by the proposed approach, computed using 10 independent realisations of the dataset.
    }
    \label{fig: PGM results}
\end{figure}

As a possible improvement, and to further stress-test the DFD-Bayes method, we considered a generalisation of the Poisson graphical model that allows for over- and under-dispersion, analogous to \citet{Conway1962}.
This model takes the form
\begin{align*}
p_\theta(\bm{x}) \propto \exp\left( \sum_{i=1}^d \theta_i x_i - \sum_{i=1}^d \sum_{j \in \mathcal{M}_i} \theta_{i,j} x_i x_j - \sum_{i=1}^d \theta_{0,i} \log(x_i!) \right)
\end{align*}
where the additional parameters $\theta_{0,i} \in [0,\infty)$ control the dispersion, with $\theta_{0,i} = 1$ recovering the standard Poisson marginal. This time, $p = 74$ as opposed to $p = 64$ for the Poisson-based model.
For this Conway--Maxwell--Poisson graphical model, the same priors as the Poisson graphical model were used for $\theta_i$ and $\theta_{i,j}$, and half-normal priors with scale $1/\sqrt{2}$ were used for each $\theta_{0,i}$.
Results in \Cref{fig: PGM results} demonstrate an improved fit to the dataset.
Indeed, the optimal $\beta$ for the Poisson graphical model was $\beta_* = 0.2150$, which is smaller than the corresponding value $\beta_* = 0.9971$ for the Conway--Maxwell--Poisson graphical model, resulting in a conservative inference outcome when the statistical model is most misspecified and supporting the effectiveness of the proposed approach to calibration.

The right panel of \Cref{fig: PGM results} shows the sampling distributions of estimators for the weight $\beta$ in the context of the Conway--Maxwell--Poisson graphical model, computed using bootstrap resampling of the gene expression dataset.
It can be seen that the asymptotic approach proposed in \citet{Lyddon2018} is severely numerically unstable and can even lead to a negative weight, while the approach proposed in \Cref{sec:calibration} remains stable within a reasonable range between $0$ and $3.5$.
The lack of stability of the approach by \citet{Lyddon2018} arises from the need to invert a covariance matrix of derivatives of the loss, which can become numerically singular if the parameter dimension is high.
In contrast, our approach involves no matrix inversion.
This real-data analysis using flexible parametric models highlights the value in being able to perform rapid and automatic (i.e. free from user-specified degrees of freedom) generalised Bayesian inference for discrete intractable likelihood.

\section{Conclusion}

This paper proposed a novel generalised Bayesian inference procedure for discrete intractable likelihood.
The approach, called DFD-Bayes, is distinguished by its lack of user-specified hyperparameters, its suitability for standard Markov chain Monte Carlo algorithms, and its linear (in $n$, the size of the dataset) computational cost per-iteration of the Markov chain.
Furthermore, the generalised posterior is consistent and asymptotically normal.
This paper also established a novel approach to calibration of generalised Bayesian posteriors which is computationally efficient (through embarrassing parallelism) and numerically stable, even when the parameter of the statistical model is high-dimensional.

This work focused on independent and identically distributed data, meaning that (for example) regression models were not considered.
Relaxing the independence and identical distribution assumptions represents a natural direction for future work, and a road map is provided by recent research in the score-matching literature \citep{xu2022generalized}.

One of our technical contributions is to present a discrete Fisher divergence applicable to distributions defined on multi-dimensional and countably infinite sets.
This divergence can be regarded as a proper local scoring rule, which complements existing methodology developed in the finite domain context in \citet{dawid2012proper}. 
The use of scoring rules as loss functions within a generalised Bayesian framework for continuous data was considered in \citet{Giummole2019,pacchiardi2021generalized}, and our work can be seen as an analogous approach for discrete data, with particular focus on intractable likelihood.

DFD-Bayes was demonstrated to outperform the comparative approach, KSD-Bayes, in our experiments both in terms of inferential performance and computational cost.
However, one of the significant advantages of KSD-Bayes is robustness in the presence of outliers contained in dataset \citep{Matsubara2021}. 
This is confirmed through an additional experiment on the Ising model in \Cref{subsec: KSD_robustness}
Thus, in settings where robust inference is required, the KSD-Bayes approach may be preferred. 
Future work could however focus on generalising our construction of the discrete Fisher divergence to allow for further robustness as per the diffusion score-matching framework of \citet{Barp2019}.

\paragraph{Acknowledgements}
TM, FXB and CJO were supported by the EPSRC grant EP/N510129/1 and the programme on Data Centric Engineering at the Alan Turing Institute, UK.
JK was funded by the EPSRC fellowship grant EP/W005859/1 and the Biometrika fellowship.
The authors are grateful to Pierre Alquier, Lester Mackey and Wenkai Xu, for comments on an earlier draft of the manuscript.

\bibliography{bibliography}

\newpage 
\appendix
\begin{center}
    {\Large\bf SUPPLEMENTARY MATERIAL}
\end{center}

This supplementary material is structured as follows:
Illustrative analysis of the discrete Fisher divergence and the DFD-Bayes using simple tractable models is presented in \Cref{app: walk-through_example}.
The proofs for all theoretical results are contained in \Cref{app: proofs}, with the proof of an auxiliary result reserved for \Cref{sec: proof_con_bvm}.
The relationship between discrete Fisher divergence and Stein discrepancies is explored in \Cref{app: ksd connection}.
Detailed calculations for worked examples are provided in \Cref{ap: calcs for examples}.
Full details on our numerical experiments are provided in \Cref{app:add_exp}

\section{Illustrative Analysis with Tractable Models}\label{app: walk-through_example}

This section provides illustrative analysis of DFD-Bayes, including comparison with standard Bayesian inference and KSD-Bayes, using simple tractable models.
We first demonstrate the calculation of the discrete Fisher divergence using the Bernoulli model.
We then compare the properties of DFD-Bayes with standard Bayesian inference and KSD-Bayes, using the same Bernoulli model.
We next discuss the influence of model misspecification on each posterior using the Poisson model.
Finally, we provide an empirical illustration of the limitations of the discrete Fisher divergence discussed in \Cref{sec:limitation}.
The Bernoulli and Poisson models are used for illustration and comparison in this section, since they are tractable and enable standard Bayesian inference to be performed.

\subsection{The Discrete Fisher Divergence for the Bernoulli Model}

For $x \in \{0, 1\}$, the Bernoulli model can be expressed by
\begin{align}
	p_\theta(x) = \theta^x (1 - \theta)^{1 - x}
\end{align}
where $\theta$ is the probability of $x = 1$.
Recall that $p_\theta(1^+) = p_\theta(0)$ and $p_\theta(0^-) = p_\theta(1)$ under our increment/decrement rule. 
Both the increment and decrement of $p_\theta(1)$ are simply equal to $p_\theta(0)$, and likewise both the increment and decrement of $p_\theta(0)$ are equal to $p_\theta(1)$.
Hence, they can be expressed by
\begin{align}
	p_\theta(x^+) =  p_\theta(x^-) = \theta^{1 - x} (1 - \theta)^{x} ,
\end{align}
that is $p_\theta(x^+) = \theta$ if $x = 0$ and $p_\theta(x^+) = 1 - \theta$ if $x = 1$. 
Plugging these into equation (5) in the manuscript with $d = 1$ gives an explicit form of the discrete Fisher divergence:
\begin{align}
	\text{DFD}(p \| q) \overset{\theta}{=} \frac{1}{n} \sum_{i=1}^{n} \left( \frac{ \theta^{1 - x_i} (1 - \theta)^{x_i} }{ \theta^{x_i} (1 - \theta)^{1 - x_i} } \right)^2 - 2 \left( \frac{ \theta^{x_i} (1 - \theta)^{1 - x_i} }{ \theta^{1 - x_i} (1 - \theta)^{x_i} } \right) \label{eq:dfd_ber}
\end{align}
\Cref{fig:figure_ber_loss} shows the discrete Fisher divergence in \eqref{eq:dfd_ber} computed in three cases where 500 random samples are generated from the Bernoulli model with $\theta = 0.1$, $\theta = 0.5$ and $\theta = 0.9$, comparing the loss surface geometry with that of the negative log-likelihood.
Both of the losses identify the parameter correctly in each case.

\begin{figure}[t]
	\hspace{-50pt}
	\includegraphics[width=1.2\textwidth]{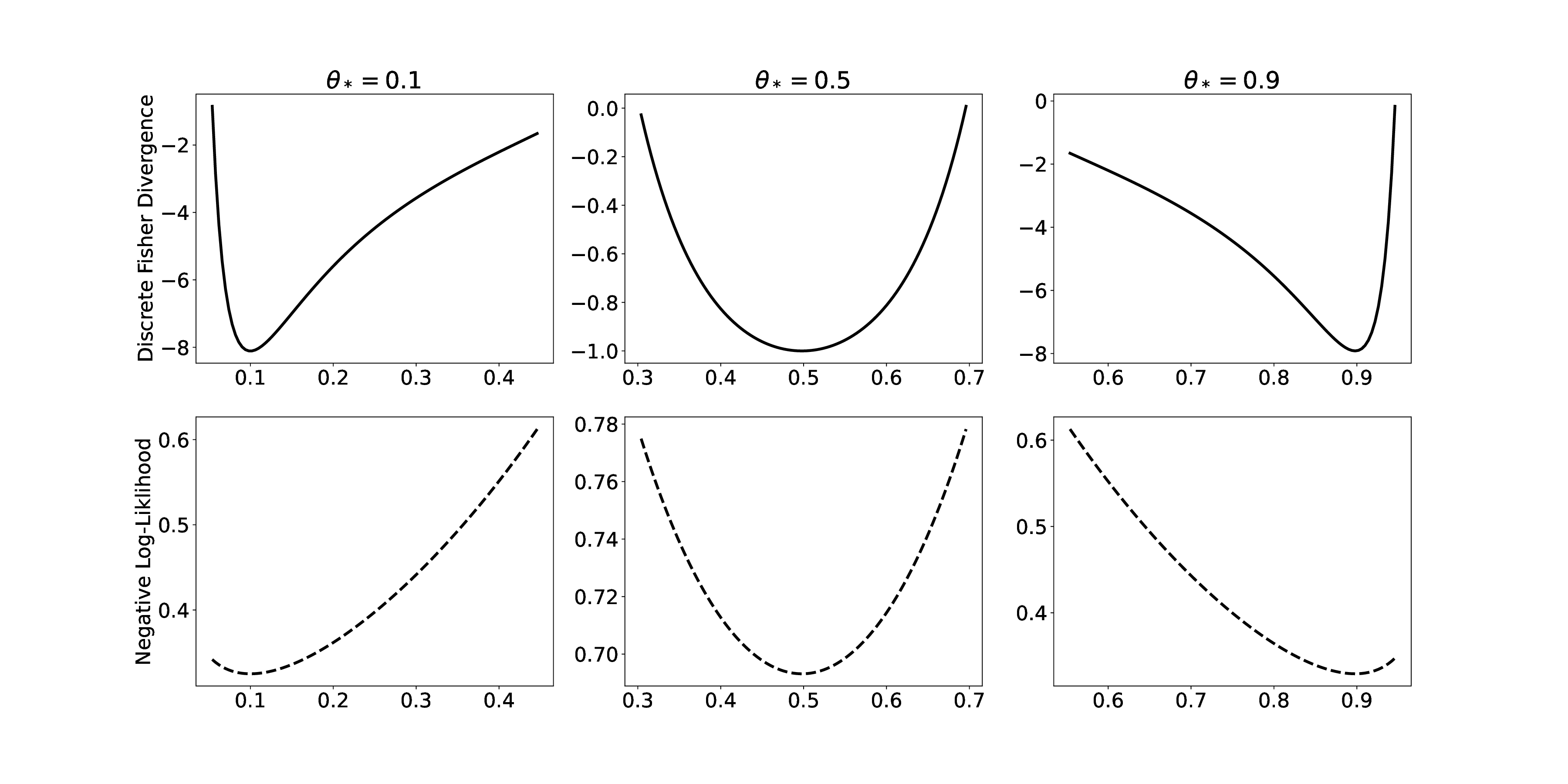}
	\caption{The discrete Fisher divergence (top, solid) and the negative log-likelihood (bottom, dash) between the Bernoulli model and data generated from the Bernoulli model of three different parameters $\theta_* = 0.1$ (left), $\theta_* = 0.5$ (centre), and $\theta_* = 0.9$ (right). They both identify the correct parameter $\theta_*$ in each case albeit the different loss surface geometries.} \label{fig:figure_ber_loss}
\end{figure}

Although the geometrical shape of \eqref{eq:dfd_ber} is different from the negative log-likelihood, we can observe in \Cref{fig:figure_ber_loss} that the discrete Fisher divergence is symmetric under the relabelling $y_i = 1 - x_i$ similarly to the negative log-likelihood in this example.
This can indeed be verified as follows.
If all data are relabelled, the above formula corresponds to
\begin{align}
	\text{DFD}(p \| q) \overset{\theta}{=} \frac{1}{n} \sum_{i=1}^{n} \left( \frac{ \theta^{y_i} (1 - \theta)^{1 - y_i} }{ \theta^{1 - y_i} (1 - \theta)^{y_i} } \right)^2 - 2 \left( \frac{ \theta^{1 - y_i} (1 - \theta)^{y_i} }{ \theta^{y_i} (1 - \theta)^{1 - y_i} } \right) .
\end{align}
With a transform of the parameter $\rho = 1 - \theta$ applied, it further corresponds to
\begin{align}
	\text{DFD}(p \| q) \overset{\theta}{=} \frac{1}{n} \sum_{i=1}^{n} \left( \frac{ \rho^{1 - y_i} (1 - \rho)^{y_i} }{ \rho^{y_i} (1 - \rho)^{1 - y_i} } \right)^2 - 2 \left( \frac{ \rho^{y_i} (1 - \rho)^{1 - y_i} }{ \rho^{1 - y_i} (1 - \rho)^{y_i} } \right) . \label{eq:dfd_ber_flip}
\end{align}
It is clear from comparison of \eqref{eq:dfd_ber} and \eqref{eq:dfd_ber_flip} here that the discrete Fisher divergence of $\theta$ based on the original data $x_i$ is equivalent to that of $\rho = 1 - \theta$ based on the relabelled data $y_i = 1 - x_i$.

\subsection{Illustrative Comparison of DFD-Bayes with standard Bayes and KSD-Bayes} \label{sec:comparison_bernoulli}

First, we derive the negative log-likelihood and the kernel Stein discrepancy for the Bernoulli model.
The negative log-likelihood is 
\begin{align}
\text{NLL}(p_\theta \| p_n) & = - \frac{1}{n} \sum_{i=1}^{n} x_i \log(\theta) + (1 - x_i) \log(1 - \theta) . \label{eq:nll_ber}
\end{align} 
The kernel Stein discrepancy in the discrete context was considered in \citet{Yang2018}.
Letting $\rho_-(\theta, x) := p_\theta(x^-) / p_\theta(x) = \theta^{1 - 2 x} (1 - \theta)^{-1 + 2 x}$, the kernel Stein discrepancy given a kernel function $k: \X \times \X \to \R$ is derived as
\begin{multline}
\text{KSD}(p_\theta \| p) \overset{\theta}{=} \frac{1}{n^2} \sum_{i=1}^{n} \sum_{j=1}^{n} \left( 1 - \rho_-(\theta, x_i) \right) k(x_i, x_j) \left( 1 - \rho_-(\theta, x_j) \right) + \\
\left( 1 - \rho_-(\theta, x_i) \right) \left( k(x_i, x_j) - k(x_i, x_j^-) \right) + \left( k(x_i, x_j) - k(x_i^-, x_j) \right) \left( 1 - \rho_-(\theta, x_j) \right) \big] . \label{eq:ksd_ber}
\end{multline}
The DFD-Bayes posterior, the standard posterior, and the KSD-Bayes posterior are recovered from generalised posterior \eqref{eq: beta in gen bayes} built upon losses \eqref{eq:dfd_ber}, \eqref{eq:nll_ber}, and \eqref{eq:ksd_ber}, where $\beta$ is set to $1$ for the standard posterior.

Next, we provide an analytical comparison of the credible regions of each posterior.
As discussed in \Cref{sec: methods}, a generalised posterior produces a credible region that differs from that of a standard posterior even in the asymptotic regime.
For illustration, we derive the asymptotic variance of each posterior for the Bernoulli model.
The asymptotic distribution of each posterior (appropriately centred) follows a Gaussian distribution $\mathcal{N}(0, \sigma^2)$ whose variance $\sigma^2$ is the inverse loss-Hessian at the minimiser $\theta_*$.
To simplify the derivation, we use the Hamming distance kernel $k(x, x') = \mathbbm{1}_{x = x'}$, that is $1$ when $x = x'$ and otherwise $0$, for the kernel Stein discrepancy.
Let $\rho_+(\theta, x) := p_\theta(x) / p_\theta(x^+) = \theta^{- 1 + 2 x} (1 - \theta)^{1 - 2 x}$.
By routine calculation, the second derivatives of each loss in the limit $n \to \infty$ are
\begin{align*}
    \frac{\partial^2}{\partial^2 \theta} \text{NLL}(p_\theta \| p) & = \E_{X \sim p}\bigg[ \frac{X}{\theta^2} + \frac{1 - X}{( 1 - \theta )^2} \bigg] , \\
    \frac{\partial^2}{\partial^2 \theta} \text{DFD}(p_\theta \| p) & = \E_{X \sim p}\bigg[ 2 \rho_-(\theta, X) \frac{\partial^2}{\partial^2 \theta} \rho_-(\theta, X) + 2 \bigg( \frac{\partial}{\partial \theta}  \rho_-(\theta, X) \bigg)^2  - 2 \frac{\partial^2}{\partial^2 \theta} \rho_+(\theta, X) \bigg] , \\
    \frac{\partial^2}{\partial^2 \theta} \text{KSD}(p_\theta \| p) & = \E_{X \sim p}\bigg[ 2 \rho_-(\theta, X) \frac{\partial^2}{\partial^2 \theta} \rho_-(\theta, X) + 2 \bigg( \frac{\partial}{\partial \theta}  \rho_-(\theta, X) \bigg)^2  - 2 \frac{\partial^2}{\partial^2 \theta} \rho_-(\theta, X) \bigg] .
\end{align*} 
For the kernel Stein discrepancy, given that $k(x_1, x_2) - k(x_1, x_2^-)$ and $k(x_1, x_2) - k(x_1^-, x_2)$ are $1$ when $x = x'$ and otherwise $-1$, we simplify the expression as
\begin{align*}
    \text{KSD}(p_\theta \| p) & \overset{\theta}{=} \E_{X_1, X_2 \sim p} \big[ \left( 1 - \rho_-(\theta, X_1) \right) k(X_1, X_2) \left( 1 - \rho_-(\theta, X_2) \right) \big] \\
    & \overset{\theta}{=} \E_{X \sim p} \big[  \left( 1 - \rho_-(\theta, X) \right)^2 \big] \overset{\theta}{=} \E_{X \sim p} \big[  \left( \rho_-(\theta, X) \right)^2 - 2  \rho_-(\theta, X) \big]  .
\end{align*}
Suppose that the population loss minimiser is $\theta_* = 0.5$, meaning that the data-generating distribution $p$ is the Bernoulli model with $\theta_* = 0.5$.
We then have $\rho_-(\theta_*, x) = 1$, $\frac{\partial}{\partial \theta} \rho_-(\theta_*, x) = 2^2 ( 1 - 2 x )$, $\frac{\partial^2}{\partial^2 \theta} \rho_-(\theta_*, x) = 2^4 ( 1 - 2 x )^2$, and $\frac{\partial^2}{\partial^2 \theta} \rho_+(\theta_*, x) = - 2^4 ( 1 - 2 x  )^2$.
These gives us that
\begin{align*}
    ( \partial^2 / \partial^2 \theta ) \text{NLL}(p_\theta \| p) |_{\theta = \theta_*} & = \E_{X \sim p}\left[ 2^2 \times (X + 1 - X) \right] = 4, \\
    ( \partial^2 / \partial^2 \theta ) \text{DFD}(p_\theta \| p)  |_{\theta = \theta_*} & = \E_{X \sim p}\left[ 3 \times 2^5 \times (1 - 2 X)^4 \right] = 96, \\
    ( \partial^2 / \partial^2 \theta ) \text{KSD}(p_\theta \| p)  |_{\theta = \theta_*} & = \E_{X \sim p} \left[ 2 \times 2^4 \times ( 1 - 2 X )^2 \right] = 32 .
\end{align*}
By taking the inverse, the asymptotic variance $\sigma^2$ for the standard Bayes, the DFD-Bayes, and the KSD-Bayes is each given by $1 / 4$, $1 / 96$, and $1 / 32$.
In this example, the above calculation suggests that the DFD-Bayes has the narrowest credible region.
The difference in these values emphasises the importance of calibrating $\beta$, which we do for all of our experiments in the manuscript.

Finally, we empirically demonstrate the difference between the posteriors and the influence of $\beta$.
We computed each posterior in cases where (i) $\beta$ is \emph{not} calibrated i.e.~$\beta = 1$ and (ii) $\beta$ \emph{is} calibrated (except for the standard posterior, which has $\beta = 1$).
A Metropolis--Hastings algorithm was adopted to sample from all the posteriors.
A Gaussian random walk proposal with covariance $\sigma^2=0.01$ was used.
In total, 100 samples were obtained from each posterior by thinning 2,000 samples, after an initial burn-in of length 2,000.
\Cref{fig:figure_ber_pos} shows each posterior computed without $\beta$ calibrated.
It confirms that, without calibration of $\beta$, the DFD-Bayes posterior has the narrowest credible region, which agrees with the analytical illustration provided above.
\Cref{fig:figure_ber_pos_beta} shows each posterior computed with $\beta$ calibrated, where the result for the standard posterior is identical to \Cref{fig:figure_ber_pos} as $\beta = 1$.
For the DFD-Bayes and the KSD-Bayes, calibration of $\beta$ was performed by our proposal in \Cref{sec:calibration}, where we used 100 bootstrap minimisers to compute the analytical solution of $\beta_*$ in \eqref{eq:beta_optimal}.
It demonstrates that calibration of $\beta$ prevents over-concentration of the DFD-Bayes and the KSD-Bayes.

\newpage

\begin{figure}[h!]
    \hspace{-50pt}
    \includegraphics[width=1.2\textwidth]{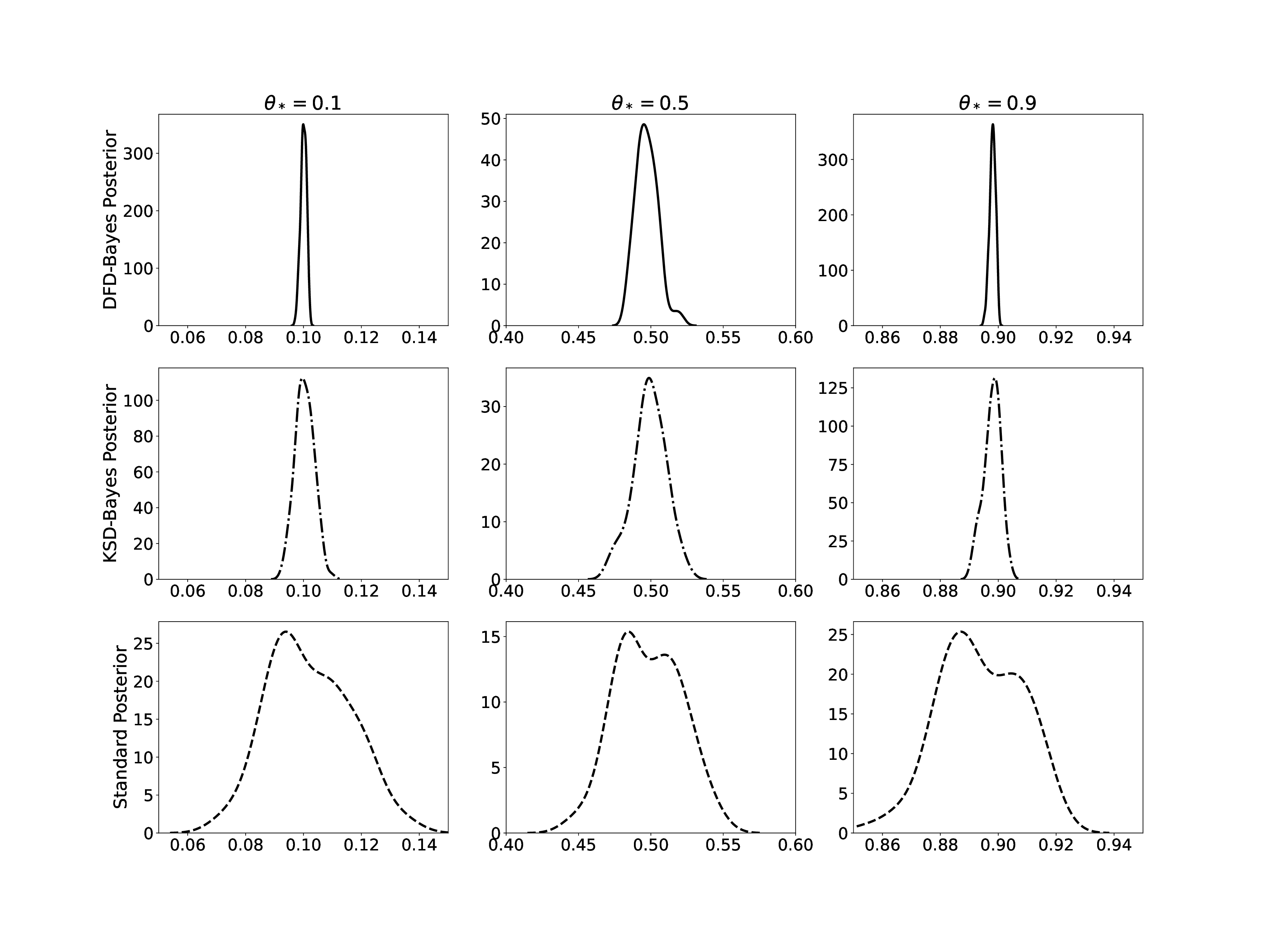}
    \caption{The DFD-Bayes posterior (top, solid), the KSD-Bayes posterior (middle, dash-dot), and the standard posterior (bottom, dash) computed without $\beta$ calibrated, for data generated from the Bernoulli model with three different parameters $\theta_* = 0.1$ (left), $\theta_* = 0.5$ (centre), and $\theta_* = 0.9$ (right). While their scales and geometries are different, all methods identify the correct parameter $\theta_*$.} \label{fig:figure_ber_pos}
\end{figure}

\newpage

\begin{figure}[h!]
    \hspace{-50pt}
    \includegraphics[width=1.2\textwidth]{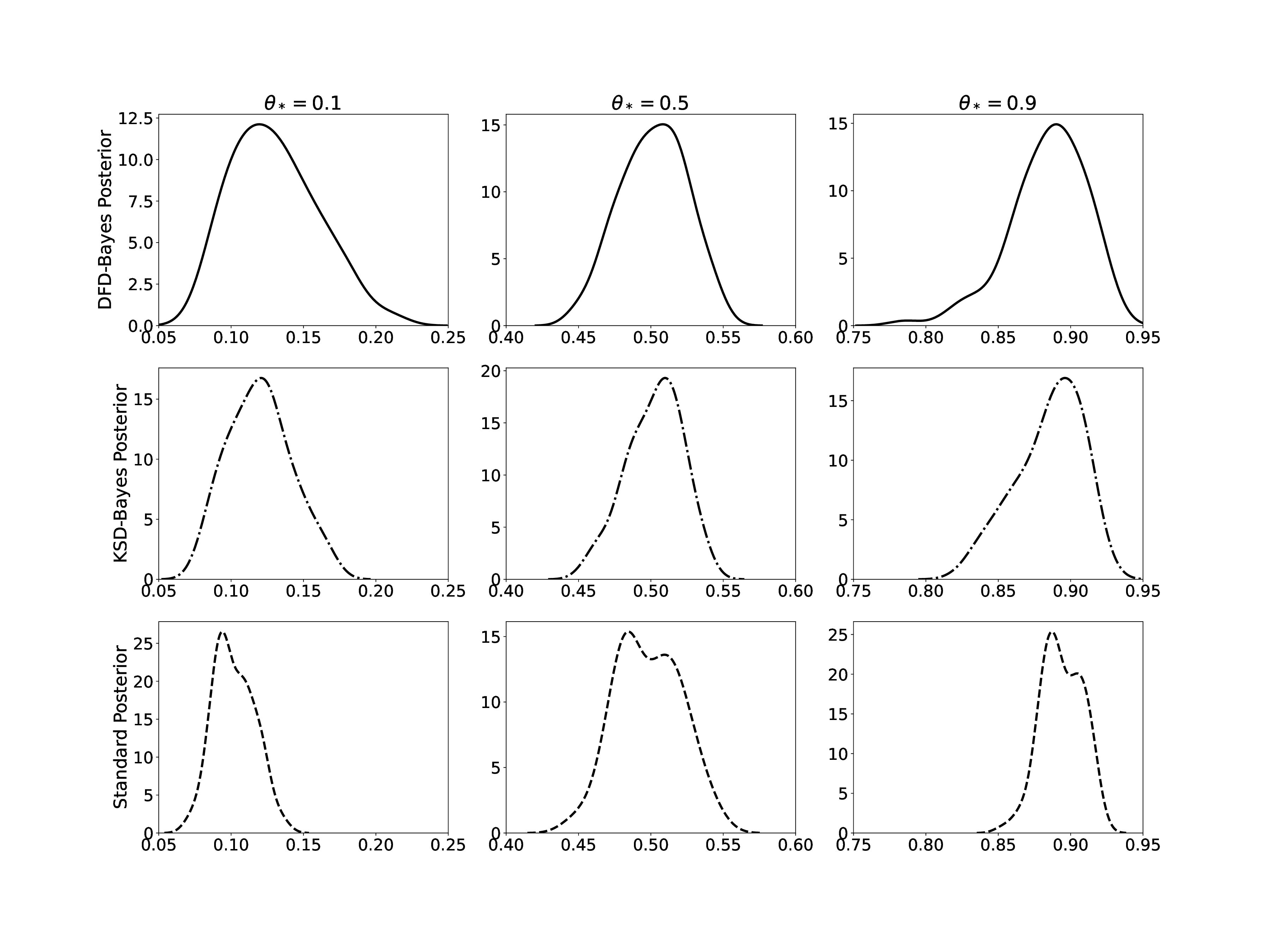}
    \caption{The DFD-Bayes posterior (top, solid), the KSD-Bayes posterior (middle, dash-dot), and the standard posterior (bottom, dash) computed with $\beta$ calibrated, for data generated from the Bernoulli model with three different parameters $\theta_* = 0.1$ (left), $\theta_* = 0.5$ (centre), and $\theta_* = 0.9$ (right). While their scales and geometries are different, all methods identify the correct parameter $\theta_*$.} \label{fig:figure_ber_pos_beta}
\end{figure}

\newpage

\subsection{Influence of Model Misspecification}

Next we turn our attention to the influence of model misspecification on each method.
It is convenient to consider the Poisson model to introduce a synthetic model misspecification.
For $x \in \N_0$, the Poisson model is 
\begin{align}
    p_\theta(x) = \frac{\theta^x \exp(\theta)}{x!} .
\end{align}
Then, the negative log-likelihood and the discrete Fisher divergence are 
\begin{align}
    \text{NLL}(p_\theta \| p_n) & \overset{\theta}{=} \frac{1}{n} \sum_{i=1}^{n} - x_i \log(\theta) + \theta , \label{eq:nll_poi} \\
    \text{DFD}(p_\theta \| p_n) & \overset{\theta}{=} \frac{1}{n} \sum_{i=1}^{n} \bigg( \frac{x_i}{\theta} \bigg)^2 - 2 \frac{x_i + 1}{\theta} . \label{eq:dfd_poi}
\end{align}
Letting $\rho_-(\theta, x) := p_\theta(x^-) / p_\theta(x) = x_i / \theta$, the kernel Stein discrepancy is
\begin{multline}
\text{KSD}(p_\theta \| p) \overset{\theta}{=} \frac{1}{n^2} \sum_{i=1}^{n} \sum_{j=1}^{n} \left( 1 - \rho_-(\theta, x_i) \right) k(x_i, x_j) \left( 1 - \rho_-(\theta, x_j) \right) + \\
\left( 1 - \rho_-(\theta, x_i) \right) \left( k(x_i, x_j) - k(x_i, x_j^-) \right) + \left( k(x_i, x_j) - k(x_i^-, x_j) \right) \left( 1 - \rho_-(\theta, x_j) \right) . \label{eq:ksd_poi}
\end{multline}
For the kernel Stein discrepancy, we use a similar choice of kernel to \cite{Matsubara2021}, that induces a robustness suitable for this example: $k(x, x') = m(x) \exp( - \mathbbm{1}_{x = x'} ) m(x')$ where $m(x) = \sigma( 15 - x )$ based on a sigmoid function $\sigma(t) = ( 1 + \exp( - t ) )^{-1}$.

For illustration, we synthetically introduce model misspecification by mixing outliers into the data.
We sampled 500 data points $\{ x_i \}_{i=1}^{n}$ from the Poisson model with the parameter $\theta_* = 5$, and replaced the $100 \times \epsilon$ percent of data with an outlier $y = 20$ that is larger than the $99.9$\% percentile of the Poisson distribution of $\theta_* = 5$.
This causes a synthetic model misspecification because the dataset is generated from a mixture of the Possion model and the Dirac distribution at $y = 20$, which cannot be adequately explained by only the Poisson model.
The sensitivity of each posterior to the outlier can be analytically investigated.
The standard Bayesian posterior is modestly impacted by the outlier $y$, given that the negative log-likelihood \eqref{eq:nll_poi} is a linear function of each datum $x_i$.
On the other hand, in this example, DFD-Bayes may be more severely impacted, given the discrete Fisher divergence \eqref{eq:dfd_poi} is a quadratic function of each datum $x_i$.
The growth rate of the kernel Stein discrepancy with respect to each datum $x_i$ is determined by the choice of kernel $k$.
We compute each posterior for two cases when $\epsilon = 0.0$ (no outlier contained) and $\epsilon = 0.1$ (10\% outliers contained), to empirically demonstrate the impact of the model misspecification.
The Metropolis--Hastings algorithm with the Gaussian random walk proposal of $\sigma^2 = 0.1$ is used to sample from each posterior with calibration applied.
In total, 100 samples were obtained from each posterior by thinning 2,000 samples, after an initial burn-in of length 2,000.

\begin{figure}[t]
    \hspace{-55pt}
    \includegraphics[width=1.2\textwidth]{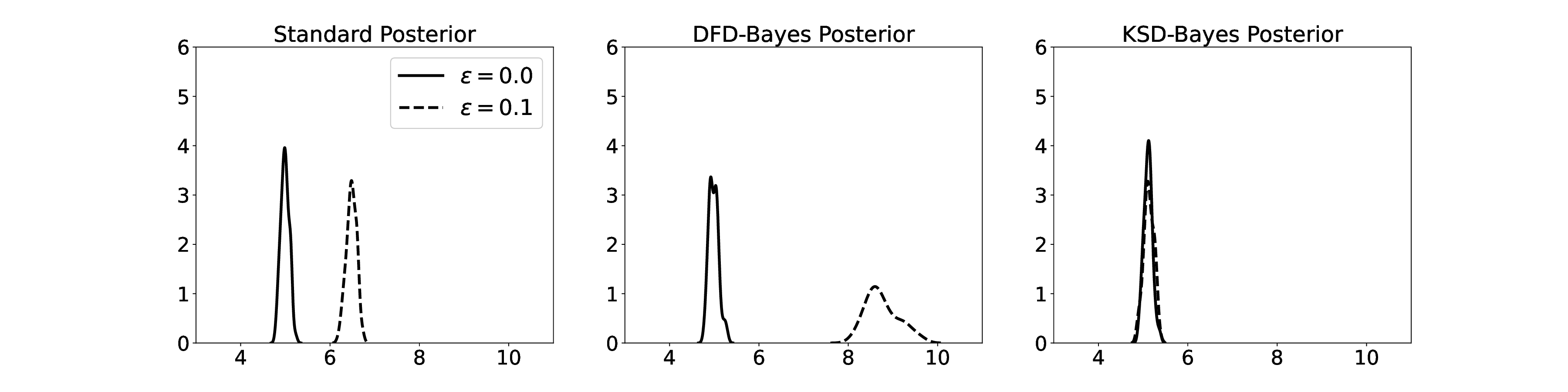}
    \caption{The standard posterior (left), The DFD-Bayes posterior (centre), and the KSD-Bayes posterior (right) computed with $\beta$ calibrated for data when $\epsilon = 0.0$ (solid line) and $\epsilon = 0.1$ (dash line), that is, the $10\%$ of data is replaced with outlier $y$.} \label{fig:figure_bin_pos}
\end{figure}

\Cref{fig:figure_bin_pos} demonstrates the sensitivity of the standard Bayesian posterior and DFD-Bayes to the outliers, whlie KSD-Bayes shows insensitivity due to the careful choice of kernel.
See also \Cref{subsec: KSD_robustness} for more discussion on robustness of KSD-Bayes.
In this example, the sensitivity of the DFD-Bayes to the outlier was higher than the standard Bayesian posterior, as anticipated.
\cite{Barp2019} proposed a robust analogue of the Fisher divergence in the continuous case. 
Although this is not a focus of this work, a similar approach may be applied to the discrete case when severe model misspecification is anticipated.
This would be an interesting avenue for further work, but our present interest is in computation for discrete intractable likelihood.

\subsection{Limitation of DFD-Bayes for Inference of Mixture Parameters} \label{subsec:limitation_poi_mix}

Finally, we provide an empirical illustration of the  limitation of score-based methods in \Cref{sec:limitation}.
It has been pointed out that score-based methods generally exhibit insensitivity to mixing proportions when mixture components have isolated high-probability regions \citep{wenliang2020blindness,zhang2022towards}.
In the continuous case, this can be observed using a mixture model of two Gaussian distributions $\P_\theta(x) = (1 - \theta) \times \mathcal{N}(-\mu, 1) + \theta \times \mathcal{N}(\mu, 1)$ whose parameter is the mixture ratio.
\cite{zhang2022towards} illustrated how the Fisher divergence is approximately constant over $\Theta$ if $\mu$ is large enough to isolate the components $\mathcal{N}(-\mu, 1)$ and $\mathcal{N}(\mu, 1)$.
We illustrate the same limitation for the discrete Fisher divergence using a mixture model of two Poisson distributions $p_\theta(x) = (1 - \theta) \times q_{\lambda_1}(x) + \theta \times q_{\lambda_2}(x)$, where $q_{\lambda_1}$ and $q_{\lambda_2}$ are the Poisson distributions with rate parameters $\lambda_1 > 0$ and $\lambda_2 > 0$.
\Cref{fig:limitation_poi_mix} shows the geometry of the discrete Fisher divergence between the mixture model $p_\theta$ and data generated from the mixture model $p_{\theta_*}$ with the true mixture proportion $\theta_*$, for two cases when the supports of the two Poisson distributions are highly isolated and when they are not isolated.
The correct mixture proportion $\theta_*$ was identified only in the latter case, while in the former case the discrete Fisher divergence was approximately constant.
See \cite{zhang2022towards} for a potential approach to remedy this general limitation of score-based methods.

\begin{figure}[t]
    \hspace{-50pt}
    \includegraphics[width=1.2\textwidth]{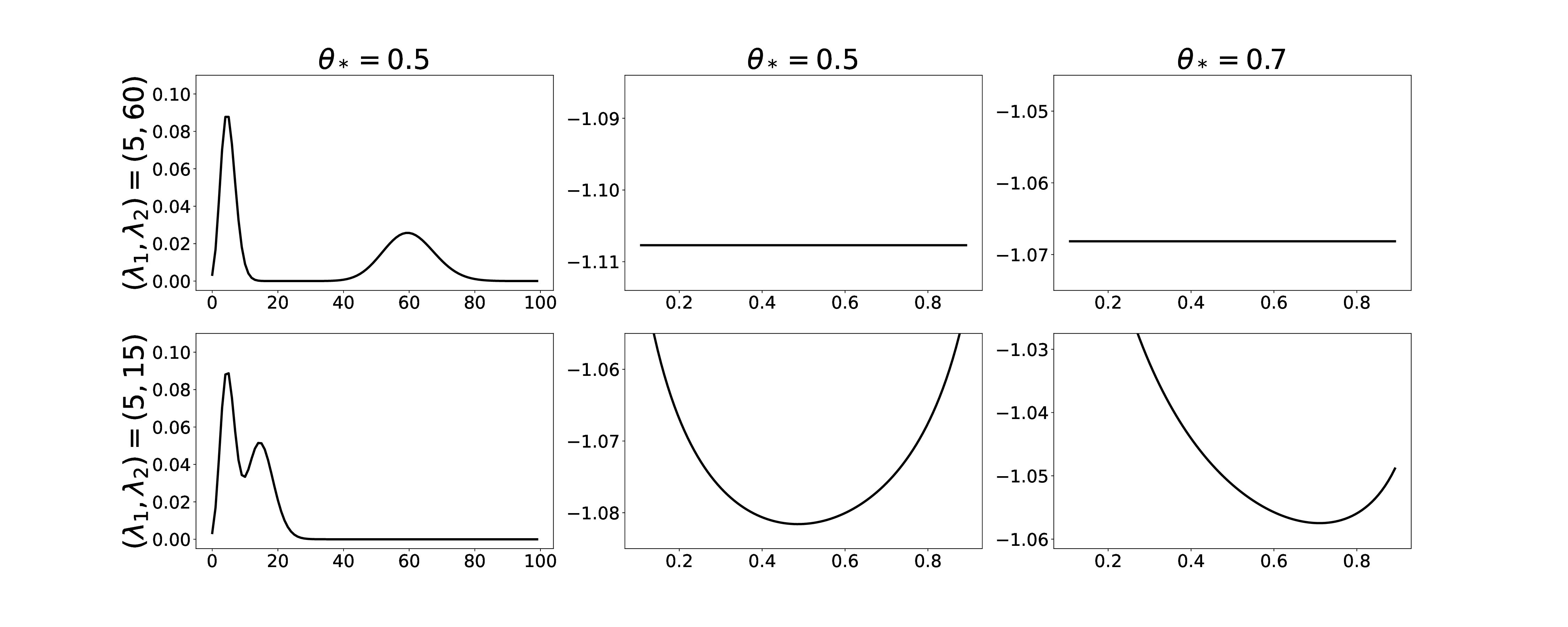}
    \caption{The form of the Poisson mixture model $p_{\theta_*}$ when $\theta_* = 0.5$ (left), the discrete Fisher divergence computed for data generated from the model $p_{\theta_*}$ with $\theta_* = 0.5$ (middle), and the discrete Fisher divergence computed for data generated from the model $p_{\theta_*}$ with $\theta_* = 0.7$ (right), for two cases where $\lambda_1 = 5, \lambda_2 = 60$ (top) and $\lambda_1 = 5, \lambda_2 = 15$ (bottom).}
    \label{fig:limitation_poi_mix}
\end{figure}

\section{Proofs of Theoretical Results}\label{app: proofs}

This section contains the proof of all theoretical results in the paper, including \Cref{prop: dif_sd_1}, \Cref{thm:bvm} and \Cref{thm:beta_choice}.

\subsection{Proof of \Cref{prop: dif_sd_1}} \label{apx: proof_dif_sd_1}

First we introduce three technical lemmas that will be useful:

\begin{lemma} \label{prop:invertible}
	For any $\bm{x} \in \X$ and $i = 1, \dots, d$, it holds that $( \bm{x}^{i-} )^{i+} = \bm{x}$ and $( \bm{x}^{i+} )^{i-} = \bm{x}$.
\end{lemma}

\begin{proof}
	Since $\X = S_1 \times \dots \times S_d$ from the Standing Assumption, 
	\begin{align}
		\bm{x}^{i-} = (x_1, \cdots, x_i^-, \cdots, x_d), \qquad \bm{x}^{i+} = (x_1, \cdots, x_i^+, \cdots, x_d) .
	\end{align}
	It is thus sufficient to show that $( x_i^- )^+ = x_i$ and $( x_i^+ )^- = x_i$ for any $i = 1, \dots, d$.
	Consider, therefore, a 
	set $S \cong I \subseteq \mathbb{Z}$ with more than one element.
	Our aim is to establish the identity $( s^- )^+ = s$ and $( s^+ )^- = s$ for all $s \in S$.
	Existence of the least and greatest element $s_{\min}$ and $s_{\max}$ of $S$ determines four qualitatively distinct cases to be checked: (i) neither of them exist; (ii) both of them exists; (iii) only $s_{\min}$ exists; (iv) only $s_{\max}$ exists.
	Recall that we identify the case (iv) with (iii) without loss of generality by reversing the ordering of $S$.
	The identity for (i) \& (ii) is trivial since the maps $s \mapsto s^-$ is bijective from $S$ to itself with inverse $s \mapsto s^+$.
	For case (iii), we have $( s^- )^+ = s$ for $s \ne s_{\min}$ and $( s^+ )^- = s$ for all $s \in S$,
	Recalling the definition $s_{\min}^- = \star$ and $\star^+ = s_{\min}$ completes the argument.
\end{proof}

\begin{lemma} \label{prop:sumbypart}
	For any $f, g: \X \to \R$ and any $i = 1, \dots, d$, suppose $\sum_{\bm{x} \in \X} | f(\bm{x}) g(\bm{x}^{i-}) | < \infty$, that is, the series is absolutely convergent.
	Then we have 
	\begin{align}
		\sum_{\bm{x} \in \X} f(\bm{x}) g(\bm{x}^{i-}) = \sum_{\bm{x} \in \X} f(\bm{x}^{i+}) g(\bm{x}) . \label{eq: various series}
	\end{align}
\end{lemma}

\begin{proof}
	Since $\X = S_1 \times \dots \times S_d$ from the Standing Assumption, the series can be expressed as
	\begin{align*}
		\sum_{\bm{x} \in \X} f(\bm{x}) g(\bm{x}^{i-}) & = \sum_{x_1 \in S_1} \cdots \sum_{x_i \in S_i} \cdots \sum_{x_d \in S_d} f(x_1, \cdots, x_i, \cdots, x_d) g(x_1, \cdots, x_i^-, \cdots, x_d) , \\ 
		\sum_{\bm{x} \in \X} f(\bm{x}^{i+}) g(\bm{x})	& = \sum_{x_1 \in S_1} \cdots \sum_{x_i \in S_i} \cdots \sum_{x_d \in S_d} f(x_1, \cdots, x_i^{+}, \cdots, x_d) g(x_1, \cdots, x_i, \cdots, x_d) .
	\end{align*}
	Holding the coordinates $x_1,\dots,x_{i-1},x_{i+1},\dots,x_d$ fixed, and exploiting absolute convergence to justify the interchange of summations, the claimed result follows if
	\begin{align}
		\sum_{x_i \in S_i} \tilde{f}(x_i) \tilde{g}(x_i^{-}) = \sum_{x_i \in S_i} \tilde{f}(x_i^{+}) \tilde{g}(x_i) \label{eq:identitysi}
	\end{align}
	where $\tilde{f}(x_i) := f(x_1, \cdots, x_i, \cdots, x_d)$ and $\tilde{g}(x_i) := g(x_1, \cdots, x_i, \cdots, x_d)$ are viewed as functions on $S_i$.
	
	Consider, therefore, an arbitrary set $S \cong I \subseteq \mathbb{Z}$, for which we aim to establish the identity $\sum_{s \in S} h(s) k(s^-) = \sum_{s \in S} h(s^+) k(s)$ for any functions $h,k: S \to \R$ s.t. $\sum_{s \in S} | h(s) k(s^-) | < \infty$.
	From the definition of an order isomorphism, the elements of $S$ can be indexed as $S = \{s_i : i \in I\}$, where $s_i < s_j$ if and only if $i < j$.
	The identity therefore can be written as $\sum_{i \in I} h(s_i) k(s_i^-) = \sum_{i \in I} h(s_i^+) k(s_i)$, and will be verified for the three qualitatively distinct cases of index set $I$ described in the proof of \Cref{prop:invertible}:
	\begin{enumerate}
		\item[(i)] $I = \mathbb{Z}$. The result is immediate, since $(s_i, s_i^-) = (s_i,s_{i-1})$ and $(s_i^+, s_i) = (s_{i+1},s_i)$ range over the same set for $i \in I$. The series $\sum_{i \in I} h(s_i^+) k(s_i)$ is absolutely convergent since the sets $\{ h(s_i) k(s_i^-) \}_{i \in I}$ and $\{ h(s_i^+) k(s_i) \}_{i \in I}$ in the two series are equal.
		\item[(ii)] $I = \{1,\dots,n\}$ for some $n \in \mathbb{N}$.  In this case $s_{\min} = s_1$ and $s_{\max} = s_n$, and it follows from the definition of decrements and increments that
		\begin{align*}
			\sum_{i \in I} h(s_i) k(s_i^-) & = h(s_1) k(s_1^-) + h(s_2) k(s_1) + \cdots + h(s_n) k(s_{n-1}) \\
			& = h(s_n^+) k(s_n) + h(s_2) k(s_1) + \cdots + h(s_n) k(s_{n-1})  = \sum_{i \in I} h(s_i^+) k(s_i) ,
		\end{align*}
		where the sets $\{ h(s_i) k(s_i^-) \}_{i \in I}$ and $\{ h(s_i^+) k(s_i) \}_{i \in I}$ are again equal.
		\item[(iii)] $I = \{1,2,\dots\}$.  In this case $s_{\min} = s_1$, and it follows from the definition $s_1^- = \star$ and $k(\star) = 0$ that
		\begin{align*}
			\sum_{i \in I} h(s_i) k(s_i^-) & = \underbrace{ h(s_1) k(\star)}_{=0} + h(s_2) k(s_1) + h(s_3) k(s_2) + \cdots \\
			& = h(s_2) k(s_1) + h(s_3) k(s_2) + \cdots = \sum_{i \in I} h(s_i^+) k(s_i) .
		\end{align*}
		The series $\sum_{i \in I} h(s_i^+) k(s_i)$ is absolutely convergent since the set $\{ h(s_i^+) k(s_i) \}_{i \in I}$ is a subset of the absolutely summable set $\{ h(s_i) k(s_i^-) \}_{i \in I}$.
	\end{enumerate}
	This completes the proof.
\end{proof}

Let $F(\X,S)$ denote the set of all functions $f$ of the form $f : \X \to S$.

\begin{lemma} \label{prop: injection}
    For probability mass function $p:\X \rightarrow (0,\infty)$, the map $\mu_p := (\nabla^- p) / p$ is an injection $\mu : F(\X,(0,\infty)) \to F(\X, \R^d)$.
\end{lemma}
\begin{proof}
    It suffices to show that each value $p(\bm{x})$, for $\bm{x} \in \X$, can be explicitly recovered from $\mu_p$.
    Note that, since $p$ takes values in $(0,\infty)$, the embedding $\mu_p$ is well-defined.
    From the Standing Assumption, we have that $\X = S_1 \times \dots \times S_d$, where each $S_i \cong I_i \subseteq \mathbb{Z}$ is a set with more than one element.
    Since the $S_i$ serve only as index sets, we can without loss of generality assume that $S_i$ is a consecutive subset of $\mathbb{Z}$ and that $0 \in S_i$, for each $i = 1,\dots,d$.
    The idea of the proof is to demonstrate that each of the quantities $p(\bm{x})$ can be explicitly expressed in terms of $\mu_p$, $p(\bm{0})$ and $\{ p(\bm{y}) : \|\bm{y}\|_1 < \|\bm{x}\|_1 \}$, where $\|\bm{x}\|_1 := |x_1| + \dots + |x_d|$.
    It would then follow from a simple inductive argument that $p(\bm{x})$ can be expressed in terms of $\mu_p$ and $p(\bm{0})$.
    Finally, the constraint that $\sum_{\bm{x} \in \X} p(\bm{x}) = 1$ uniquely determines $p(\bm{0})$, demonstrating that $p(\bm{x})$ can be explicitly recovered.
    
    Given $\bm{x} \in \X$, assume $\bm{x} \neq \bm{0}$, for otherwise the claim will trivially hold.
    Then let $i \in \{1,\dots,d\}$ be such that $x_i \neq 0$.
    If $x_i > 0$, then from the definition of $\mu_p(\bm{x})_i = 1 - p(\bm{x}^{i-}) / p(\bm{x})$ we have the relation 
    \begin{align*}
        p(\bm{x}) = \frac{ p(\bm{x}^{i-}) }{ 1 - \mu_p(\bm{x})_i }
    \end{align*} 
    where $\|\bm{x}^{i-}\|_1 = \|\bm{x}\|_1 - 1$.
    Conversely, if $x_i < 0$, then using \Cref{prop:invertible} we have $\mu_p(\bm{x}^{i+})_i = 1 - p(\bm{x}) / p(\bm{x}^{i+})$ and we have the relation 
    \begin{align*}
        p(\bm{x}) = [ 1 - \mu_p(\bm{x}^{i+})_i ] p(\bm{x}^{i+})
    \end{align*} 
    where $\|\bm{x}^{i+}\|_1 = \|\bm{x}\|_1 - 1$.
    The previously described inductive argument completes the proof.
\end{proof}

Now we prove the main result:

\begin{proof}[Proof of \Cref{prop: dif_sd_1}]
    Expanding the square gives that
    \begin{multline*}
	\operatorname{DFD}(p \| q) = \E_{X \sim q}\bigg[ \sum_{i=1}^{d} \bigg( \frac{p(X) - p(X^{j-})}{p(X)} \bigg)^2 \\
        - 2 \underbrace{ \frac{p(X) - p(X^{j-})}{p(X)} \frac{q(X) - q(X^{j-})}{q(X)} }_{=:(*)} + \bigg( \frac{q(X) - q(X^{j-})}{q(X)} \bigg)^2 \bigg] .
    \end{multline*}
    Denote by $\text{supp}(q)$ the support of $q$ i.e.~$\{ \bm{x} \in \X \mid q(\bm{x}) \ne 0 \}$.
    For the term $\E_{X \sim q}[ (*) ]$, it follows from the definition $\E_{X \sim q}[ f(X) ] = \sum_{\bm{x} \in \text{supp}(q)} f(\bm{x}) q(\bm{x})$ that
    \begin{align*}
        \E_{X \sim q}[ (*) ] & = \sum_{j=1}^{d} \E_{X \sim q}\bigg[ \frac{p(X) - p(X^{j-})}{p(X)} - \frac{p(X) - p(X^{j-})}{p(X)} \frac{q(X^{j-})}{q(X)} \bigg] \\
        & = \sum_{j=1}^{d} \bigg\{ \sum_{\bm{x} \in \text{supp}(q)} \frac{p(\bm{x}) - p(\bm{x}^{j-})}{p(\bm{x})} q(\bm{x}) - \underbrace{ \sum_{\bm{x} \in \text{supp}(q)} \frac{p(\bm{x}) - p(\bm{x}^{j-})}{p(\bm{x})} q(\bm{x}^{j-}) }_{(**)} \bigg\} ,
    \end{align*}
    We apply \Cref{prop:sumbypart} to the term $(**)$ with $f(\bm{x}) = ( p(\bm{x}) - p(\bm{x}^{j-}) ) / p(\bm{x})$ and $g(\bm{x}) = q(\bm{x})$, where $f(\bm{x}^{j+})$ is well-defined for all $\bm{x} \in \text{supp}(q)$ due to the assumption that $p(\bm{x}^{j+}) > 0$ for $\bm{x} \in \text{supp}(q)$ and \Cref{prop:sumbypart} is thus applicable.
    This reveals that 
    \begin{align*}
        (**) = \sum_{\bm{x} \in \text{supp}(q)} \frac{p(\bm{x}^{j+}) - p(\bm{x})}{p(\bm{x}^{j+})} q(\bm{x}) 
    \end{align*}
    for each $j = 1, \dots, d$, where \Cref{prop:invertible} is used to deduce that $(\bm{x}^{j-})^{j+} = \bm{x}$.
    Hence, we have
    \begin{align*}
        \E_{X \sim q}[ (*) ] & = \E_{X \sim q}\bigg[ \sum_{i=1}^{d} \frac{p(X) - p(X^{j-})}{p(X)} - \frac{p(X^{j+}) - p(X)}{p(X^{j+})} \bigg] = \E_{X \sim q}\bigg[ \sum_{i=1}^{d} - \frac{p(X^{j-})}{p(X)} + \frac{p(X)}{p(X^{j+})} \bigg] .
    \end{align*}
    Plugging this equality in the discrete Fisher divergence at the top and completing the expansion establish that
    \begin{align*}
	\operatorname{DFD}(p \| q) & = \E_{X \sim q}\bigg[ \sum_{i=1}^{d} \bigg( 1 - \frac{p(X^{j-})}{p(X)} \bigg)^2 + 2 \frac{p(X^{j-})}{p(X)} - 2 \frac{p(X)}{p(X^{j+})} + \bigg( 1 \frac{q(X^{j-})}{q(X)} \bigg)^2 \bigg] \\
        & = \E_{X \sim q}\bigg[ \sum_{i=1}^{d} \bigg( \frac{p(X^{j-})}{p(X)} \bigg)^2 - 2 \frac{p(X)}{p(X^{j+})} \bigg] + \underbrace{ \E_{X \sim q}\bigg[ \sum_{i=1}^{d} 1 + \bigg( 1 - \frac{q(X^{j-})}{q(X)} \bigg)^2 \bigg] }_{=: C(q)} .
    \end{align*}
    
    Finally we verify that $\operatorname{DFD}(p \| q) = 0$ if and only if $p = q$.
    From \Cref{prop: injection} we have the injective embedding $p \mapsto \mu_p := (\nabla^- p) / p$ of a positive density $p : \X \rightarrow (0,\infty)$ into $F(\X, \R^d)$.
    Since $q > 0$, the map $p \mapsto \mu_p$ is also an injection into $L^2(q,\R^d)$, equipped with the canonical norm $\|\nu\|_{L^2(q,\R^d)} := \E_{X \sim q}[\|\nu(X)\|^2],~\forall \nu \in L^2(q,\R^d)$. 
    From \eqref{eq: DFD def} we recognise that $\operatorname{DFD}(p \| q) = \| \mu_p - \mu_q \|_{L^2(q,\R^d)}^2 $ is the squared distance between $\mu_p$ and $\mu_q$ according to the canonical norm of $L^2(q,\R^d)$.
    Since $\| \mu_p - \mu_q \|_{L^2(q,\R^d)} = 0$ if and only if $\mu_p = \mu_q$ in $L^2(q,\R^d)$, it follows from injectivity of $p \mapsto \mu_p$ that $\operatorname{DFD}(p \| q) = 0$ if and only if $p = q$, as required.
\end{proof}

\subsection{Proof of \Cref{thm:bvm}} \label{apx: proof_bvm}

This appendix contains the proof of \Cref{thm:bvm}.
\cite{Miller2019} provided sufficient conditions for consistency and asymptotic normality of generalised Bayesian posteriors of the form $\pi_n^D(\mathrm{d}\theta) \propto \exp( - n D_n(\theta)) \pi(\mathrm{d}\theta)$, where $D_n : \Theta \rightarrow \mathbb{R}$ is a loss function that may depend on the data $\{ \bm{x}_i \}_{i=1}^{n}$. 
These results can be leveraged to analyse DFD-Bayes, by setting
\begin{align}
    D_n(\theta) \stackrel{\theta}{=}  \frac{\beta}{n} \sum_{i=1}^{n} \sum_{i=1}^{d} \bigg( \frac{ p_{\theta}(\bm{x}_i^{j-})}{p_{\theta}(\bm{x}_i)} \bigg)^2 - 2 \bigg( \frac{ p_{\theta}(\bm{x}_i) }{ p_{\theta}(\bm{x}_i^{j+}) } \bigg) . \label{eq: our fn}
\end{align}
These conditions were refined into more applicable forms in \cite{Matsubara2021}.
While \cite{Matsubara2021} focused on their particular case of losses based on kernelised Stein discrepancies, their argument can be directly applied for essentially any arbitrary loss $D_n$.
We repeat this argument by modifying it so that it can be applied for any loss $D_n$.
Let $B_\epsilon(\theta_*) := \{ \theta \in \Theta : \| \theta - \theta_* \| < \epsilon \}$.
\begin{theorem} \label{asmp:general_bvm}
	Let $\Theta \subseteq \R^v$ be Borel.
    Let $D : \Theta \rightarrow \R$ be a fixed measurable function and $\{ D_n \}_{n=1}^{\infty}$ be a sequence s.t.~$D_n : \Theta \rightarrow \R$ is a measurable function dependent on random data $\{ \bm{x}_i \}_{i=1}^{n} \subset \X$.
    Let $H_n(\theta) := \nabla_\theta^2 D_n(\theta)$.
	Suppose that, for some bounded convex open set $U \subseteq \Theta$, the following hold:
	\begin{enumerate}
		\item[C1] $D_n$ a.s. converges pointwise to $D$;
		\item[C2] $D_n$ is $r$ times continuously differentiable in $U$ and $\limsup_{n \to \infty} \sup_{\theta \in U} \| \nabla_\theta^r D_n(\theta) \| < \infty$ a.s.~for $r = 1, 2, 3$;
		\item[C3] for all $n$ sufficiently large, $\theta_n \in U$ for any $\theta_n \in \argmin D_n$ a.s., and a point $\theta_* \in U$ uniquely attains $D(\theta_*) = \inf_{\theta \in \Theta} D(\theta)$. 
		\item[C4] $H_n(\theta_*) \overset{a.s.}{\rightarrow} H_*$ for some nonsingular $H_*$;
		\item[C5] $\pi$ is continuous and positive at $\theta_*$.
	\end{enumerate}
	Then, for any $\epsilon > 0$, the generalised posterior $\pi_n^D(\mathrm{d}\theta) \propto \exp( - n D_n(\theta)) \pi(\mathrm{d}\theta)$ satisfies
	\begin{align*}
		\int_{B_\epsilon(\theta_*) } \pi_n^D(\theta) \; \mathrm{d} \theta \as 1 .
	\end{align*}
	Let $( \theta_n )_{n=1}^{\infty} \subset \Theta$ be a sequence s.t.~$\theta_n$ minimises $D_n$ for all $n$ sufficiently large. 
	Denote by $\widetilde{\pi}_n^D$ a density on $\R^v$ of the random variable $\sqrt{n} (\theta - \theta_n)$, where $\theta \sim \pi_n^D$.
	Then
	\begin{align*}
		\int_{\R^d} \left| \widetilde{\pi}_n^D(\theta) - \frac{1}{Z_*} \exp\left( - \frac{1}{2} \theta \cdot H_* \theta \right) \right| \mathrm{d} \theta  \as 0 
	\end{align*}
	where $Z_*$ is the normalising constant of $\exp(-\frac{1}{2} \theta \cdot H_* \theta )$.
\end{theorem}

\noindent
The proof of \Cref{asmp:general_bvm} is deferred to \Cref{sec: proof_con_bvm}.
The main proof of \Cref{thm:bvm} aims to show that the preconditions C1-C5 of \Cref{asmp:general_bvm} are satisfied for the particular function $D_n$ in \eqref{eq: our fn}, defining the DFD-Bayes generalised posterior.

\begin{proof}[Proof of \Cref{thm:bvm}]
	Without loss of generality, we will give the proof for $\beta = 1$ for notational convenience.\footnote{To extend the proof to arbitrary $\beta > 0$, simply replace $D_n(\theta) = \operatorname{DFD}(p_\theta \| p_n)$ in all arguments by $D_n(\theta) = \beta \operatorname{DFD}(p_\theta \| p_n)$. All the arguments hold immediately since $\beta$ is a constant.}
	Let $r_{j-}(\bm{x}, \theta) := p_{\theta}( \bm{x}^{j-} ) / p_{\theta}( \bm{x} )$ and $r_{j+}(\bm{x}, \theta) := p_{\theta}( \bm{x} ) / p_{\theta}( \bm{x}^{j+} )$ for each $j = 1, \dots, d$.
	We can write $D_n$ as
	\begin{align*}
		D_n(\theta) & \stackrel{\theta}{=} \frac{1}{n} \sum_{i=1}^{n} \underbrace{ \sum_{j=1}^{d} \big( r_{j-}(\bm{x}_i, \theta) \big)^2 - 2 r_{j+}(\bm{x}_i, \theta) }_{=: R(\bm{x}_i,\theta)} .
	\end{align*}
	In what follows we set $D(\theta) := \E_{X \sim p}[ R(X, \theta) ]$ and verify that preconditions C1-C5 of \Cref{thm:bvm} are satisfied.
	Note that C3 holds directly by \Cref{asmp:minimiser} and C5 is also assumed directly in \Cref{thm:bvm}.
	
	\vspace{5pt}
	\noindent
	\textbf{C1:} By the strong law of large numbers \citep[Theorem~2.5.10]{Durrett2010b},
	\begin{align}
		D_n(\theta) & = \frac{1}{n} \sum_{i=1}^{n} R(\bm{x}_i, \theta) \quad \as \quad \E_{X \sim p}[ R(X, \theta) ] = D(\theta), \label{eq:pwcon}
	\end{align}
	provided that $\E_{X \sim p}[ | R(X, \theta) | ] < \infty$ for each $\theta \in \Theta$.
	Thus we must check that $\E_{X \sim p}[ | R(X, \theta) | ] < \infty$.
	By the triangle inequality, 
	\begin{align*}
		\E_{X \sim p}[ | R(X, \theta) | ] & = \E_{X \sim p}\left[ \left| R(X, \theta) \right| \right] + C(p) - C(p) \\
		& = \E_{X \sim p}\left[ \left| R(X, \theta) + 1 + \left\| \frac{ \nabla^- p( X ) }{ p( X ) } \right\|^2 \right| \right] + \E_{X \sim p}\left[ 1 + \left\| \frac{ \nabla^- p( X ) }{ p( X ) } \right\|^2 \right] \\
		& = \E_{X \sim p}\left[ \left\| \frac{ \nabla^- p_{\theta}(X)}{p_{\theta}(X)} - \frac{ \nabla^- p( X ) }{ p( X ) } \right\|^2 \right] + 1 + \E_{X \sim p}\left[ \left\| \frac{ \nabla^- p( X ) }{ p( X ) } \right\|^2 \right]
	\end{align*}
	where the last equality holds from \Cref{prop: dif_sd_1} and both the quantities are finite by Standing Assumption 1.
	Hence \eqref{eq:pwcon} holds for every $\theta \in \Theta$.
	
	\vspace{5pt}
	\noindent
	\textbf{C2:} From \Cref{asmp:derivative}, we have that $r_{j+}(\bm{x}, \theta)$ and $r_{j-}(\bm{x}, \theta)$ are three times continuously differentiable with respect to $\theta \in U$ for all $\bm{x} \in \X$, and thus $D_n(\theta)$ is three times continuously differentiable with respect to $\theta \in U$.
	For any $s \in \{ 1, 2, 3 \}$, 
	\begin{align}
		\nabla_\theta^s D_n(\theta) & = \frac{1}{n} \sum_{i=1}^{n} \nabla_\theta^s R(\bm{x}_i, \theta) = \frac{1}{n} \sum_{i=1}^{n} \sum_{j=1}^{d} \nabla_\theta^s ( r_{j-}(\bm{x}_i, \theta)^2 ) - 2 \nabla_\theta^s r_{j+}(\bm{x}_i, \theta) . \label{eq:hsd_derivative}
	\end{align}
	By the triangle inequality, we have an upper bound
	\begin{align*}
		\sup_{\theta \in U} \| \nabla_\theta^s D_n(\theta) \| \le \frac{1}{n} \sum_{i=1}^{n} \underbrace{ \sum_{j=1}^{d} \sup_{\theta \in U} \left\| \nabla_\theta^s ( r_{j-}(\bm{x}_i, \theta)^2 ) \right\| + 2 \sup_{\theta \in U} \left\| \nabla_\theta^s r_{j+}(\bm{x}_i, \theta) \right\| }_{=: G(\bm{x}_i)} .
	\end{align*}
	The quantity $\frac{1}{n} \sum_{i=1}^{n} G(\bm{x}_i)$ is a random variable dependent on $\{ \bm{x}_i \}_{i=1}^{n}$.
	By the strong law of large numbers \citep[Theorem~2.5.10]{Durrett2010b},
	\begin{align*}
	    \frac{1}{n} \sum_{i=1}^{n} G(\bm{x}_i) \as \E_{X \sim p}[ G(X) ] < \infty
	\end{align*}
	provided that $\E_{X \sim p}[ | G(X) | ] < \infty$. 
	Indeed, this condition holds since from positivity of $G$
	\begin{align*}
		\E_{X \sim p}[ |G(X) | ] = \sum_{j=1}^{d} \E_{X \sim p}\left[ \sup_{\theta \in U} \left\| \nabla_\theta^s ( r_{j-}(X, \theta)^2 ) \right\| \right] + 2 \E_{X \sim p}\left[ \sup_{\theta \in U} \left\| \nabla_\theta^s r_{j+}(X, \theta) \right\| \right] ,
	\end{align*}
	where the right hand side is finite by \Cref{asmp:derivative}.
	Then
	\begin{align*}
		\limsup_{n \to \infty} \sup_{\theta \in U} \| \nabla_\theta^s D_n(\theta) \| \le \limsup_{n \to \infty} \frac{1}{n} \sum_{i=1}^{n} G(\bm{x}_i) = \lim_{n \to \infty} \frac{1}{n} \sum_{i=1}^{n} G(\bm{x}_i) \overset{\text{\normalfont{a.s.}}}{=} \E_{X \sim p}[ G(X) ] < \infty
	\end{align*}
	for any $s \in \{ 1, 2, 3 \}$, which establishes C2.

	\vspace{5pt}
	\noindent
	\textbf{C4:} Let $h(\bm{x}, \theta) := \nabla_\theta^2 R(\bm{x}, \theta)$.
	From \eqref{eq:hsd_derivative}, $H_n(\theta) = \frac{1}{n} \sum_{i=1}^{n} h(\bm{x}_i, \theta)$.
	By the strong law of large numbers \citep[Theorem~2.5.10]{Durrett2010b}, we have $H_n(\theta) \as \E_{X \sim p}[ h(X, \theta) ]$ provided that $\E_{X \sim p}[ \| h(X, \theta) \| ] < \infty$.
	Indeed, this condition holds for all $\theta \in U$, since we have the upper bound
	\begin{align*}
		\E_{X \sim p}[ \| h(X, \theta_*) \| ] \le \E_{X \sim p}\left[ \sup_{\theta \in U} \| h(X, \theta) \| \right] \le \E_{X \sim p}\left[ | G(X) | \right] < \infty
	\end{align*}
	where the right hand side is bounded by the preceding argument.
	It remains to verify that $H_* := \lim_{n \to \infty} H_n(\theta_*)$ is equal to $\nabla_\theta^2 \operatorname{DFD}(p_\theta \| p) |_{\theta=\theta_*}$, from which C4 follows since $H_*$ was assumed to be nonsingular in the statement of \Cref{thm:bvm}.
	By the Lebesgue's dominated convergence theorem, for each $\theta \in U$,
	\begin{align*}
		\lim_{n \rightarrow \infty} H_n(\theta) = \E_{X \sim p}[ \nabla_\theta^2 R(\bm{x}, \theta) ] = \nabla_\theta^2 \E_{X \sim p}[ R(\bm{x}, \theta) ] = \nabla_\theta^2 D(\theta) .
	\end{align*}
	provided that $\E_{X \sim p}[ \sup_{\theta \in U} \| \nabla_\theta^2 R(\bm{x}, \theta) \| ] < \infty$.
	This condition holds for all $\theta \in U$ since $\E_{X \sim p}[ \sup_{\theta \in U} \| \nabla_\theta^2 R(\bm{x}, \theta) \| ] \le \E_{X \sim p}\left[ | G(X) | \right] < \infty$.
	Since $\theta_* \in U$ in particular, $H_* = \nabla_\theta^2 D(\theta) \mid_{\theta = \theta_*} = \nabla_\theta^2 \operatorname{DFD}(p_\theta \| p) |_{\theta=\theta_*}$, as claimed.
	
	\vspace{5pt} 
	\noindent
	Thus preconditions C1-C5 are satisfied and the result follows from \Cref{asmp:general_bvm}.
\end{proof}

\subsection{Proof of \Cref{thm:beta_choice}} \label{apx: proof_betaclib}

\begin{proof}
	We first calculate the Fisher divergence between the generalised posterior $\pi_n^D$ and an empirical distribution $\delta_\theta^B$ of the bootstrap minimisers $\{ \theta_n^{(b)} \}_{b=1}^{B}$, and then minimise it as a function of the weighting constant $\beta$.
	Recall that the score-matching divergence \citep{Hyvarinen2005} is given by
	\begin{align*}
            \operatorname{D}(\pi_n^D \| \delta_\theta^B) & = \frac{1}{B} \sum_{b=1}^{B} \underbrace{ \vphantom{\Big(} \left\| \nabla_\theta \log \pi_n^D(\theta_n^{(b)}) \right\|^2 }_{=(*_1)} + \underbrace{ 2 \operatorname{Tr}\Big( \nabla_\theta^2 \log \pi_n^D(\theta_n^{(b)}) \Big) }_{=(*_2)} .
	\end{align*}
	The score function of $\pi_n^D$ is given by
	\begin{align*}
		\nabla_\theta \log \pi_n^D(\theta) = - \beta \nabla_\theta D_n(\theta) + \nabla_\theta \log \pi(\theta) ,
	\end{align*}
	which is independent of the normalising constant of $\pi_n^D$.
	Similarly, the second derivative is $\nabla_\theta^2 \log \pi_n^D(\theta) = - \beta \nabla_\theta^2 D_n(\theta) + \nabla_\theta^2 \log \pi(\theta)$.
	Therefore the terms $(*_1)$ and $(*_2)$ in the Fisher divergence can be written as
	\begin{align*}
		(*_1) & = \beta^2 \| \nabla_\theta D_n(\theta_n^{(b)}) \|^2 - 2 \beta \nabla_\theta D_n(\theta_n^{(b)}) \cdot \nabla_\theta \log \pi(\theta_n^{(b)}) + \| \nabla_\theta \log \pi(\theta_n^{(b)}) \|^2 \\
		(*_2) & = - \beta \operatorname{Tr}\left( \nabla_\theta^2 D_n(\theta_n^{(b)}) \right) + \operatorname{Tr}\left( \nabla_\theta^2 \log \pi(\theta_n^{(b)}) \right)
	\end{align*}
	Now consider minimising the Fisher divergence $\operatorname{D}(\pi_n^D \| \delta_\theta^B)$ with respect to the weighting constant $\beta$.
	Plugging the terms $(*_1)$ and $(*_2)$ in the Fisher divergence, we have
	\begin{align*}
		\operatorname{D}(\pi_n^D \| \delta_\theta^B) & = \frac{1}{B} \sum_{b=1}^{B} \beta^2 \| \nabla_\theta D_n(\theta_n^{(b)}) \|^2 - 2 \beta \nabla_\theta D_n(\theta_n^{(b)}) \cdot \nabla_\theta \log \pi(\theta_n^{(b)}) - 2 \beta \operatorname{Tr}\Big( \nabla_\theta^2 D_n(\theta_n^{(b)}) \Big) + C
	\end{align*}
	where we denote any term independent of $\beta$ by $C$ in this proof.
	Exchanging the order of the summation and the constant $\beta$, the Fisher divergence turns out to be a quadratic function of $\beta$ as follows:
	\begin{align*}
		\operatorname{D}(\pi_n^D \| \delta_\theta^B) & = \beta^2 \underbrace{ \frac{1}{B} \sum_{b=1}^{B} \| \nabla_\theta D_n(\theta_n^{(b)}) \|^2 }_{=(a)} - 2 \beta \underbrace{ \frac{1}{B} \sum_{b=1}^{B}  \nabla_\theta D_n(\theta_n^{(b)}) \cdot \nabla_\theta \log \pi(\theta_n^{(b)}) + \operatorname{Tr}\Big( \nabla_\theta^2 D_n(\theta_n^{(b)}) \Big) }_{=(b)} + C \\
		& = a \beta^2 - 2 b \beta + C = a \left( \beta - \frac{b}{a} \right)^2 - \frac{b^2}{4 a^2} + C 
	\end{align*}
	where the last equality follows from completing the square.
	Therefore the Fisher divergence $\operatorname{D}(\pi_n^D \| \delta_\theta^B)$ is minimised at $\beta_* = b / a$, that is,
	\begin{align*}
		\beta_* = \frac{ \sum_{b=1}^{B} \nabla_\theta D_n(\theta_n^{(b)}) \cdot \nabla_\theta \log \pi(\theta_n^{(b)}) + \operatorname{Tr}\Big( \nabla_\theta^2 D_n(\theta_n^{(b)}) \Big) }{ \sum_{b=1}^{B} \| \nabla_\theta D_n(\theta_n^{(b)}) \|^2 } ,
	\end{align*}
	as claimed, where the denominator and numerator are positive immediately from the first and second assumption respectively, which assures that $\beta_* > 0$.
\end{proof}

\section{Proof of \Cref{asmp:general_bvm}: Simplified Conditions for \cite{Miller2019}} \label{sec: proof_con_bvm}

Before showing that the preconditions C1-C5 of \Cref{asmp:general_bvm} are sufficient for \cite[Theorem~4]{Miller2019}, we introduce the following lemma on a.s.~uniform convergence used in the proof.

\begin{lemma} (a.s. uniform convergence) \label{lem:uc_conv}
	Suppose that the preconditions C1 and C2 in \Cref{asmp:general_bvm} holds for $r=1$.
	Then $D_n$ a.s.~converges uniformly to $D$ on the bounded convex open set $U$ in \Cref{asmp:general_bvm}.
\end{lemma}

\begin{proof}
	\citet[Theorem 21.8]{Davidson1994} showed that $D_n \overset{a.s.}{\longrightarrow} D$ uniformly on $U$ if and only if (a) $D_n \overset{a.s.}{\longrightarrow} D$ pointwise on $U$ and (b) $\{ D_n \}_{n=1}^{\infty}$ is strongly stochastically equicontinuous on $U$.
	The condition (a) is immediately implied by the precondition C1 of \Cref{asmp:general_bvm} and hence the condition (b) is shown in the remainder.
	By \citet[Theorem~21.10]{Davidson1994}, $\{ D_n \}_{n=1}^{\infty}$ is strongly stochastically equicontinuous on $U$ if there exists a stochastic sequence $\{ \mathcal{L}_n \}_{n=1}^{\infty}$ independent of $\theta$ s.t.
	\begin{align*}
		| D_n(\theta) - D_n(\theta') | & \le \mathcal{L}_n \| \theta - \theta' \|_2, \quad \forall \theta, \theta' \in U \qquad \text{and} \qquad \limsup_{n \to \infty} \mathcal{L}_n < \infty \ \text{a.s.} 
	\end{align*}
	Since $D_n$ is continuously differentiable on the set $U$ by the precondition C2 of \Cref{asmp:general_bvm} with $r = 1$, the mean value theorem yields that
	\begin{align*}
		| D_n(\theta) - D_n(\theta') | & \le \sup_{\theta \in U} \| \nabla_{\theta} D_n(\theta) \|_2 \| \theta - \theta' \|_2, \quad \forall \theta, \theta' \in U.
	\end{align*}
	Again by the precondition C2 of \Cref{asmp:general_bvm} with $r = 1$, we have $\limsup_{n \to \infty} \sup_{\theta \in U} \| \nabla_{\theta} D_n(\theta) \|_2 < \infty$ a.s.
	Therefore, setting $\mathcal{L}_n = \sup_{\theta \in U} \| \nabla_{\theta} D_n(\theta) \|_2$ concludes the proof.
\end{proof}

We now show that \cite[Theorem~4]{Miller2019} holds a.s.~under the preconditions C1-C5 of \Cref{asmp:general_bvm}, which in turn implies \Cref{asmp:general_bvm} directly.
A main argument in the proof is essentially same as that of \cite{Matsubara2021} but that is modified here to allow for an arbitrary loss $D_n$.

\begin{proof}
	In order to apply \cite[Theorem~4]{Miller2019}, we first extend $\pi$ and $D_n$ from $\Theta$ to $\R^v$ by setting $\pi(\theta) = 0$ and $D_n(\theta) = \sup_{\theta \in \Theta} | D_n(\theta) | + 1$ for all $\theta \in \R^v \setminus \Theta$, so that we have $\pi: \R^v \to \R$, $D_n: \R^v \to \R$ and $\pi_n^D: \R^v \to \R$.
	Note that in \citet[Theorem~4]{Miller2019}, $\{ D_n \}_{n=1}^{\infty}$ is regarded as a sequence of deterministic functions, while here $\{ D_n \}_{n=1}^{\infty}$ is a sequence of stochastic functions dependent of random data $\{ X_i \}_{i=1}^{n}$.
	It will be shown that \citet[Theorem~4]{Miller2019} holds a.s. for the stochastic sequence $\{ D_n \}_{n=1}^{\infty}$.
	We hence verify the following prerequisites (1)--(6) of \cite[Theorem~4]{Miller2019} a.s.~hold.
	Recall that $H_n(\theta) = \nabla_\theta^2 D_n(\theta)$ and $H_* = \lim_{n \to \infty} H_n(\theta_*)$ from \Cref{asmp:general_bvm}:
	
	\begin{enumerate}
		\item the prior density $\pi$ is continuous at $\theta_*$ and $\pi(\theta_*) > 0$.
		\item $\theta_n \overset{a.s.}{\to} \theta_*$.
		\item the Taylor expansion $D_n(\theta) = D_n(\theta_n) + (1 / 2) (\theta - \theta_n) \cdot H_n(\theta_n) (\theta - \theta_n) + r_n( \theta - \theta_n )$ holds on $U$ a.s. where $r_n$ is the reminder term.
		\item the remainder $r_n$ of the Taylor expansion satisfies that $| r_n( \theta ) | \le C \| \theta \|_2^3,\ \forall \theta \in B_{\epsilon}(0)$ a.s. for all $n$ sufficiently large and some $\epsilon > 0$.
		\item $H_n(\theta_n) \overset{a.s.}{\to} H_*$, $H_n(\theta_n)$ is symmetric for all $n$ sufficiently large and $H_*$ is positive definite.
		\item $\liminf_{n\to\infty} \Big( \inf_{\theta \in \R^v \setminus \mathcal{B}_{\epsilon}(\theta_n)} D_n(\theta) - D_n(\theta_n) \Big) > 0$ a.s. for any $\epsilon > 0$.
	\end{enumerate}
	
	\vspace{5pt}
	\noindent \textbf{Part (1):} The precondition C5 of \Cref{asmp:general_bvm}.
	
	\vspace{5pt}
	\noindent \textbf{Part (2):} The strong consistency $\theta_n \overset{a.s.}{\to} \theta_*$ is shown by an argument similar to \citet[Theorem~5.7]{Vaart1998} or essentially same as \citet[Lemma~3]{Matsubara2021}.
	First, it follows from \Cref{lem:uc_conv} that $D_n \overset{a.s.}{\to} D$ uniformly on $U$ under the conditions of \Cref{asmp:general_bvm}.
	Thus, for all $n$ sufficiently large, we can take $\delta > 0$ s.t. $| D_n(\theta) - D(\theta) | < \delta/2$ a.s. over $\theta \in U$, which in turn leads to (a) $D(\theta) < D_n(\theta) + \delta/2$ and (b) $D_n(\theta) < D(\theta) + \delta/2$ a.s. over $\theta \in U$.
	Then applying both (a) and (b), the following bound on $D(\theta_n)$ holds for all $n$ sufficiently large:
	\begin{align}
		D(\theta_n) \overset{(a)}{<} D_n(\theta_n) + \delta/2 \overset{(*)}{\le} D_n(\theta_*) + \delta/2 \overset{(b)}{<} D(\theta_*) + \delta \quad \text{a.s.} \label{eq:cb_proof_p2}
	\end{align}
	where the second inequality $(*)$ follows from the fact that $\theta_n$ is the minimiser of $D_n$.
	Since $\inf_{\theta \in \R^v} D(\theta) = \inf_{\theta \in \Theta} D(\theta)$ is uniquely attained at $\theta_* \in U$ by \Cref{asmp:general_bvm} (3), for any $\epsilon > 0$ we have $D(\theta) - D(\theta_*) > 0$ for all $\theta \in \R^v \setminus B_\epsilon(\theta_*)$.
	Given an arbitrary $\epsilon > 0$, let $\delta = \inf_{\theta \in \Theta \setminus B_\epsilon(\theta_*)} D(\theta) - D(\theta_*) > 0$.
	It then follows from \eqref{eq:cb_proof_p2} that, for all $n$ sufficiently large,
	\begin{align*}
		D(\theta_n) < \inf_{\theta \in \R^v \setminus B_\epsilon(\theta_*)} D(\theta) \quad \text{a.s.}
	\end{align*}
	This implies that $\theta_n \in B_\epsilon(\theta_*)\ \text{a.s.}$ for any $\epsilon > 0$ arbitrary small for all $n$ sufficiently large.
	Therefore $\theta_n \overset{a.s.}{\to} \theta_*$ by definition of convergence.
	
	\vspace{5pt}
	\noindent \textbf{Part (3):} From the precondition C2 of \Cref{asmp:general_bvm}, $D_n$ is 3 times continuously differentiable over $U$.
	Noting that $\nabla_\theta D_n(\theta) = 0$ at a minimiser $\theta_n$ of $D_n$, the Taylor expansion of $D_n$ around the minimiser $\theta_n$ gives that
	\begin{align*}
		D_n(\theta) & = D_n(\theta_n) + \frac{1}{2} (\theta - \theta_n) \cdot H_n(\theta_n) (\theta - \theta_n) + r_n(\theta - \theta_n)
	\end{align*}
	where $r_n$ is the remainder of the Taylor expansion.
	
	\vspace{5pt}
	\noindent \textbf{Part (4):} Since $r_n$ is the remainder of the Taylor expansion, we have an upper bound
	\begin{align*}
		| r_n( \theta - \theta_n ) | \le \frac{1}{6} \sup_{\theta \in U} \| \nabla_{\theta}^3 D_n(\theta) \|_2 \| \theta - \theta_n \|_2^3, \quad \forall \theta \in U.
	\end{align*}
	The precondition C2 of \Cref{asmp:general_bvm} guarantees that $\limsup_{n \to \infty} \sup_{\theta \in U} \| \nabla_{\theta}^3 D_n(\theta) \|_2 < \infty$ a.s.
	It is thus possible to take some positive constant $C$ s.t. $(1 / 6) \sup_{\theta \in U} \| \nabla_{\theta}^3 D_n(\theta) \|_2 \le C$ a.s. for all $n$ sufficiently large.
	For all $n$ sufficiently large, there exists some open $\epsilon$-neighbour $B_\epsilon(\theta_n)$ contained in the open set $U$ since $\theta_n \in U$. 
	Combining these two facts concludes that
	\begin{align*}
		| r_n( \theta - \theta_n ) | \le C \| \theta - \theta_n \|_2^3, \quad \forall \theta \in B_\epsilon(\theta_n) \quad \Longrightarrow \quad | r_n( \theta ) | \le C \| \theta \|_2^3, \quad \forall \theta \in B_\epsilon(0) 
	\end{align*}
	holds for some $\epsilon > 0$.
	
	\vspace{5pt}
	\noindent \textbf{Part (5):} We first show that $\| H_n(\theta_n) - H_* \|_2 \overset{a.s.}{\to} 0$.
	By the triangle inequality,
	\begin{align*}
		\left\| H_n(\theta_n) - H_* \right\|_2 \le \left\| H_n(\theta_n) - H_n(\theta_*) \right\|_2 + \left\| H_n(\theta_*) - H_* \right\|_2 .
	\end{align*}
	For the first term, it follows from the mean value theorem that
	\begin{align*}
		\left\| H_n(\theta_n) - H_n(\theta_*) \right\|_2 & \le \sup_{\theta \in U} \| \nabla_{\theta} H_n(\theta) \|_2 \| \theta_n - \theta_* \|_2 = \sup_{\theta \in U} \| \nabla_{\theta}^3 D_n(\theta) \|_2 \| \theta_n - \theta_* \|_2 .
	\end{align*}
	The precondition C2 of \Cref{asmp:general_bvm} guarantees that $\limsup_{n \to \infty} \sup_{\theta \in U} \| \nabla_{\theta}^3 D_n(\theta) \|_2 < \infty$ a.s.
	It is thus possible to take some positive constant $C'$ s.t. $\| H_n(\theta_n) - H_n(\theta_*) \|_2 \le C' \| \theta_n - \theta_* \|_2$ for all $n$ sufficiently large.
	Then we have $\left\| H_n(\theta_n) - H_n(\theta_*) \right\|_2 \overset{a.s.}{\to} 0$ by the preceding part (2) $\theta_n \overset{a.s.}{\to} \theta_*$.
	For the second term, it is directly implied by the precondition C4 of \Cref{asmp:general_bvm} that $\left\| H_n(\theta_*) - H_* \right\|_2 \overset{a.s.}{\to} 0$.
	Combining these two facts concludes that $\| H_n(\theta_n) - H_* \|_2 \overset{a.s.}{\to} 0$.
	We next show that $H_n(\theta_n)$ is symmetric.
	The $(i,j)$ entry of $H_n(\theta) = \nabla_\theta^2 D_n(\theta)$ is given by the partial derivative $( \partial^2 / \partial \theta_i \partial \theta_j ) D_n(\theta)$ with respect to $i$-th and $j$-th entry of $\theta$.
	Since $D_n$ is twice continuously differentiable by the precondition C2 of \Cref{asmp:general_bvm}, the Schwartz's theorem implies that the commutation $( \partial^2 / \partial \theta_i \partial \theta_j ) D_n(\theta) = ( \partial^2 / \partial \theta_j \partial \theta_i ) D_n(\theta)$ holds and therefore $H_n(\theta)$ is symmetric for any $\theta \in \Theta$.
	Finally we show positive definiteness of $H_*$.
	For all $n$ sufficiently large, $H_n(\theta_n)$ is positive semi-definite by the fact that $\theta_n$ is the minimiser of $D_n$ and accordingly the limit $H_*$ is positive semi-definite.
	Then $H_*$ is positive definite since $H_*$ is nonsingular by the precondition C4 of \Cref{asmp:general_bvm}.
	
	\vspace{5pt}
	\noindent \textbf{Part (6):} It holds for any sequence $a_n, b_n \in \R$ that $\liminf_{n \to \infty} (a_n - b_n) \ge \liminf_{n \to \infty} a_n + \liminf_{n \to \infty} ( - b_n )$. 
	Furthermore from the property that $\liminf_{n \to \infty} ( - b_n ) = - \limsup_{n \to \infty} b_n$, we have $\liminf_{n \to \infty} (a_n - b_n) \ge \liminf_{n \to \infty} a_n - \limsup_{n \to \infty} b_n$.
	Applying this, we have
	\begin{align*}
		\liminf_{n\to\infty} \left( \inf_{\theta \in \R^v \setminus B_\epsilon(\theta_n)} D_n(\theta) - D_n(\theta_n) \right) = \underbrace{ \liminf_{n\to\infty} \inf_{\theta \in \R^v \setminus B_\epsilon(\theta_n)} D_n(\theta) }_{=:(*_1)} - \underbrace{ \limsup_{n\to\infty} D_n(\theta_n) }_{=:(*_2)} .
	\end{align*}
	For the first term $(*_1)$, it is obvious from the way of extending $D_n$ from $\Theta$ to $\R^v$ that
	\begin{align*}
		(*_1) = \liminf_{n\to\infty} \inf_{\theta \in \R^v \setminus B_\epsilon(\theta_n)} D_n(\theta) \ge \liminf_{n\to\infty} \inf_{\theta \in \Theta \setminus B_\epsilon(\theta_n)} D_n(\theta) \quad \text{ a.s. }
	\end{align*}
	For any set $A \subset \R^v$ and function $g: \R^v \to \R$, define $\inf_{\theta \in A \setminus B_\epsilon(\theta_n)} g(\theta) := \sup_{\theta \in A} g(\theta)$ if $A \setminus B_\epsilon(\theta_n)$ is empty.
	Decomposing $\Theta$ into two sets $U$ and $\Theta \setminus U$ leads to
	\begin{align*}
		(*_1) \ge \liminf_{n\to\infty} \inf_{\theta \in \Theta \setminus B_\epsilon(\theta_n)} D_n(\theta) \ge \min\Big( \underbrace{ \liminf_{n\to\infty} \inf_{\theta \in U \setminus B_\epsilon(\theta_n)} D_n(\theta) }_{=: (*_{11})} ,\ \underbrace{ \liminf_{n\to\infty} \inf_{\theta \in \Theta \setminus ( U \cup B_\epsilon(\theta_n) )} D_n(\theta) }_{=: (*_{12})} \Big)\ \text{ a.s. }
	\end{align*}
	For the term $(*_{11})$, since $D_n \overset{a.s.}{\to} D$ uniformly on $U$ by \Cref{lem:uc_conv} and $\theta_n \overset{a,s.}{\to} \theta_*$ by the preceding part (2), 
	\begin{align*}
		(*_{11}) = \liminf_{n\to\infty} \inf_{\theta \in U \setminus B_\epsilon(\theta_n)} D_n(\theta) = \lim_{n\to\infty} \inf_{\theta \in U \setminus B_\epsilon(\theta_n)} D_n(\theta) = \inf_{\theta \in U \setminus B_\epsilon(\theta_*)} D(\theta) \quad \text{ a.s. }
	\end{align*}
	For the term $(*_{12})$, since the global minimiser $\theta_n$ of $D_n$ is contained in $U$ a.s. for all $n$ sufficiently large by the precondition C3 of \Cref{asmp:general_bvm}, 
	\begin{align*}
		(*_{12}) = \liminf_{n\to\infty} \inf_{\theta \in \Theta \setminus ( U \cup B_\epsilon(\theta_n) )} D_n(\theta) > \liminf_{n\to\infty} \inf_{\theta \in U} D_n(\theta) = \inf_{\theta \in U} D(\theta) = D(\theta_*) \quad \text{ a.s. }
	\end{align*}
	where the second equality follows from the a.s. uniform convergence of $D_n$ on $U$ by \Cref{lem:uc_conv}.
	For the second term $(*_2)$, again since $D_n \overset{a.s.}{\to} D$ uniformly on $U$ and $\theta_n \overset{a,s.}{\to} \theta_*$, we have
	\begin{align*}
		(*_2) = \limsup_{n\to\infty} D_n(\theta_n) = \lim_{n\to\infty} D_n(\theta_n) = D(\theta_*) \quad \text{ a.s. }
	\end{align*}
	The original term $(*_1) - (*_2)$ is lower bounded by $(*_1) - (*_2) \ge \min( (*_{11}) - (*_2) , (*_{12}) - (*_2) )$ a.s., and both the term  $(*_{11}) - (*_2)$ and $(*_{12}) - (*_2)$ are then further lower bounded by
	\begin{align*}
		(*_{11}) - (*_2) = \inf_{\theta \in U \setminus B_\epsilon(\theta_*)} D(\theta) - D(\theta_*) > 0 \quad \text{and} \quad (*_{12}) - (*_2) > D(\theta_*) -  D(\theta_*) = 0 \quad \text{a.s.},
	\end{align*}
	where the first inequality follows from the precondition C3 of \Cref{asmp:general_bvm} indicating that $\inf_{\theta \in \Theta} D(\theta)$ is uniquely attained at $\theta_* \in U$.
	Therefore we have $(*_1) - (*_2) \ge \min( (*_{11}) - (*_2) , (*_{12}) - (*_2) ) > 0$ a.s., which concludes the proof.
\end{proof}

\section{Relation to Stein Discrepancies}
\label{app: ksd connection}

Fisher divergences can be related to a more general class of divergences called Stein discrepancies.
Since their introduction, Stein discrepancies have demonstrated utility over a range of statistical applications, including hypothesis testing, parameter estimation, variational inference, and post-processing of Markov chain Monte Carlo; see \cite{Anastasiou2021} for a review.

This appendix clarifies the sense in which discrete Fisher divergence can be seen as a special case of a discrete Stein discrepancy with an $L^2$-based Stein set. The continuous case was previously covered by Theorem 2 in \cite{Barp2019}.
As a consequence, we deduce that the discrete Fisher divergence is stronger than the popular class of Stein discrepancies based on reproducing kernels.

\subsection{Background on Stein Discrepancies}

Let $\X_*$ be a locally compact Hausdorff space.
For a set $\H$ of functions $f: \X_* \to \R^d$, an operator $S_p: \H \to L^1( p, \R^m)$ depending on a probability distribution $p$ on $\X_*$ is called a \emph{Stein operator} if $\E_{X \sim p}[ S_p[h](X) ] = 0$ for all $h \in \H$.
In these circumstances, we refer to $\H$ as a \emph{Stein set}.
The next proposition defines a particular Stein operator that arises naturally when considering discrete domains $\X_* = \X$, where we recall that $\X$ is a countable space in Standing Assumption 1.
The reader is referred to \cite{Shi2022} for discussion of alternative Stein operators in the discrete context.
Define the forward divergence operator $\nabla^+ \cdot$ for a $\R^d$-valued function $h: \X_* \to \R^d$ by $\nabla^+ \cdot h(\bm{x}) = \sum_{j=1}^{d} h(\bm{x}^{j+}) - h(\bm{x})$.

\begin{proposition} \label{prop: dif_sp_1} 
	Let $p$ be a positive probability distribution on $\X$, such that $(\nabla^- p) / p \in L^2( p, \R^d)$.
	Define an operator $S_p$, acting on functions $h \in L^2( p, \R^d)$, by
	\begin{align}
		S_p[h](\bm{x}) := \frac{ \D^- p(\bm{x}) }{ p(\bm{x}) } \cdot h(\bm{x}) + \D^+ \cdot h(\bm{x}) . \label{eq: dif_sp_2} 
	\end{align}
	Then it holds that $\E_{X \sim p}[ S_p[h](X) ] = 0$.
\end{proposition}
\begin{proof}
	By positivity of $p$ and Cauchy--Schwarz, observe that
	\begin{align}
		\sum_{\bm{x} \in \X} | p(\bm{x}^{i-}) h_i(\bm{x}) | = \sum_{\bm{x} \in \X} p(\bm{x}) \frac{p(\bm{x}^{i-})}{p(\bm{x})} | h_i(\bm{x}) | & = \E_{X \sim p}\left[ \frac{p(X^{i-})}{p(X)} | h_i(X) | \right] \label{eq: abs conv intermed} \\
	& \le \E_{X \sim p}\left[ \frac{p(X^{i-})^2}{p(X)^2} \right] \E_{X \sim p}\left[ h_i(X)^2 \right] < \infty  \nonumber
	\end{align}
	where the first and second term are implied to be finite for any $i = 1, \dots, d$ since $h \in L^2( p, \R^d)$ and $(\nabla^- p) / p \in L^2( p, \R^d)$ which implies $[\nabla^- p(\bm{x}) / p(\bm{x})]_i = 1 - p(\bm{x}^{i-}) / p(\bm{x})$ is square integrable with respect to $p$.
	
	Now, using the definition of $\D^-$ and $\D^+ \cdot$, the Stein operator $S_p$ can be simplified as
	\begin{align}
		S_p[h](\bm{x}) = \sum_{i=1}^{d} h_i(\bm{x}^{i+}) - \frac{ p(\bm{x}^{i-}) }{ p(\bm{x}) } h_i(\bm{x}) . \label{eq:SO_s}
	\end{align}
	The expectation of interest can then be expressed as
	\begin{align*}
		\E_{X \sim p}[ S_p[h](X) ] & = \sum_{\bm{x} \in \X} p(\bm{x}) S_p[h](\bm{x}) = \sum_{i=1}^{d} \sum_{\bm{x} \in \X} p(\bm{x}) h_i(\bm{x}^{i+}) - \sum_{\bm{x} \in \X} p(\bm{x}^{i-}) h_i(\bm{x}) ,
	\end{align*}
	where we have used the absolute convergence of the series, established in \eqref{eq: abs conv intermed}, to justify the re-ordering of terms.
	The result is then immediate from \Cref{prop:sumbypart}.
\end{proof}

\noindent The Stein operator \eqref{eq: dif_sp_2} can be considered a discrete analogue of the Langevin Stein operator for continuous domains; see \citet{Yang2018}.

Given a Stein operator $S_p$ and Stein set $\H$, the \emph{Stein discrepancy} between probability distributions $p$ and $q$ on $\X$ is defined as the maximum deviation between expectations of the test functions $S_p[h]$ for $h \in \H$:
\begin{align}
	\operatorname{SD}(p \| q) := \sup_{h \in \H} \left| \E_{X \sim q}\left[ S_p[h](X) \right] - \E_{X \sim p}\left[ S_p[h](X) \right] \right| = \sup_{h \in \H} \left| \E_{X \sim q}\left[ S_p[h](X) \right] \right|  \label{eq:SD}
\end{align}
The final equality follows from \Cref{prop: dif_sp_1}, and our discussion in this appendix implicitly assumes all relevant quantities are well-defined.
The Stein discrepancy is computable\footnote{That is, the expectations do not involve the normalising constant; whether the supremum over the Stein set is computable depends on how the Stein set is selected.} without knowing the normalising constant of $p$ since it depends on $p$ only through the ratio $(\nabla^{-} p) / p$, in a similar manner to discrete Fisher divergence in the main text.

\subsection{Discrete Fisher Divergence as a Stein Discrepancy}

We now establish that the discrete Fisher divergence, introduced in the main text, is in fact a Stein discrepancy, corresponding to the Stein operator in \Cref{prop: dif_sp_1} and a Stein set equal to the unit ball of $L^2(q,\R^d)$.
This observation will allow us to conclude, in \Cref{subsec: KSD vs DFD}, that discrete Fisher divergence is stronger than popular kernel Stein discrepancies.

\begin{proposition} \label{prop: dif_sp_2} 
	Let $p$ and $q$ be positive distributions on $\X$, such that $(\nabla^- p) / p, (\nabla^- q) / q \in L^2(q,\R^d)$.
	Consider a Stein discrepancy whose Stein operator is \eqref{eq: dif_sp_2} and whose Stein set is $\H = \{ h: \X \to \R^d \mid \sum_{i=1}^{d} \E_{X \sim q}[ h_i(X)^2 ] \le 1 \}$.
	Then
	\begin{align}
		\operatorname{SD}(p \| q) = \sqrt{ \operatorname{DFD}(p \| q) } .
	\end{align}
\end{proposition}
\begin{proof}
	From \eqref{eq: dif_sp_2} and \eqref{eq:SD}, we have that
	\begin{align*}
		\operatorname{SD}(p \| q) & = \sup_{h \in \H} \left| \E_{X \sim q}\left[ \frac{\nabla^- p(X)}{p(X)} \cdot h(X) - \frac{\nabla^- q(X)}{q(X)} \cdot h(X) \right] \right| .
	\end{align*}
	Note that $L^2(q,\R^d)$ is a Hilbert space when equipped with the inner product $\langle f, g \rangle_{L^2(q,\R^d)} := \E_{X \sim q}[ f(X) \cdot g(X) ]$.
	Thus we can view $\operatorname{SD}(p \| q)$ as the maximum of the inner product
	\begin{align}
		\operatorname{SD}(p \| q) & = \sup_{h \in \H} \left| \left\langle \frac{\nabla^- p}{p} - \frac{\nabla^- q}{q}, h \right\rangle_{L^2(q,\R^d)} \right| , \label{eq: stein in prod}
	\end{align}
	which is well-defined since $u := (\nabla^-p) / p - (\nabla^-q) / q \in L^2(q,\R^d)$.
	Let $\| \cdot \|_{L^2(q,\R^d)}$ denote the norm of $L^2(q,\R^d)$, so that $\H$ is the set of $f \in L^2(q,\R^d)$ for which $\|f\|_{L^2(q,\R^d)} \leq 1$.
	By the Cauchy--Schwarz inequality, the inner product in \eqref{eq: stein in prod} attains its supremum at $h = u / \| u \|_{L^2(q,\R^d)} \in \H$.
	Therefore 
	\begin{align*}
		\operatorname{SD}(p \| q) & = \sup_{h \in \H} | \langle u, h \rangle_{L^2(q,\R^d)} | = \| u \|_{L^2(q,\R^d)} = \sqrt{ \E_{X \sim q}\left[ \left\| \frac{\nabla^- p(X)}{p(X)} - \frac{\nabla^- q(X)}{q(X)} \right\|^2 \right] } ,
	\end{align*}
	which concludes the proof.
\end{proof}

\subsection{The Fisher Divergence Dominates the Kernel Stein Discrepancy} \label{subsec: KSD vs DFD}

A popular choice of Stein set $\H$, that can lead to a closed form Stein discrepancy, is the unit ball of a reproducing kernel Hilbert space.
The resulting \emph{kernel Stein discrepancy} was recently considered in the discrete context in \citet{Yang2018}.
In this appendix we establish that our discrete Fisher divergence, introduced in the main text, is a stronger notion of divergence than kernel Stein discrepancy.
This may render the discrete Fisher divergence more statistically efficient in applications where a statistical model is well-specified, in addition to the computational advantage (\Cref{rem: cost}) and the non-reliance on a user-specified kernel discussed in the main text.

A symmetric, positive definite function $k: \X_* \times \X_* \to \R$ is called a kernel.
For every kernel $k$, there exists a unique associated Hilbert space of real-valued functions on $\X_*$, called a reproducing kernel Hilbert space, denoted $\H_k$; see e.g.~\cite{berlinet2011reproducing} for background.
Let $\H_k^d := \H_k \times \cdots \times \H_k$, that is, a space of functions $h: \X_* \to \R^d$ whose each $i$-th output-coordinate $h_i: \X_* \to \R$ belongs to $\H_k$.
\cite{Yang2018} studied the Stein discrepancy for a discrete space $\X_* = \X$, of finite cardinality only, using the Stein operator \eqref{eq: dif_sp_2} and a Stein set $\H = \{h \in \H_k^d : \sum_{i=1}^d \|h_i\|_{\H_k}^2 \leq 1 \}$. 
Here we first establish that the Stein set $\{h \in \H_k^d : \sum_{i=1}^d \|h_i\|_{\H_k}^2 \leq 1 \}$ constructed from $\H_k^d$ is contained in another Stein set $\{h \in L^2(q, \R^d): \|h\|_{L^2(q, \R^d)}^2 \leq 1 \}$ constructed from $L^2(q, \R^d)$ for any general domain $\X_*$, under a standard condition on the reproducing kernel.
This in turn shows that the discrete Fisher divergence dominates the kernel Stein discrepancy.

\begin{proposition} \label{prop: dif_sp_pre3} 
	Let $q$ be a probability distribution on $\X_*$.
	Let $k : \X_* \times \X_* \rightarrow \R$ be a kernel such that $k(\bm{x}, \bm{x}) \le 1$ for all $\bm{x} \in \X_*$.
	Then the unit ball of $\H_k^d$ is contained in the unit ball of $L^2(q,\R^d)$.
\end{proposition}
\begin{proof}
	First let $f: \X_* \to \R^d$ be any element of $\H_k^d$, where its $i$-th output-coordinate $f_i: \X_* \to \R$ belongs to $\H_k$ each.
	From the reproducing property of $\H_k$, followed by the Cauchy--Schwartz inequality, the norm of $f$ in $L^2(q, \R^d)$ is upper bounded as follows:
	\begin{align*}
		\|f\|_{L^2(q,\R^d)}^2 = \sum_{i=1}^d \E_{X \sim q}[f_i(X)^2] & = \sum_{i=1}^d \E_{X \sim q}[ \langle f_i(\cdot), k(X, \cdot) \rangle_{\H_k}^2 ] \\
		& \le \sum_{i=1}^d \E_{X \sim q}\left[ \| f_i \|_{\H_k}^2 \| k(X, \cdot) \|_{\H_k}^2 \right]  = \sum_{i=1}^d \E_{X \sim q}\left[ \| f_i \|_{\H_k}^2 k(X, X) \right] \\
		& = \left( \sum_{i=1}^d \| f_i \|_{\H_k}^2 \right) \E_{X \sim q}\left[ k(X, X) \right]
		= \|f\|_{\H_k^d}^2 \; \E_{X \sim q}  \left[ k(X, X) \right].
	\end{align*}
	The continuous embedding of $\H_k^d$ in $L^2(q, \R^d)$ therefore holds, and moreover the embedding constant is at most one, since $\E_{X \sim q} [k(X,X)] \leq 1$ due to the assumption that $k(\bm{x}, \bm{x}) \le 1$ for all $\bm{x} \in \X_*$.
	In particular, it follows that the unit ball of $\H_k^d$ is contained in the unit ball of $L^2(q,\R^d)$.
\end{proof}

Built upon \Cref{prop: dif_sp_pre3}, we can immediately show the discrete Fisher divergence dominates the kernel Stein discrepancy for the case where $\X_* = \X$.

\begin{proposition} \label{prop: dif_sp_3} 
	Let $p$ and $q$ be positive distributions on $\X$, such that $( \nabla^- p ) / p, ( \nabla^- q ) / q \in L^2(q,\R^d)$.
	Let $k : \X \times \X \rightarrow \R$ be a kernel such that $k(\bm{x}, \bm{x}) \le 1$ for all $\bm{x} \in \X$.
	Let $\S_p$ be a Stein operator in \eqref{eq: dif_sp_2}.
	Then the kernel Stein discrepancy, denoted $\operatorname{SD}_k$, satisfies $\operatorname{SD}_k(p \| q) \le \sqrt{ \operatorname{DFD}(p \| q) }$.
\end{proposition}
\begin{proof}
	From \eqref{eq:SD} (which in turn relies on \Cref{prop: dif_sp_1}), it is straightforward to see that
	\begin{align*}
		\operatorname{SD}_k(p \| q) = \sup_{\|h\|_{\H_k^d} \leq 1} \left| \E_{X \sim q}\left[ S_p[h](X) \right] \right| \le \sup_{\|h\|_{L^2(q,\R^d)} \leq 1} \left| \E_{X \sim q}\left[ S_p[h](X) \right] \right| = \sqrt{ \operatorname{DFD}(p \| q) } ,
	\end{align*}
	where the inequality follows from \Cref{prop: dif_sp_pre3} immediately and the final equality is \Cref{prop: dif_sp_2}.
\end{proof}

This argument is not restricted to the discrete case but is immediately applicable for the continuous case.
One of the most common Stein operator for a continuous domain $\X_* = \R^d$ is
\begin{align}
	\S_p[h](\bm{x}) = \nabla \log p(\bm{x}) \cdot h(\bm{x}) + \nabla \cdot h(\bm{x}) \label{eq:lagevin_stein}
\end{align}
The Fisher divergence $\operatorname{FD}(p \| q) = \E_{X \sim q}[ \| \nabla \log p(X) - \nabla \log q(X) \|^2 ]$ for densities $p, q$ on $\X_*$ dominates the kernel Stein discrepancy constructed from the above Stein operator and the kernel on $\X_*$.

\begin{proposition} \label{prop: dif_sp_4} 
	Let $\X_* = \R^d$.
	Let $p$ and $q$ be positive continuously differentiable densities on $\X_*$, such that $\nabla \log p, \nabla \log q \in L^2(q,\R^d)$.
	Let $k : \X_* \times \X_* \rightarrow \R$ be a kernel such that $k(\bm{x}, \bm{x}) \le 1$ for all $\bm{x} \in \X_*$.
	Let $\S_p$ be a Stein operator in \eqref{eq:lagevin_stein}
	Then the kernel Stein discrepancy, denoted $\operatorname{SD}_k$, satisfies $\operatorname{SD}_k(p \| q) \le \sqrt{ \operatorname{FD}(p \| q) }$.
\end{proposition}
\begin{proof}
	We repeat the same argument as \Cref{prop: dif_sp_3}.
	From \cite[Theorem~2]{Barp2019}, $\operatorname{FD}(p \| q)$ can be written as the Stein discrepancy constructed by the Stein set $\{h \in L^2(q, \R^d): \|h\|_{L^2(q, \R^d)}^2 \leq 1 \}$.
	Then from \eqref{eq:SD} (which in turn relies on \Cref{prop: dif_sp_1}), it is straightforward to see that
	\begin{align*}
		\operatorname{SD}_k(p \| q) = \sup_{\|h\|_{\H_k^d} \leq 1} \left| \E_{X \sim q}\left[ S_p[h](X) \right] \right| \le \sup_{\|h\|_{L^2(q,\R^d)} \leq 1} \left| \E_{X \sim q}\left[ S_p[h](X) \right] \right| = \sqrt{ \operatorname{FD}(p \| q) } ,
	\end{align*}
	where the inequality follows from \Cref{prop: dif_sp_pre3} immediately and the final equality is \Cref{prop: dif_sp_2}.
\end{proof}

An interesting recent observation in \citet{Shi2022} was that alternative Stein operators \citep[such as Gibbs and Barker operators; see Table 1 of][]{Shi2022} gave rise to kernel Stein discrepancies that performed better in their particular context (low-variance gradient estimation). 
It would be interesting to explore the analogous alternatives to discrete Fisher divergence that would result from such operators, but this is left to future work.

\subsection{Robustness of the Kernel Stein Discrepancy} \label{subsec: KSD_robustness}

\Cref{subsec: KSD vs DFD} indicates statistical efficiency of the discrete Fisher divergence over the kernel Stein discrepancy.
If one's model is well-specified, minimising the discrete Fisher divergence leads us to a correct model faster than the kernel Stein discrepancy.
However this does not mean that the use of the discrete Fisher divergence is always better than the kernel Stein discrepancy.
In particular, the kernel Stein discrepancy can be equipped with strong robustness by choosing an appropriate kernel.
To demonstrate this, we compare three posteriors of Pseudo-Bayes, DFD-Bayes, and KSD-Bayes for the same Ising model as \Cref{subsec: Ising assessment} with $d = 100~(m=10)$ in a setting where dataset contains extreme outliers with a proportion $\epsilon$.

We approximately draw 1000 samples $\{ x_i \}_{i=1}^{1000}$ from the Ising model $p_\theta$ with $\theta = 5$ by the same Metropolis–Hastings algorithm as \Cref{subsec: Ising assessment}.
To study the robustness of the posteriors, we replaced a proportion $\epsilon = 0.1$ of the data with the vector $( 1, 1, \cdots, 1 )$ corresponding to the extreme value in $\X$ that is rarely drawn from the model.
\cite{Matsubara2021} showed that KSD-Bayes can satisfy strong qualitative robustness called ``global bias-robustness" by choosing a kernel appropriately.
For this example, we use the same choice of kernel as \cite{Matsubara2021} below:
\begin{align*}
	k(\bm{x}, \bm{x}') = m(\bm{x}) \exp\left( - \frac{1}{d} \sum_{i=1}^d \mathbbm{1}( x_{i} - x_{i}' ) \right) m(\bm{x})
\end{align*}
where $m(x) = \sigma( 90 - | \sum_i x_{i} | )$ based on a sigmoid function $\sigma(t) = ( 1 + \exp( - t ) )^{-1}$.
This is indeed a proper choice of kernel, and the function $m(\bm{x})$ in the definition of kernel is designed to restrict the influence of extreme data whose norm is closer to or larger than $90$.

\begin{figure}[t!]
	\centering
	\hfill
	\includegraphics[height=0.155\textheight]{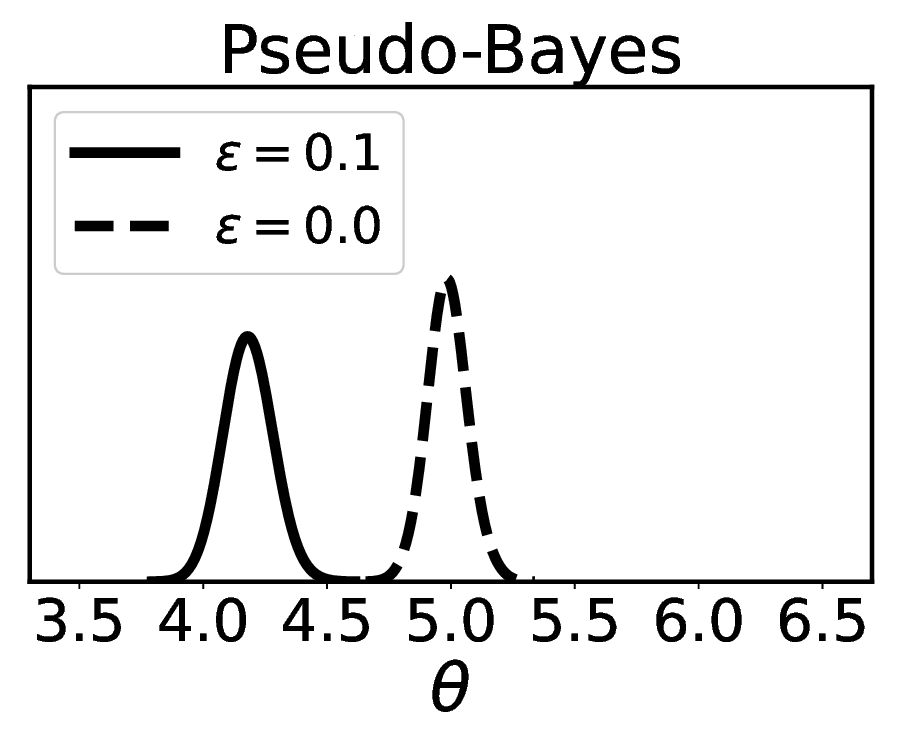}
	\hfill
	\includegraphics[height=0.155\textheight]{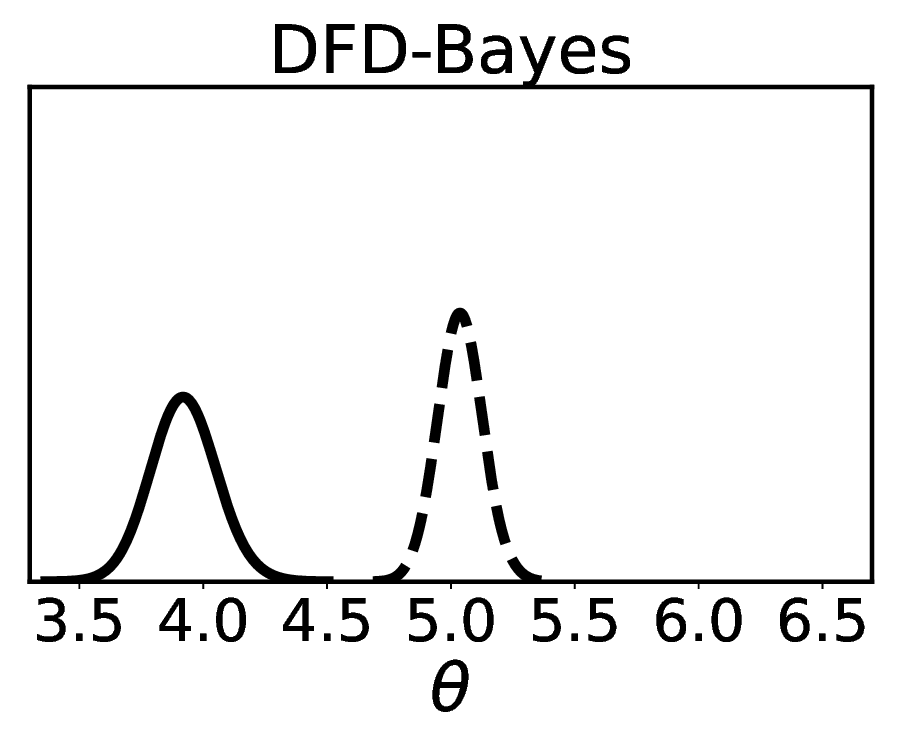}
	\hfill
	\includegraphics[height=0.155\textheight]{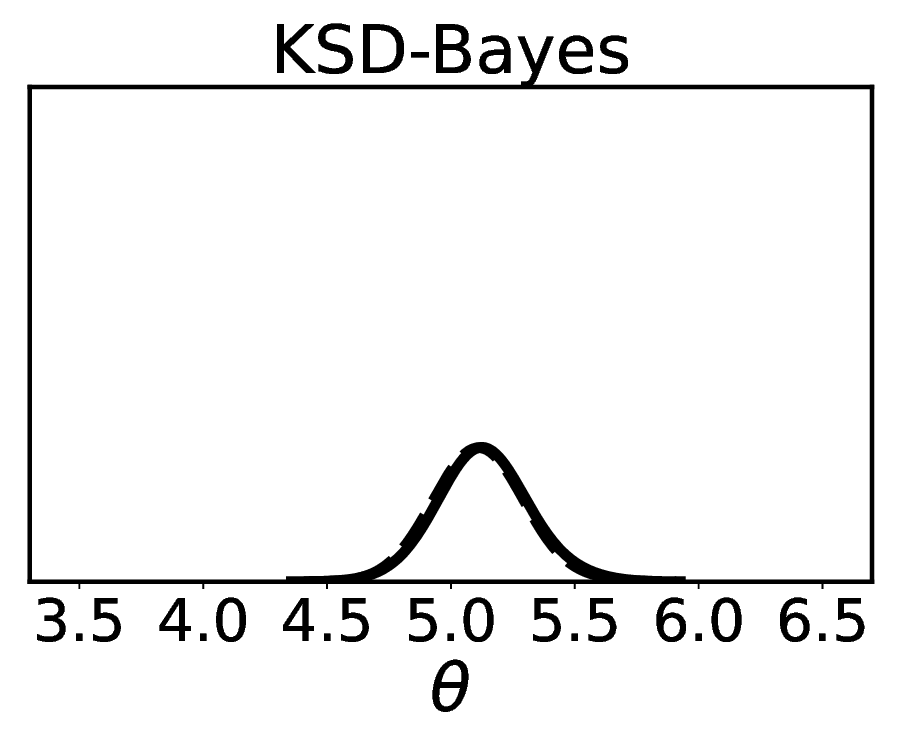}
	\hfill
	\hfill
	\caption{Posteriors of Pseudo-Bayes (left), DFD-Bayes (centre), and KSD-Bayes (right) for the Ising model in the presence of outlier with $\epsilon = 0.1$ and no outlier with $\epsilon = 0.0$.}
	\label{fig: robustness}
\end{figure}

In \Cref{fig: robustness} demonstrated that KSD-Bayes offered a correct inference outcome even when the dataset contains outliers, being less affected by the outliers. 
On the other hand, the Pseudo-Bayes and DFD-Bayes posteriors placed the majority of the probability mass on smaller $\theta$ than the correct value $\theta = 5$.
The extreme value $( 1, 1, \cdots, 1 )$ of the outliers is more likely to be drawn from the model of $\theta \ll 1$; the posteriors of Pseudo-Bayes and DFD-Bayes were thus pulled in the direction of smaller $\theta$.

\section{Calculations for Worked Examples}  \label{ap: calcs for examples}

\subsection{\Cref{asmp:derivative}  for \Cref{ex:exp_asmp}} \label{apx:example}
	The aim of this section is to establish when \Cref{asmp:derivative} is satisfied for the exponential family model in \Cref{ex:exp_asmp}.
	For better presentation, let $T_{j-}(\bm{x}) := T(\bm{x}^{j-}) - T(\bm{x})$ and $b_{j-}(\bm{x}) := b(\bm{x}^{j-}) - b(\bm{x})$ to see that $r_{j-}(\bm{x}, \theta) = \exp( \eta(\theta) \cdot T_{j-}(\bm{x}) + b_{j-}(\bm{x}) )$.
	In addition, let $T_{j+}(\bm{x}) := T(\bm{x}) - T(\bm{x}^{j+})$ and $b_{j+}(\bm{x}) := b(\bm{x}) - b(\bm{x}^{j+})$ to see that $r_{j-}(\bm{x}^{j+}, \theta) = \exp( \eta(\theta) \cdot T_{j+}(\bm{x}) + b_{j+}(\bm{x}) )$.
	It is straightforward to see that, for any $\bm{x} \in \X$,
	\begin{align*}
		\nabla_\theta r_{j-}(\bm{x}^{j+}, \theta) & = \nabla_\theta \eta(\theta) \cdot T_{j+}(\bm{x}) \, \exp( \eta(\theta) \cdot T_{j+}(\bm{x}) + b_{j+}(\bm{x}) ) \\
		& = \nabla_\theta \eta(\theta) \cdot T_{j+}(\bm{x}) \, \exp( \eta(\theta) \cdot T_{j+}(\bm{x}) ) \, \exp( b_{j+}(\bm{x}) ) \\
		\nabla_\theta ( r_{j-}(\bm{x}, \theta)^2 ) & = 2 r_{j-}(\bm{x}, \theta) \, \nabla_\theta r_{j-}(\bm{x}, \theta) \\
		& = 2 \nabla_\theta \eta(\theta) \cdot T_{j-}(\bm{x}) \, \exp( 2 \eta(\theta) \cdot T_{j-}(\bm{x}) ) \, \exp( 2 b_{j-}(\bm{x}) ) 
	\end{align*}
	By assumption, $T_{j-}(\bm{x})$ is bounded over all $\bm{x} \in \X$, which in turn shows that $T_{j+}(\bm{x}) = T_{j-}(\bm{x}^{j+})$ is bounded over all $\bm{x} \in \X$ since $\bm{x}^{j+} \in \X$.
	Further, by assumption, $\sup_{\theta \in U} \| \nabla_\theta \eta(\theta) \| < \infty$ and $\sup_{\theta \in U} \| \eta(\theta) \| < \infty$.
	Let $M$ be a constant that upper bounds all the terms $\sup_{\bm{x} \in \X} \| T_{j-}(\bm{x}) \|$, $\sup_{\bm{x} \in \X} \| T_{j+}(\bm{x}) \|$, $\sup_{\theta \in U} \| \nabla_\theta \eta(\theta) \|$ and $\sup_{\theta \in U} \| \eta(\theta) \|$.
	Then we have
	\begin{align*}
		\sup_{\theta \in U} \| \nabla_\theta r_{j-}(\bm{x}^{j+}, \theta) \| & \le M^2 \exp\left(M^2 \right) \exp( b_{j+}(\bm{x}) ) , \\
		\sup_{\theta \in U} \| \nabla_\theta ( r_{j-}(\bm{x}, \theta )^2 ) \| & \le 2 M^2 \exp\left( 2 M^2 \right) \exp( 2 b_{j-}(\bm{x}) ) .
	\end{align*}
	Taking the expectations,
	\begin{align}
		\E_{X \sim p}\left[ \sup_{\theta \in U} \| \nabla_\theta r_{j-}(X^{j+}, \theta) \| \right] & \le M^2 \exp( M^2 ) \E_{X \sim p}\left[ \exp( b_{j+}(X) ) \right] , \label{eq: to be finite 1} \\
		\E_{X \sim p}\left[ \sup_{\theta \in U} \| \nabla_\theta ( r_{j-}(X, \theta)^2 ) \| \right] & \le 2 M^2 \exp( 2 M^2 ) \E_{X \sim p}\left[ \exp( 2 b_{j-}(X) ) \right] .  \label{eq: to be finite 2}
	\end{align}
	By assumption $\E_{X \sim p}\left[ \exp( 2 b_{j-}(X) ) \right] = \E_{X \sim p}\left[ \exp( b_{j-}(X) )^2 \right] < \infty$ , and we now argue that this  also implies $\E_{X \sim p}\left[ \exp( b_{j+}(X) ) \right] < \infty$.
	Indeed, from \Cref{prop:sumbypart}, 
	\begin{align*}
		\E_{X \sim p}\left[ \exp( b_{j+}(X) ) \right] & = \sum_{\bm{x} \in \X} p(\bm{x}) \exp( b(\bm{x}) - b(\bm{x}^{j+}) ) = \sum_{\bm{x} \in \X} p(\bm{x}^{j-}) \exp( b(\bm{x}^{j-}) - b(\bm{x}) ) \\
		& = \sum_{\bm{x} \in \X} p(\bm{x}) \frac{p(\bm{x}^{j-})}{p(\bm{x})} \exp( b(\bm{x}^{j-}) - b(\bm{x}) ) = \E_{X \sim p}\left[ \frac{p(X^{j-})}{p(X)} \exp( b_{j-}(X) ) \right] .
	\end{align*}
	Now, using the Cauchy--Schwartz inequality, 
	\begin{align}
		\E_{X \sim p}\left[ \exp( b_{j+}(X) ) \right] \le \E_{X \sim p}\left[ \frac{p(X^{j-})^2}{p(X)^2} \right] \E_{X \sim p}\left[ \exp( 2 b_{j-}(X) ) \right] . \label{eq: CS appendix}
	\end{align}
	Existence of the first term in \eqref{eq: CS appendix} is implied by the Standing Assumption $(\nabla^- p) / p \in L^2( p, \R^d)$, while existence of the second term in \eqref{eq: CS appendix} was assumed.
	Therefore we have shown that \eqref{eq: to be finite 1} and \eqref{eq: to be finite 2} exist.
	Repeating an essentially identical argument, it is straightforward to see also that $\E_{X \sim p}\left[ \sup_{\theta \in U} \| \nabla_\theta^s r_{j-}(X^{j+}, \theta) \| \right] < \infty$ and $\E_{X \sim p}\left[ \sup_{\theta \in U} \| \nabla_\theta^s ( r_{j-}(X, \theta )^2 ) \| \right] < \infty$ for $s = 2, 3$ as claimed.

\subsection{Derivatives of \eqref{eq:beta_optimal} for \Cref{ex:exp_asmp}} \label{apx:example_2}
	Automatic differentiation is an attractive and promising choice to compute \eqref{eq:beta_optimal} whenever it is available.
	Nonetheless, it is still straightforward for a majority of parametric models to compute the loss derivatives used in \eqref{eq:beta_optimal}.
	This section aims to demonstrate a form of loss derivatives for a model in \Cref{ex:exp_asmp}.
	The optimal $\beta$ of \eqref{eq:beta_optimal} depends on the first and second derivative of a loss $D$ specified by users.
	Consider the discrete Fisher divergence $D_n$ that this paper established.
	The discrete Fisher divergence $D_n(\theta) = \operatorname{DFD}(p_\theta \| p_n)$ between a model $p_\theta$ in \Cref{ex:exp_asmp} and an empirical distribution $p_n$ of data $\{ \bm{x}_i \}_{i=1}^{n}$ is given as
	\begin{align*}
		D_n(\theta) = \frac{1}{n} \sum_{i=1}^{n} \sum_{j=1}^{d} \left( r_{j-}(\bm{x}, \theta) \right)^2 - 2 r_{j-}(\bm{x}^{j+}, \theta)
	\end{align*}
	For simplicity, let $\eta(\theta) = \theta$ here.
	Then $r_{j-}(\bm{x}, \theta) = \exp(\theta \cdot T_{j-}(\bm{x}) + b_{j-}(\bm{x}) )$ and $r_{j-}(\bm{x}^{j+}, \theta) = \exp(\theta \cdot T_{j+}(\bm{x}) + b_{j+}(\bm{x}) )$ using the notations in \Cref{apx:example}.
	Therefore the derivatives are
	\begin{align*}
		\nabla_\theta r_{j-}(\bm{x}, \theta) & = T_{j-}(\bm{x}) \, \exp( \theta \cdot T_{j-}(\bm{x}) + b_{j-}(\bm{x}) ) , \\
		\nabla_\theta^2 r_{j-}(\bm{x}, \theta) & = T_{j-}(\bm{x}) \otimes T_{j-}(\bm{x}) \, \exp( \theta \cdot T_{j-}(\bm{x}) + b_{j-}(\bm{x}) ) , \\
		\nabla_\theta r_{j-}(\bm{x}^{j+}, \theta) & = T_{j+}(\bm{x}) \, \exp( \theta \cdot T_{j+}(\bm{x}) + b_{j+}(\bm{x}) ) , \\
		\nabla_\theta^2 r_{j-}(\bm{x}^{j+}, \theta) & = T_{j+}(\bm{x}) \otimes T_{j+}(\bm{x}) \, \exp( \theta \cdot T_{j+}(\bm{x}) + b_{j+}(\bm{x}) )
	\end{align*}
	where $\otimes$ denotes outer product.
	Built upon these components, we have the required first derivatives of $D_n(\theta)$
	\begin{align*}
		\nabla_\theta D_n(\theta) & = \frac{1}{n} \sum_{i=1}^{n} \sum_{j=1}^{d} \nabla_\theta \left( r_{j-}(\bm{x}, \theta)^2 \right) - 2 \nabla_\theta r_{j-}(\bm{x}^{j+}, \theta) \\
		& = \frac{1}{n} \sum_{i=1}^{n} \sum_{j=1}^{d} 2 r_{j-}(\bm{x}, \theta) \nabla_\theta r_{j-}(\bm{x}, \theta) - 2 \nabla_\theta r_{j-}(\bm{x}^{j+}, \theta) \\
		& = \frac{2}{n} \sum_{i=1}^{n} \sum_{j=1}^{d} T_{j-}(\bm{x}) \, \exp( \theta \cdot T_{j-}(\bm{x}) + b_{j-}(\bm{x}) )^2 - T_{j+}(\bm{x}) \, \exp( \theta \cdot T_{j+}(\bm{x}) + b_{j+}(\bm{x}) )
	\end{align*}
	as well as the second derivative of $D_n(\theta)$
	\begin{align*}
		\nabla_\theta^2 D_n(\theta) & = \frac{1}{n} \sum_{i=1}^{n} \sum_{j=1}^{d} \nabla_\theta \left( 2 r_{j-}(\bm{x}, \theta) \nabla_\theta r_{j-}(\bm{x}, \theta) \right) - \nabla_\theta \left( 2 \nabla_\theta r_{j-}(\bm{x}^{j+}, \theta) \right) \\
		& = \frac{2}{n} \sum_{i=1}^{n} \sum_{j=1}^{d} \nabla_\theta r_{j-}(\bm{x}, \theta) \otimes \nabla_\theta r_{j-}(\bm{x}, \theta) + r_{j-}(\bm{x}, \theta) \nabla_\theta^2 r_{j-}(\bm{x}, \theta) - \nabla_\theta^2 r_{j-}(\bm{x}^{j+}, \theta) \\
		& = \frac{2}{n} \sum_{i=1}^{n} \sum_{j=1}^{d} 2 T_{j-}(\bm{x}) \otimes T_{j-}(\bm{x}) \, \exp( \theta \cdot T_{j-}(\bm{x}) + b_{j-}(\bm{x}) )^2 \\
            & \hspace{175pt} - T_{j+}(\bm{x}) \otimes T_{j+}(\bm{x}) \, \exp( \theta \cdot T_{j+}(\bm{x}) + b_{j+}(\bm{x}) ) .
	\end{align*}
	Plugging these derivatives $\nabla_\theta D_n(\theta)$ and $\nabla_\theta^2 D_n(\theta)$ and a given $\nabla_\theta \log \pi(\theta)$ in \eqref{eq:beta_optimal}, the optimal $\beta$ is computed.

\subsection{\Cref{asmp:derivative} for Poisson, Ising, and Conway-Maxwell-Poisson Models} \label{apx:example_3}
	\Cref{asmp:derivative} for the Poisson and Ising models used in the experiments can be verified as a special case of \Cref{ex:exp_asmp}.
	Any Poisson model can be written in the form 
	\begin{align*}
		p_\theta(x) \propto \exp\bigg( \log(\theta_1) \; x - \sum_{k=1}^{x} \log (k) \bigg) .
	\end{align*}
	This falls into a class of exponential family in \Cref{ex:exp_asmp} by setting $\eta(\theta) = \log(\theta)$, $T(x) = x$, and $b(x) = - \sum_{k=1}^{x} \log (k)$.
	This gives that $T(x-1) - T(x) = -1$ and $b(x-1) - b(x) = \log(x)$.
	The condition in \Cref{ex:exp_asmp} is satisfied provided that $\E_{X \sim p}[ \exp( \log (X) )^2 ] = \E_{X \sim p}[ X^2 ] < \infty$, i.e.~$p$ has a second moment.
	Similarly, any Ising model can be written in the form 
	\begin{align*}
		p_\theta(x) \propto \exp( \theta \cdot T(\bm{x}) )
	\end{align*}
	where $T : \X \rightarrow \R^k$ is a vector of summary statistics that define the model.
	For Ising models, $\X$ is of finite cardinality and $T(\bm{x})$ is hence bounded for any $\bm{x} \in \X$.
	The conditions in \Cref{ex:exp_asmp} are then automatically satisfied.

	The Conway-Maxwell-Poisson model falls into a class of exponential family but it is beyond the simplified case of \Cref{ex:exp_asmp}.
	Nonetheless, \Cref{asmp:derivative} is still verifiable.
	Recall that the Conway-Maxwell-Poisson model has the form $p_\theta(x) \propto (\theta_1)^x (x!)^{-\theta_2}$ whose ratio function is given by $r_{j-}(x, \theta) = p_\theta(x-1) / p_\theta(x) = x^{\theta_2} / \theta_1 $ where $\theta_1, \theta_2 \in (0, \infty)$.
	The derivative of the ratio with respect to $\theta = (\theta_1, \theta_2)$ is then given by
	\begin{align*}
		\nabla_\theta r_{j-}(x+1, \theta) = \left( - \frac{ (x+1)^{\theta_2} }{ \theta_1^2 } , \frac{ (x+1)^{\theta_2} \log (x+1) }{ \theta_1 } \right), \quad \nabla_\theta ( r_{j-}(x, \theta) )^2 & = \left( - \frac{ x^{2 \theta_2} }{ \theta_1^3 } , \frac{ x^{2 \theta_2} \log x }{ \theta_1^2 } \right) .
	\end{align*}
	Note that the term $x^{2 \theta_2} \log x$ in $\nabla_\theta ( r_{j-}(x, \theta) )^2$ is well-defined even at $x = 0$ since it converges to $0$ as $x \to 0$ if $\theta_2 > 0$ despite the individual term $\log x$ alone is not well-defined for $x = 0$.
	Let $M_1$ and $M_2$ be the infimum value of $\theta_1$ and the supremum value of $\theta_2$ for $(\theta_1, \theta_2)$ in the bounded set $U$ to see that
	\begin{align*}
		\sup_{\theta \in U} \| \nabla_\theta r_{j-}(x+1, \theta) \| & = \left| \frac{ (x+1)^{M_2} }{ M_1^2 } \right| + \left| \frac{ (x+1)^{M_2} \log (x+1) }{ M_1 } \right| , \\
		\sup_{\theta \in U} \| \nabla_\theta ( r_{j-}(x, \theta) )^2 \| & = \left| \frac{ x^{2 M_2} }{ M_1^3 } \right| + \left| \frac{ x^{2 M_2} \log x }{ M_1^2 } \right| .
	\end{align*}
	We can derive the same quantity up to constants in the power exponent of each term for the second and third derivative.
	Then \Cref{asmp:derivative} imposes that expectations of these quantities with respect to the data generating distribution $x \sim p$ are finite.
	For example, the expectations for the first derivatives are
	\begin{align*}
		\E_{X \sim p}\left[ \sup_{\theta \in U} \| \nabla_\theta r_{j-}(X+1, \theta) \| \right] & = \frac{ 1 }{ M_1^2 } \E_{X \sim p}\left[ \left| (x+1)^{M_2} \right| \right] + \frac{1}{M_1} \E_{X \sim p}\left[ \left| (x+1)^{M_2} \log (x+1) \right| \right] , \\
		\E_{X \sim p}\left[ \sup_{\theta \in U} \| \nabla_\theta ( r_{j-}(x, \theta) )^2 \| \right] & = \frac{ 1 }{ M_1^3 } \E_{X \sim p}\left[ \left| x^{2 M_2} \right| \right] + \frac{1}{M_1^2} \E_{X \sim p}\left[ \left| x^{2 M_2} \log x \right| \right] ,
	\end{align*}
	where the boundedness is translated into the moment condition of $p$ as above.

\section{Details of Experimental Assessment} \label{app:add_exp}

This appendix contains full details for the experiments that were reported in the main text.

\subsection{Conway--Maxwell--Poisson Model}
\label{subsec: CMP extra}

\subsubsection{Settings for KSD-Bayes}
\label{subsec: CMP extra KSD Bayes}

KSD-Bayes is a generalised posterior constructed by taking a kernel Stein discrepancy as a loss function; see \citep{Matsubara2021}.
The approach requires us to specify a kernel function $k: \X \times \X \to \R$, based on which the kernel Stein discrepancy is constructed.
In these experiments, we adopted a kernel recommended by \cite{Yang2018} for the kernel Stein discrepancy in discrete domains $\X$ given by
\begin{align*}
	k(\bm{x}, \bm{x}') = \exp\left( - \frac{1}{d} \sum_{i=1}^{d} \mathbbm{1}(x_i = x_i') \right)
\end{align*}
where $\mathbbm{1}$ is an indicator function, taking values in $\{0,1\}$.
The effect of kernel choice is difficult to predict in the discrete context; for example, \cite{Yang2018} found that the closely related kernel $k(\bm{x}, \bm{x}') = \sum_{i=1}^{d} \mathbbm{1}(x_i = x_i')$, can perform poorly in moderate-to-high dimensions $d$ when employed in a Stein discrepancy.
General principles for kernel choice in the discrete setting have not yet been established.
Thus, one of the advantages of DFD-Bayes is absence of any user-specified parameters of the method.

\subsubsection{Markov Chain Monte Carlo}
\label{subsec: CMP extra MCMC}

A Metropolis--Hasting algorithm was employed to sample from the standard Bayesian posterior, as well as KSD-Bayes and DFD-Bayes.
For computational convenience, the parametrisation $\tilde{\theta}_1 = \log(\theta_1)$ and $\tilde{\theta}_2 = \log(\theta_2)$ was applied so that parameters are defined on an unbounded domain $\tilde{\theta} = (\tilde{\theta}_1, \tilde{\theta}_2) \in \R^2$.
An isotropic Gaussian random walk proposal with covariance $\sigma^2 I$ was employed, with $\sigma = 0.1$ used for all experiments.
The convergence of the Markov chain was diagnosed using univariate Gelman--Rubin statistics for each $\theta_1$ and $\theta_2$ computed from 10 independent chains.
In total, 500 samples were obtained from each chain by thinning 5,000 samples, all after an initial burn-in of length 5,000.
In all cases, the univariate Gelman--Rubin statistics were below $1.02$, respectively for $\theta_1$ and $\theta_2$.

\subsubsection{Sales Dataset of \citet{Shmueli2005}}
\label{subsec: CMP extra dataset}

This dataset consists of quarterly sales figures for a particular item of clothing, taken across the different stores of a large national retailer.
The original dataset is publicly available at \url{https://www.stat.cmu.edu/COM-Poisson/Sales-data.html}; see \citet{Shmueli2005}.
Quarterly sales at each store can be small and result in a large proportion of $0$ entries in the dataset, so that the Conway--Maxwell--Poisson model has a clear advantage against the standard Poisson model.

To obtain a maximum \textit{a posteriori} estimate for the parameters of the Conway--Maxwell--Poisson model for this sales dataset, \citet{Shmueli2005} considered a prior $\pi$ defined by
\begin{align}
	\pi(\theta) \propto \theta_1^{a - 1} \exp( - b \theta_2) \left( \sum_{j=1}^{\infty} \theta_1^j / (j!)^{\theta_2} \right)^{-c} \kappa(a, b, c) \label{eq:cmp_prior}
\end{align}
where $(a,b,c)$ is the hyper-parameter and $\kappa(a, b, c)$ is the normalising constant of $\pi$.
The motivation to use this prior is conjugacy, since the resulting posterior takes the same form as \eqref{eq:cmp_prior}.
However, the prior itself contains the intractable terms $( \sum_{j=1}^{\infty} \theta_1^j / (j!)^{\theta_2} )^{-c}$ and $\kappa(a, b, c)$.
To avoid this additional intractability, which is not a focus of the present work, we considered a simpler chi-squared prior distribution in the main text.

\subsection{Ising Model}
\label{subsec: Ising extra}

\subsubsection{Simulating Data from the Ising Model}
\label{subsubsec: Ising data}

Samples from the Ising model were obtained using the same Metropolis--Hasting algorithm used in \cite{Yang2018}.
First, all coordinates $x_i$ of $\bm{x}$ were randomly initialised to either $-1$ or $1$ with equiprobability $1/2$.
Then, at each iteration, we randomly select one coordinate $x_i$ of $\bm{x}$ and flip the value of $x_i$ either from $-1$ to $1$ or from $1$ to $-1$, where the flipped value $\tilde{x}_i$ is accepted with probability $\min( 1, \exp( -2 \tilde{x}_i \sum_{j \in \mathcal{N}_i} x_j / \theta ) )$ and otherwise rejected.
For the experiments in this paper we ran $n = 1,000$ chains in parallel, in each case taking the final state at iteration $100,000$.
This algorithm was used due to its implementational simplicity, rather than its efficiency, and we note that more sophisticated Markov chain Monte Carlo algorithms are available \citep[e.g.][]{elcci2018lifted}.

\subsubsection{Settings for KSD-Bayes}
\label{subsec: Ising KSD}

The same choice of kernel as \Cref{subsec: CMP extra KSD Bayes} is used.

\subsubsection{Markov Chain Monte Carlo} \label{subsec: Ising extra MCMC}

The same Metropolis--Hasting algorithm as \Cref{subsec: CMP extra MCMC} was used, in this case in dimension $p = 1$ with proposal standard deviation $\sigma = 0.1$.
The convergence of the Markov chain was again diagnosed using univariate Gelman--Rubin statistics computed from 10 independent chains.
In total, 100 samples were obtained after thinning from 2000 samples, with an initial burn-in of length 2000.
In all cases, the univariate Gelman--Rubin statistics were below $1.002$.

\subsection{Multivariate Count Data}
\label{subsec: multivar extra}

\subsubsection{Description of the Dataset}
\label{subsubsec: multivar dataset}

The original data were gathered by the Cancer Genome Atlas Program, run by the National Cancer Institute in the United States, who have built large-scale genomic profiles of cancer patients with the aim to discover the genetic substructures of cancer \citep{Wan2015}.
It contains molecular profiles of biological samples of more than 30 cancer types e.g.~measured via RNA sequencing technology.
The raw data were pre-processed using the TCGA2STAT software developed by \cite{Wan2015}.
\cite{inouye2017review} studied a subset of these data relevant to breast cancer, consisting of a total count of each gene profile found in biological samples.
They applied a ``log-count" transform, a common preprocessing technique for RNA sequencing data, for every datum, that is a floor function of a log transformed value of the datum.
Gene profiles were then sorted by variance of the counts in descending order, with the top 10 gene profiles constituting the final dataset.
The preprocessed data studied in \cite{inouye2017review} can be found in \url{https://github.com/davidinouye/sqr-graphical-models}.

\subsubsection{Markov Chain Monte Carlo}
\label{subsec: multivar extra MCMC}

The Metropolis-Hasting Markov Chain Monte Carlo was applied for this experiment.
The detail for the Conway--Maxwell--Poisson graphical model is described first as the Poisson graphical model is the special case.
For computational convenience, we work with the square of the interaction and dispersion parameters, i.e. $\tilde{\theta}_{i,j} := \theta_{i,j}^2$ and $\tilde{\theta}_{0,i} = \theta_{0,i}$, which modify the model as
\begin{align*}
	p_\theta(\bm{x}) \propto \exp\left( \sum_{i=1}^d \theta_i x_i - \sum_{i=1}^d \sum_{j \in \mathcal{M}_i} \tilde{\theta}_{i,j}^2 x_i x_j - \sum_{i=1}^d \tilde{\theta}_{0,i}^2 \log(x_i!) \right)
\end{align*}
The domain of each original parameter $\theta_{i,j}$ and $\theta_{0,j}$ is $[0,\infty)$.
With this modification, $\tilde{\theta}_{i,j}$ and $\tilde{\theta}_{0,i}$ can be extended to $\R$, making the model $p_\theta(\bm{x})$ differentiable with respect to $\theta \in \R^v$.
The derivatives of the corresponding $\operatorname{DFD}$-Bayes posterior is then available to implement an efficient gradient-based Markov chain Monte Carlo method.
We place a standard normal distribution as a prior on each $\theta_i$, a normal distribution with mean $0$ and scale $( d (d - 1) / 2 )^{-1}$ as a prior on each $\tilde{\theta}_{i,j}$, and a standard normal distribution as a prior on each $\tilde{\theta}_{0,i}$, that corresponds to the original priors of each $\theta_i$, $\theta_{i,j}$, and $\theta_{0,j}$.
The small scale of the half normal distribution prior on $\tilde{\theta}_{i,j}$ was chosen to suppress rapid increase of the quadratic term $x_i x_j$ as opposed to the linear term $x_i$ in the first summation.
After the Markov chain finished, the absolute value was taken for the sampled values of $\tilde{\theta}_{i,j}$ and $\tilde{\theta}_{0,i}$ to convert them as the original parameters $\theta_{i,j}$ and $\theta_{0,j}$.
The same setting is applied for the Poisson graphical model by fixing the dispersion parameter $\tilde{\theta}_{0,i} = \theta_{0,i} = 1$.

A No-U-Turn Sampler was used to approximate the DFD-Bayes posterior of both the models.
In total, $100$ points were obtained thinning from $5,000$ samples, with an initial burn-in of length $5,000$.
The posterior predictive of each model $p_\theta(\bm{x})$ was computed by generating $500,000$ samples from $p_\theta(\bm{x})$ at every $\theta$ sampled from the DFD-Bayes posterior.
Each $500,000$ predictive samples were thinned to $878$ points to make it comparable with the original data of $n = 878$.
The number of bootstrap minimisers $B$ used to calibrate $\beta$ for this experiment was $B = 100$.

\subsubsection{Gradient of the Discrete Fisher Divergence}
\label{subsec: DFD gradient}

For a model $p_\theta(\bm{x})$, denote the normalisation constant by $C(\theta)$ and the non-normalised part by $q_\theta(\bm{x})$, so that $p_\theta(\bm{x}) = q_\theta(\bm{x}) / C(\theta)$.
The discrete Fisher divergence is differentiable whenever the non-normalised part $q_\theta(\bm{x})$ is differentiable with respect to $\theta$ at any $\bm{x} \in \X$.
Indeed, the discrete Fisher divergence between a model $p_\theta$ and data $p_n$ is given by
\begin{align*}
\text{DFD}(p_\theta \| p_n) & \overset{\theta}{=} \frac{1}{n} \sum_{i=1}^{n} \sum_{j=1}^{d} \bigg( \frac{p_\theta(\bm{x}_i^{j-})}{p_\theta(\bm{x}_i)} \bigg)^2 - 2 \bigg( \frac{p_\theta(\bm{x}_i)}{p_\theta(\bm{x}_i^{j+})} \bigg) \\
& \overset{\theta}{=} \frac{1}{n} \sum_{i=1}^{n} \sum_{j=1}^{d} \bigg( \frac{q_\theta(\bm{x}_i^{j-})}{q_\theta(\bm{x}_i)} \bigg)^2 - 2 \bigg( \frac{q_\theta(\bm{x}_i)}{q_\theta(\bm{x}_i^{j+})} \bigg) 
\end{align*}
where the $\theta$-independent term is ignored and the equality holds because the normalising constant $C(\theta)$ is cancelled out.
By routine calculation, the gradient of the discrete Fisher divergence is further given by
\begin{multline*}
	\nabla_\theta \text{DFD}(p_\theta \| p_n) = \frac{1}{n} \sum_{i=1}^{n} \sum_{j=1}^{d} 2 \bigg( \frac{q_\theta(\bm{x}_i^{j-})}{q_\theta(\bm{x}_i)} \bigg) \bigg( \frac{\nabla_\theta q_\theta(\bm{x}_i^{j-}) q_\theta(\bm{x}_i) - q_\theta(\bm{x}_i^{j-}) \nabla_\theta q_\theta(\bm{x}_i) }{q_\theta(\bm{x}_i)^2} \bigg) \\
	- 2 \bigg( \frac{\nabla_\theta q_\theta(\bm{x}_i) q_\theta(\bm{x}_i^{j+}) - q_\theta(\bm{x}_i) \nabla_\theta q_\theta(\bm{x}_i^{j+})}{q_\theta(\bm{x}_i^{j+})^2} \bigg) .
\end{multline*}
Therefore, the gradient of the discrete Fisher divergence is well-defined as long as $q_\theta(\bm{x}_i) \ne 0$ and $q_\theta(\bm{x}_i^{j+}) \ne 0$, which in any case are prerequisites for computation of the discrete Fisher divergence.

\end{document}